\documentclass{article}
\usepackage{graphicx} % Required for inserting images
\usepackage{wright}
\usepackage{physics}

\usepackage{tikz}
\usepackage{adjustbox}
\usepackage{caption}
\usetikzlibrary{matrix, positioning, calc}

\newcommand{\Span}{\mathrm{span}}
\renewcommand{\epsilon}{\varepsilon}
\newcommand{\fid}{\mathrm{F}}
\newcommand{\dtr}{\mathrm{D}_\mathrm{tr}}
\newcommand{\dbur}{\mathrm{D}_\mathrm{B}}
\newcommand{\aff}{\mathrm{A}}
\usepackage{physics}
\newcommand{\dU}{\mathrm{d}U}
\newcommand{\alg}{\mathcal{A}}

\renewcommand{\epsilon}{\varepsilon}

\renewcommand{\epsilon}{\varepsilon}

\title{Optimal lower bounds for quantum state tomography}

\author{Thilo Scharnhorst\thanks{UC Berkeley. \texttt{\{thilo,jspilecki,jswright\}@berkeley.edu}}  \and Jack Spilecki\footnotemark[1] \and John Wright\footnotemark[1]}
\date{}

\begin{document}

\maketitle

\begin{abstract}
    We show that $n = \Omega(rd/\epsilon^2)$ copies are necessary to learn a rank $r$ mixed state $\rho \in \C^{d \times d}$ up to error~$\epsilon$ in trace distance.
    This matches the upper bound of $n = O(rd/\epsilon^2)$ from~\cite{OW16} and therefore settles the sample complexity of mixed state tomography.
    We prove this lower bound by studying a special case of full state tomography that we refer to as \emph{projector tomography},
    in which $\rho$ is promised to be of the form $\rho = P/r$, where $P \in \C^{d \times d}$ is a rank $r$ projector.
    A key technical ingredient in our proof, which may be of independent interest, is a reduction which converts any algorithm for projector tomography which learns to error $\epsilon$ in trace distance to an algorithm which learns to error $O(\epsilon)$ in the more stringent Bures distance.
\end{abstract}

\newpage
\hypersetup{linktocpage}
\tableofcontents
\thispagestyle{empty}

\newpage

\section{Introduction}
\renewcommand{\D}{\mathrm{D}}

In this work, we study the fundamental learning theoretic problem of \emph{quantum tomography}.
Given $n$ copies of a mixed state $\rho \in \C^{d \times d}$,
the goal is to output an estimate $\widehat{\rho}$ such that $\mathrm{D}(\rho, \widehat{\rho}) \leq \epsilon$, where $\mathrm{D}(\cdot, \cdot)$ is a given distance metric.
The two most well-studied cases are when $\mathrm{D} = \mathrm{D}_{\mathrm{tr}}$ is the trace distance and when $\mathrm{D} = \mathrm{D}_{\mathrm{B}}$ is the more challenging Bures distance.
The Bures distance is defined as $\mathrm{D}_{\mathrm{B}}(\rho, \widehat{\rho}) = \sqrt{2 (1 - \mathrm{F}(\rho, \widehat{\rho}))}$, where $\mathrm{F}(\rho, \widehat{\rho})$ is the fidelity of $\rho$ and $\widehat{\rho}$, and is related to the trace distance via the inequalities
\begin{equation}\label{eq:intro-ineqs}
\frac{1}{2} \mathrm{D}_{\mathrm{B}}(\rho, \widehat{\rho})^2 \leq \mathrm{D}_{\mathrm{tr}}(\rho, \widehat{\rho}) \leq \mathrm{D}_{\mathrm{B}}(\rho, \widehat{\rho}).
\end{equation}
Perhaps the most commonly studied special case is when $\rho$ is promised to be rank $r$, for some integer $1 \leq r \leq d$.
Many applications in both theory and practice involve states which are either low rank or approximately low rank,
with the $r = 1$ pure state case being especially important.

Surprisingly, in spite of the large amount of energy  devoted to studying quantum tomography, the optimal sample complexity of this task still remains unknown.
For trace distance,
the best known upper bound is due to O'Donnell and Wright~\cite{OW16}, who showed that $n = O(dr/\epsilon^2)$ copies are sufficient to learn $\rho$ to accuracy $\epsilon$ with high probability (by which we mean a large constant probability, say 0.99);
this was improved by Pelecanos, Spilecki, and Wright in~\cite{PSW25}, who showed that the same number of copies suffice to learn $\rho$ to Bures distance error $\epsilon$.
As for lower bounds, there is a patchwork of results which cover various parameter regimes;
of these, let us describe the lower bounds which apply to trace distance learning first.
When $r = d$, Haah et al.~\cite{HHJ+16} gave a tight lower bound, showing that $n = \Omega(d^2/\epsilon^2)$ copies are necessary.
In addition, when $\epsilon$ is constant, Wright~\cite[Section 5.5]{Wri16} gave a tight lower bound, showing that $n = \Omega(rd)$ copies are necessary.
For more general values of $r$ and $\epsilon$, Haah et al.~\cite{HHJ+16} also showed a lower bound of
\begin{equation*}
    n = \Omega\Big(\frac{rd}{\epsilon^2} \cdot \frac{1}{\log(d/(r\epsilon))}\Big).
\end{equation*}
This is off from the known upper bounds by the logarithmic factor of $\log(d/(r\epsilon))$ which gets larger as the rank $r$ decreases.
Finally, a lower bound of $n = \Omega(r/\epsilon^2)$ copies follows from the classical special case when $\rho$ is promised to be diagonal in the standard basis~\cite{Can20}.

Since (trace distance) $\leq $ (Bures distance), all of these lower bounds also hold for the harder problem of learning $\rho$ to Bures distance $\epsilon$.
Recently, Yuen~\cite{Yue23} gave an even stronger Bures distance lower bound, showing that $n = \Omega(rd/\epsilon^2)$ copies are necessary to learn $\rho$ to Bures distance error $\epsilon$.
Combined with the $n = O(rd/\epsilon^2)$ copy upper bound of Pelecanos, Spilecki, and Wright from above, this settles the copy complexity of Bures distance tomography.
Yuen proves his lower bound by a clever reduction from the rank $r$ case to the rank $r = 1$ pure state case and relies on a prior result that shows that $n = \Omega(d/\epsilon^2)$ copies are required to learn to Bures distance error $\epsilon$ in this case.

However, there is a slightly subtlety with this approach,
in that it actually shows that $n = \Omega(rd/\epsilon^2)$ copies are necessary to learn $\rho$ to Bures distance error $\epsilon$ in \emph{expectation},
rather than with high probability.
The reason is that the lower bound cited for pure state tomography, that of Bruss and Macchiavello~\cite{BM98},
applies to learning in expectation rather than with constant probability.
In particular, they show the following: if $n$ copies of a Haar random pure state $\ket{\bu} \in \C^d$ are given to a tomography algorithm $\calA$, and $\calA$ produces the estimator $\ket*{\widehat{\bu}}$, then the average fidelity of the output is upper-bounded by
\begin{equation}\label{eq:expected-fidelity}
\E_{\ket{\bu}, \ket{\widehat{\bu}}} |\braket*{\bu}{\widehat{\bu}}|^2 \leq \frac{n+1}{n+d}.
\end{equation}
When $n = o(d/\epsilon^2)$, this implies that the average fidelity is at most $1 - \omega(\epsilon^2)$.
If, for example, we knew that the fidelity was also $1 - \omega(\epsilon^2)$ not just on average but with high probability, then the Bures distance between $\ketbra{\bu}$ and $\ketbra*{\widehat{\bu}}$ would be $\omega(\epsilon)$ with high probability, ruling out Bures distance learning with $o(d/\epsilon^2)$ copies.
However, it is also consistent with this bound that the fidelity could be 0 with probability .0001 and~$1$ with probability .9999 (producing an average fidelity of $.9999 \ll 1 - \omega(\epsilon^2)$), yielding an algorithm which learns to Bures distance error 0 with probability .9999.
Hence, this bound on the expected fidelity is not sufficient to show that $n = \Omega(d/\epsilon^2)$ copies are necessary for pure state tomography with high probability

If we want to use the reduction of Yuen to produce a Bures distance lower bound of $n = \Omega(rd/\epsilon^2)$ copies with high probability,
we therefore need lower bounds on pure state learning with high probability.
Surprisingly, as far as we can tell, no lower bound of the form $n = \Omega(d/\epsilon^2)$ for pure state tomography is stated or proved anywhere in the literature.
Existing lower bounds seem to either have suboptimal sample complexity,
such as the $n = \Omega(d/(\epsilon^2 \cdot \log(1/\epsilon)))$ bound of~\cite{HHJ+16},
or they seem to only apply in the asymptotic regime of large $n$~\cite{GM00}.

As one result in this paper, we are able to show that $n = \Omega(d/\epsilon^2)$ copies are necessary for pure state tomography;
combined with Yuen's reduction, this shows an $n = \Omega(rd/\epsilon^2)$ copy fidelity tomography lower bound.
Our main result, however, strengthens this Bures distance lower bound and shows that it holds for trace distance tomography as well.

\begin{theorem}[Optimal tomography lower bound]\label{thm:main}
    Given a rank-$r$ mixed state $\rho \in \C^{d \times d}$, $n = \Omega(rd/\epsilon^2)$ copies are required to estimate it to trace distance error $\epsilon$.
\end{theorem}
\noindent
Paired with the upper bound of O'Donnell and Wright~\cite{OW16}, this resolves the sample complexity for trace distance tomography.

\subsection{Projector tomography} \label{sec:intro_projector_tomography}
To prove \Cref{thm:main}, we focus on a special case of the rank-$r$ tomography problem that we refer to as \emph{projector tomography}.
This is the case in which the unknown rank-$r$ state $\rho$ is promised to be of the form $\rho = P/r$, where $P$ is the projector onto some unknown rank-$r$ subspace;
we refer to states of this form as \emph{rank-$r$ projector states}.
When $r = 1$, the set of rank-$r$ projector states coincides with the set of all pure states.
A recurring theme throughout quantum information is that pure states have a special structure that makes them especially convenient to analyze in many different settings.
For example, the optimal tomography algorithm for pure states due to Hayashi was discovered in~\cite{Hay98}, long before the optimal tomography bounds for general mixed states were shown by O'Donnell and Wright in~\cite{OW16}.
Similarly, we will give an especially simple proof of our optimal tomography lower bound in the pure state case. 
For larger values of $r$, however, rank-$r$ projector states form a highly structured subset of the set of all rank-$r$ mixed states.
As we will argue below, we believe that
projector states possess at least some of the same features that make pure states so convenient to prove upper and lower bounds for, and we will use these features to prove the following theorem.

\begin{theorem}[Tight upper and lower bounds for projector tomography]\label{thm:projector}
    There is a tomography algorithm which learns an unknown rank-$r$ projector state to Bures distance error $\epsilon$ using $n = O(rd/\epsilon^2)$ copies.
    In addition, $n = \Omega(rd/\epsilon^2)$ copies are required to learn an unknown rank-$r$ projector state to trace distance error~$\epsilon$.
\end{theorem}

As rank-$r$ projector states are a subset of all rank-$r$ mixed states, \Cref{thm:main} follows from \Cref{thm:projector} as a corollary.
Our upper bound already follows from the more general rank-$r$ fidelity tomography result of Pelecanos, Spilecki, and Wright~\cite{PSW25}.
However, in our case, we are able to give a simplified algorithm and analysis that we think will be of independent interest.
We note that rank-$r$ projector states have appeared in several prior works in the learning theory literature.
Indeed, the lower bound proofs of Haah et al.~\cite{HHJ+16} and Wright~\cite[Section 5.5]{Wri16} all use rank-$r$ projector states, although the techniques they use to analyze them are very different from the techniques we use.
More recently, the work of Pelecanos, Tan, Tang, and Wright~\cite{PTTW25} studied the problem of estimating the spectrum of an unknown mixed state $\rho \in\C^{d \times d}$.
A key step in their algorithm is known as ``bucketing'',
and in their Section 9.2 they identify rank-$r$ projector state tomography as a bottleneck for improving the sample complexity of their bucketing step.
In particular, they conjecture that $n = \Omega(rd/\epsilon^2)$ copies are required to perform tomography of rank-$r$ projector states in Bures distance,
and they show that if this were true, then bucketing up to a ``threshold value'' of $0 \leq B \leq 1$ would require $\Omega(d B^{-1}/\epsilon)$ copies.
Our \Cref{thm:projector} proves their conjectured lower bound and improves it to hold for trace distance as well.

Below we describe how we prove the lower and upper bounds in \Cref{thm:projector}.

\subsubsection{Proof outline: the lower bound}

Our proof of the lower bound is in two steps.
\begin{enumerate}
    \item Prove that $n = \Omega(rd/\epsilon^2)$ copies are necessary for rank-$r$ projector tomography in Bures distance.
    \item ``Bootstrap'' this result to hold for trace distance as well.
\end{enumerate}
\noindent
A special case of our first step is that $n = \Omega(d/\epsilon^2)$ copies are necessary for learning pure states, which, as mentioned before, patches the hole in Yuen's $n = \Omega(rd/\epsilon^2)$ lower bound for Bures distance tomography~\cite{Yue23}.
However, our first step also gives a proof of this Bures distance lower bound for general $r$ which is different than Yuen's proof.
The reason we need to reprove this lower bound is that our bootstrapping step only works for projector tomography, and the class of states produced by Yuen's lower bound are not projector states.
We now describe these steps in more detail.

\paragraph{Step 1: a fidelity lower bound.}
Let us begin by describing the first step in the $r = 1$ pure state case.
Actually, it is not especially hard to combine existing results in the literature to prove an $n = \Omega(d/\epsilon^2)$ lower bound in this case.
The simplest proof we found uses the fact that the optimal pure state tomography algorithm has already been identified in the literature~\cite{Hay98}.
One can then show that this particular algorithm's fidelity $|\braket{\bu}{\widehat{\bu}}|^2$ concentrates well about its mean by computing its variance;
this allows one to convert the bound on its expected fidelity from \Cref{eq:expected-fidelity} into a bound on its fidelity with high probability.
However, we were unable to generalize this argument to the case of rank-$r$ projector tomography for $r \geq 2$,
as although we have an algorithm that we believe is optimal for this case, we were unable to prove that it is indeed optimal.
Instead, we will describe an alternative proof which we were able to generalize to the rank-$r$ case.

As we saw in \Cref{eq:expected-fidelity}, it is possible to bound the expected fidelity of any pure state tomography algorithm.
In fact, it is well-known that one can derive tight bounds on the $k$-th moment of the fidelity, for any integer $k \geq 1$.
A standard reference for these bounds is~\cite[Section 2.1]{Har13}, which shows that
\begin{equation*} %\label{eq:moment-bound}
    \E_{\ket{\bu}, \ket{\widehat{\bu}}} |\braket*{\bu}{\widehat{\bu}}|^{2k} \leq \frac{\binom{d+n-1}{n}}{\binom{d+n+k-1}{n+k}}.
\end{equation*}
Expanding out the binomial, we have that
\begin{equation}\label{eq:moment-bound}
\E_{\ket{\bu}, \ket{\widehat{\bu}}} |\braket*{\bu}{\widehat{\bu}}|^{2k}
\leq
    \frac{\binom{d+n-1}{n}}{\binom{d+n+k-1}{n+k}}
    =\frac{(n+1)\dots(n+k)}{(n+d)\dots(n+d+k-1)} \leq \Big(\frac{n+k}{d+n+k-1}\Big)^{k}.
\end{equation}
Roughly, this states that the $k$-th moment of the fidelity decays exponentially as $k$ grows larger (so long as $k$ does not grow \emph{too} large), and the rate of decay is greater when the number of samples $n$ is small.
This places a bound on the performance of any algorithm which uses a small number of samples, and we show that this moment bound implies that any algorithm for learning in Bures distance requires $n = \Omega(d/\epsilon^2)$ copies.
As an example of why higher moments might be useful to do this, recall our counterexample from earlier of an algorithm $\calA$ which uses $n = o(d/\epsilon^2)$ copies and outputs an estimate whose fidelity is 0 with probability .0001 and 1 with probability .9999. Then the $k$-th moment of its fidelity is equal to .9999, for any value of $k$. This may not contradict the bound on the expected fidelity (the $k = 1$ case of \Cref{eq:moment-bound}), but for larger $k$ the right-hand side will eventually decay to a number smaller than .9999, which is a contradiction. 

Now let us try to extend this argument to higher ranks $r$.
Let $\calA$ be a rank-$r$ projector tomography algorithm.
We will consider the following experiment: sample a Haar random rank-$r$ subspace $\bP$ in $\C^d$, and provide $\calA$ with $n$ copies of the projector state $\brho = \bP/r$.
Let $\widehat{\brho} = \bQ/r$ be its output.
Our $r = 1$ proof suggests that we would like to understand the expected fidelity $\fid(\brho, \widehat{\brho})$ and its moments.
However, the fidelity is not a particularly ``nice'' function of its inputs, and it is not clear at all how we would compute the expectation of the fidelity, much less its higher moments.
Instead, we will use the \emph{affinity} $\aff(\brho, \widehat{\brho}) = \mathrm{tr}(\sqrt{\brho} \sqrt{\widehat{\brho}})$, which satisfies $\aff(\brho, \widehat{\brho}) \leq \fid(\brho, \widehat{\brho}) \leq \sqrt{\aff(\brho, \widehat{\brho})}$ and is therefore closely related to the fidelity in the regime where $\brho$ and $\widehat{\brho}$ are similar to each other.
For general states $\brho$ and $\widehat{\brho}$, the affinity can still be difficult to work with due to the square roots, but when $\brho$ and $\widehat{\brho}$ are rank-$r$ projector states, we have $\sqrt{\brho} = \sqrt{r} \cdot \brho$ and $\sqrt{\widehat{\brho}} = \sqrt{r} \cdot \widehat{\brho}$, in which case the affinity simplifies nicely to $\aff(\brho, \widehat{\brho}) = r\cdot \tr(\brho \cdot \widehat{\brho})$.
Higher moments of the affinity behave nicely too: in particular, $\aff(\brho, \widehat{\brho})^k = r^k \cdot \mathrm{\tr}(\brho^{\otimes k} \cdot \widehat{\brho}^{\otimes k})$.
Since the affinity and its higher powers are simple expressions involving tensor powers of $\brho$ and $\widehat{\brho}$, we can compute and bound the moments using tools from representation theory such as Schur-Weyl duality.
Eventually, we are able to derive a bound on the $k$-th moment that resembles \Cref{eq:moment-bound},
and from there we can derive our $n = \Omega(rd/\epsilon^2)$ lower bound.

\paragraph{Step 2: bootstrapping.}
For our second step, let $\calA$ be an algorithm which solves rank-$r$ projector state tomography with trace distance error $\epsilon$ using $n$ copies. 
In other words, given a rank-$r$ projector state $\rho = P/r$, $\calA$ outputs a random estimator $\widehat{\brho} = \bQ/r$ such that $\mathrm{D}_{\mathrm{tr}}(\rho, \widehat{\brho}) \leq \epsilon$.
In our bootstrapping step, we would like to show that $\calA$ can be converted into an algorithm $\calA'$ which performs tomography with Bures distance error $O(\epsilon)$ using roughly the same number of copies $O(n)$.
If we could do this, then because we proved that Bures distance tomography to error $O(\epsilon)$ requires $\Omega(rd/\epsilon^2)$ copies in our first step, this would show that $n = \Omega(rd/\epsilon^2)$ copies are also needed for trace distance tomography to error $O(\epsilon)$, which would complete the proof.
Now, from \Cref{eq:intro-ineqs}, we know that  $\mathrm{D}_{\mathrm{tr}}(\rho, \widehat{\brho}) \leq \mathrm{D}_{\mathrm{B}}(\rho, \widehat{\brho}) \leq \sqrt{2\mathrm{D}_{\mathrm{tr}}(\rho, \widehat{\brho})}$, and so if we imagine that $\mathrm{D}_{\mathrm{tr}}(\rho, \widehat{\brho}) \approx \epsilon$, then we have, roughly, $\epsilon \leq \mathrm{D}_{\mathrm{B}}(\rho, \widehat{\brho}) \leq \sqrt{2\epsilon}$.
If $\mathrm{D}_{\mathrm{B}}(\rho, \widehat{\brho})$ is closer to the lower bound, then we are happy and $\widehat{\brho}$ itself is already a good Bures distance estimator for $\rho$.
But $\mathrm{D}_{\mathrm{B}}(\rho, \widehat{\brho})$ could very well be closer to the upper bound, in which case it is off from our desired trace distance error $\epsilon$ by a square root factor.
Thus, simply running $\calA$ once and directly returning its output is not good enough to bootstrap it into a trace distance learning algorithm.
However, we show that to construct the bootstrapped Bures distance tomography algorithm $\calA'$, it actually suffices to call $\calA$ as a subroutine \emph{twice}, as well as use $O(r^2/\epsilon^2)$ additional copies of $\rho$.
This gives a Bures distance tomography algorithm with $2n + O(r^2/\epsilon^2)$ copies in total, and as we know it must use $\Omega(rd/\epsilon^2)$ copies for this task, this proves the bound $n = \Omega(rd/\epsilon^2)$.

To motivate our bootstrapping algorithm, let us try to understand the following question: if $\mathrm{D}_{\mathrm{tr}}(\rho, \widehat{\brho}) = \epsilon$,
when is $\mathrm{D}_{\mathrm{B}}(\rho, \widehat{\brho})$ closer to $\epsilon$, and when is it closer to $\sqrt{2\epsilon}$?
This entails understanding the relationship between the true projector $P$ and the estimated projector $\bQ$,
and there is a well-known technique for understanding the relationship between two projectors known as \emph{Jordan's lemma}.
In our setting, Jordan's lemma states, roughly, that $P$ and $\bQ$ can be simultaneously block diagonalized into $2 \times 2$ blocks known as \emph{Jordan blocks}, and within each Jordan block $P$ and $\bQ$ both act as rank-1 projectors.
This means that we can diagonalize $P$ and $\bQ$ according to these blocks as
\begin{equation*}
    P = \sum_{i=1}^r \ketbra{\bu_i}
    \quad\text{and}\quad
    \bQ = \sum_{i=1}^r \ketbra{\bv_i},
\end{equation*}
where $\ketbra{\bu_i}$ and $\ketbra{\bv_i}$ are the restrictions of $P$ and $\bQ$ to the $i$-th subspace, respectively.
Across different blocks, $\ket{\bu_i}$ and $\ket{\bv_j}$ are orthogonal, and within the $i$-th block, let us write $\bomega_i = |\braket{\bu_i}{\bv_i}|$ for their overlap.
As it turns out, there are simple formulas for the trace and Bures distance in terms of these overlaps.
For example,
\begin{align}
    \dtr(\rho, \widehat{\brho})
    = \frac{1}{2} \norm{ \rho - \widehat{\brho}}_1
    &= \frac{1}{2r} \norm{P - \bQ}_1\nonumber\\
    &= \frac{1}{2r} \sum_{i=1}^r \norm{ \ketbra{\bu_i} - \ketbra{\bv_i}}_1 = \frac{1}{r} \sum_{i=1}^r \sqrt{1 - \left|\braket{\bu_i}{\bv_i} \right|^2}
    = \frac{1}{r} \sum_{i=1}^r \sqrt{1 - \bomega_i^2}.\label{eq:exact-td-formula}
\end{align}
Similarly, we can write the fidelity as
\begin{equation*}
    \fid(\rho, \widehat{\brho}) = \frac{1}{r} \sum_{i=1}^r \bomega_i,
\end{equation*}
which gives a formula for the Bures distance via $\dbur = \sqrt{2(1 - \fid)}$.

Recalling that we are assuming $\dtr(\rho, \widehat{\brho}) = \epsilon$, let us consider two different extreme cases for how the overlaps $\bomega_1, \ldots, \bomega_r$ might behave.
\begin{itemize}
\item
In one extreme, let us suppose that $\bomega_1 = \cdots = \bomega_r$. Then from \Cref{eq:exact-td-formula}, we must have $\bomega_i^2 = 1 - \epsilon^2$ for all $i$, so that $\bomega_i \approx 1- \frac{1}{2}\epsilon^2$ for all $i$.
In this case, we have $\fid(\rho, \widehat{\brho}) \approx 1 - \frac{1}{2} \epsilon^2$, which implies that $\dbur(\rho, \widehat{\brho}) \approx \epsilon$.
Intuitively, this is the case in which $\bQ$ is basically equal to $P$, except with a slight, uniform error across all of the Jordan blocks. And in this case, we have seen that $\widehat{\brho}$ is itself a good Bures distance estimate for $\rho$.
\item
In the other extreme, let us suppose that $P$ is exactly equal to $Q$ on the first $r^* < r$ Jordan blocks and orthogonal to $Q$ on the remaining Jordan blocks.
In other words, $\bomega_1 = \cdots = \bomega_{r^*} = 1$ and $\bomega_{r^*+1} = \cdots = \bomega_r = 0$.
If $\dtr(\rho, \widehat{\brho}) = \epsilon$, then \Cref{eq:exact-td-formula} implies that $r^* = (1 - \epsilon) r$.
In this case, we have $\fid(\rho, \widehat{\brho}) = \epsilon$, and so $\dbur(\rho, \widehat{\brho}) = \sqrt{2 \epsilon}$, which is off from our desired bound by a square root factor.
Intuitively, this extreme is the case where the estimator nails a large part of $P$ but completely misses the rest,
and this is the problematic case for getting good Bures distance estimates.
\end{itemize}
Since in the first case $\widehat{\brho}$ is already a good Bures distance estimate for $\rho$,
let us imagine that we are in the second case.
In this case, for the sake of intuition we can imagine that the algorithm $\calA$ begins with the true projector $P$ and forms $\bQ$ by adversarially choosing a subspace $S \subseteq P$ of size $\epsilon r$, discarding it from $P$, and substituting it with another adversarially-chosen subspace $S' \subseteq \overline{P}$ of the same size.
To improve the estimate $\bQ$ of $P$ so that it has Bures distance error $\epsilon$, the bootstrapped algorithm $\calA'$ must somehow ``rediscover'' the discarded subspace $S$ and add it back to $\bQ$, and it is allowed to perform multiple executions of $\calA$ to aid it in its rediscovery.
However, if the algorithm $\calA$ does indeed act adversarially, it may decide to simply discard the same subspace $S$ every time, meaning that $\calA'$ will never be able to rediscover it.

To fix this issue, suppose we sample a Haar random unitary $\bU$ and provide $\calA$ with $n$ copies of $\bU \rho \bU^{\dagger} = (\bU P  \bU^{\dagger})/r$, rather than just giving it $n$ copies of $\rho$. 
If $\bQ/r$ is its output, then $\widehat{\bP}/r = \bU^{\dagger} \bQ \bU$ should be a good estimate for $P$.
But why go through the trouble of rotation $\rho$ prior to giving it to $\calA$?
The answer is that since $\calA$ does not know the original projector $P$ nor the unitary $\bU$ which was applied to rotate it, $\calA$'s ability to adversarially pick the subspace $S$ to discard is hampered.
In particular, it can be shown that the output $\widehat{\bP}$ of this process has the same distribution as
\begin{equation*}
    (\bU_{P}^\dagger \oplus \bU_{\overline{P}}^\dagger) \cdot \widehat{\bP} \cdot (\bU_{P} \oplus \bU_{\overline{P}}),
\end{equation*}
where $\bU_P$ is a Haar random unitary within the $P$ subspace and $\bU_{\overline{P}}$ is an independent Haar random within the $\overline{P}$ subspace.
This means that even if $\widehat{\bP}$ is ``missing'' a subspace $\bS$ from $P$ of dimension $\epsilon r$, then this subspace is not chosen adversarially but instead uniformly at random from all subspaces of $P$ of this dimension.
In particular, if we run this process twice to generate two projectors $\widehat{\bP}_1$ and $\widehat{\bP}_2$, then the subspace of $P$ missing in $\widehat{\bP}_1$ will largely be present in $\widehat{\bP}_2$, and similarly the subspace missing in $\widehat{\bP}_2$ will largely be present in $\widehat{\bP}_1$.
Hence, the span of these two subspaces will likely contain all of $P$, meaning that $\calA'$ can indeed ``rediscover'' the discarded subspaces.

This intuition was for solving our second extreme case above.
Our ultimate bootstrapping algorithm must also work for the first extreme case above, as well as the more general case, which might fall along the spectrum between these two extremes.
Our final bootstrapping algorithm $\calA'$, which achieves this, looks as follows.
\begin{enumerate}
    \item Pick a random unitary $\bU$. Give $n$ copies of $\bU \rho \bU^{\dagger}$ to $\calA$ and let $\bQ/r$ be its output. Write $\widehat{\bP}_1 = \bU^{\dagger} (\bQ/r) \bU$.
    \item Repeat this process a second time to construct $\widehat{\bP}_2$.
    \item Let $\bR$ be the projector onto $\Span\{\widehat{\bP}_1, \widehat{\bP}_2\}$.
    \item Take $O(r^2/\epsilon^2)$ copies of $\rho$ and measure each of them with $\{\bR, \overline{\bR}\}$.
    Discard the post-measurement states corresponding to the outcome $\overline{\bR}$.
    \item The remaining post-measurement states $\rho|_{\bR}$ live inside $\bR$, which is a subspace of dimension at most $2r$.
    Using the Bures distance tomography algorithm of Pelecanos, Spilecki, and Wright \cite{PSW25}, we can compute an estimate $\widehat{\brho}$ of $\rho|_{\bR}$ with Bures distance error $\epsilon$ using only $O(r^2/\epsilon^2)$ copies of $\rho|_{\bR}$.
    \item Output $\widehat{\brho}$ as the estimate for $\rho$.
\end{enumerate}
The key difficulty is showing that $\bR$ does indeed contain almost all of $\rho$.
Technically, our goal is to prove that $\tr(\bR \cdot \rho) \geq 1 - O(\epsilon^2)$.
This will imply two things: first, that measuring our $O(r^2/\epsilon^2)$ copies of $\rho$ will with high probability leave us with $O(r^2/\epsilon^2)$ copies of $\rho|_{\bR}$.
Second, it implies that $\rho|_{\bR}$ is $\epsilon$-close to $\rho$ in Bures distance.
With these two facts established, the correctness of the algorithm follows immediately.

\subsubsection{Proof outline: the upper bound}

Historically, designing and analyzing optimal algorithms for full state tomography has proved to be quite challenging.
Part of the reason for this is that it is not even clear what exactly the right full state tomography algorithm to use is:
between Keyl's algorithm~\cite{Key06} and the two algorithms proposed by Haah et al.~\cite{HHJ+16},
we know of three different tomography algorithms which achieve optimal or near-optimal sample complexities,
and none of these seems to have a strong claim to be \emph{the} canonical full state tomography algorithm.
(Though perhaps the new debiased Keyl's algorithm of~\cite{PSW25} might finally lay claim to that title.)
Beyond that, actually analyzing these algorithms is also difficult, as it tends to involve somewhat complicated representation theory.

One notable exception to this is the case of pure state tomography. In this case, there is a well-known ``canonical'' algorithm due to Hayashi~\cite{Hay98} which simply performs the POVM
\begin{equation*}
    \Big\{\binom{d+n-1}{n} \cdot \ketbra{v}^{\otimes n} \cdot \mathrm{d}v\Big\},
\end{equation*}
and outputs the measurement outcome $\ket{v}$ as its estimator (here, $\mathrm{d}v$ is the Haar measure on pure states). This algorithm and its analysis are so clean that they can be taught in both undergraduate and graduate classes on quantum computing~\cite{Wri15,Wal17,Wri24b}, and they give a good introduction to the power of representation theory in designing quantum algorithms.
One way of viewing this algorithm is as an instantiation of the \emph{Pretty Good Measurement (PGM)}~\cite{Bel75,Hol79,HW94} from the field of quantum hypothesis testing.
In quantum hypothesis testing,
there is a probability distribution $\alpha = (\alpha_1, \ldots, \alpha_m)$ over $m$ mixed quantum states $\rho_1, \ldots, \rho_m$.
One is given the state $\rho_{\bi}$, where $\bi$ is sampled according to $\alpha$, and the goal is to correctly identify the state $\rho_{\bi}$ with as high a probability as possible.
The PGM is the measurement $M = \{M_1, \ldots, M_m\}$ defined by $M_i = S^{-1/2} \cdot \alpha_i \rho_i \cdot \cdot S^{-1/2}$,
where $S = \alpha_1 \rho_1 + \cdots + \alpha_m \rho_m$.
The PGM is in general not the optimal strategy for quantum hypothesis testing, but it is known that its success probability is always at least $P_\mathrm{OPT}^2$, where $P_\mathrm{OPT}$ is the best possible success probability~\cite{BK02}.
If we view pure state tomography as a sort of hypothesis testing problem in which each state $\ketbra{v}^{\otimes n}$ occurs with measure $\mathrm{d}v$, then carrying out the PGM construction gives us
\begin{equation*}
    S = \E_{\ket{\bv} \sim \mathrm{Haar}} \ketbra{\bv}^{\otimes n} = \frac{1}{\binom{d+n-1}{n}} \cdot \Pi_{\mathrm{Sym}},
\end{equation*}
where $\Pi_{\mathrm{Sym}}$ is the projector onto the symmetric subspace (see \cite[Proposition 6]{Har13} for a proof of this fact).
And then the measurement outcome corresponding to to the state $\ketbra{v}^{\otimes n}$ is
\begin{equation*}
    S^{-1/2} \cdot \ketbra{v}^{\otimes n} \cdot \mathrm{d}v \cdot S^{-1/2} = \binom{d+n-1}{n} \cdot \ketbra{v}^{\otimes n} \cdot \mathrm{d}v,
\end{equation*}
exactly as in Hayashi's measurement.

Generalizing this construction to mixed state tomography is difficult, partially because there is no obvious canonical measure on mixed states analogous to the Haar measure on pure states.
(Though one of the two tomography algorithms in Haah et al.\ \emph{is} derived from the PGM using some distributions on mixed states~\cite[Section 5]{HHJ+16}.)
However, for the special case of projector tomography, there is a natural measure we can use, which is the Haar measure on rank-$r$ projectors.
Using this, we can carry out the PGM construction, and we wind up with a measurement which is a natural generalization of Hayashi's measurement.

\section{Preliminaries}
\newcommand{\content}{\mathrm{cont}}
\newcommand{\hook}{\mathrm{hook}}
\newcommand{\specht}{\mathrm{Sp}}
\newcommand{\sign}{\mathrm{sign}}
\renewcommand{\P}{\mathcal{P}}
\renewcommand{\Q}{\mathcal{Q}}
\newcommand{\Sym}{\mathrm{Sym}^n(\C^d)}
\newcommand{\projSym}[1]{\Pi^{(#1)}_{\mathrm{Sym}}}

\newcommand{\SYT}{\mathrm{SYT}}
\newcommand{\SSYT}{\mathrm{SSYT}}
\newcommand{\USW}{\calU_{\mathrm{SW}}}
\newcommand{\USWdagger}{\calU_{\mathrm{SW}}^\dagger}

Throughout this paper, we will use the following conventions:
\begin{itemize}
    \item Random variables will be written in \textbf{boldface}. We use $\bx \sim \calD$ to denote that $\bx$ is drawn from the distribution $\calD$. 
    \item If $n$ is a positive integer, $[n]$ denotes the set $\{1, 2, \dots, n\}$. 
    \item We write $S_n$ for the symmetric group on $[n]$, and $U(d)$ for the group of $d \times d$ unitary matrices. 
    \item We will always take \emph{projectors} to mean \emph{orthogonal projectors}, i.e.\ projectors $\Pi$ which satisfy $\Pi = \Pi^\dagger$ and $\Pi^2 = \Pi$. Moreover, we write $\overline{\Pi} \coloneq I - \Pi$ for the projector onto the orthogonal complement of $\Pi$.
    \item If $\ket{\psi}$ is a pure state, we may also write $\psi$ for the corresponding mixed state $\ketbra{\psi}$.
\end{itemize}

\subsection{Quantum distance measures}

\begin{definition}[Schatten $p$-norm]
    Let $M \in \C^{d \times d}$ be an operator with singular values $\lambda_1, \dots, \lambda_d$. For $p \geq 1$, the \emph{Schatten $p$-norm} of $M$ is
    \begin{equation*}
        \norm{M}_p = \bigg( \sum_{i=1}^d |\lambda_i|^p \bigg)^{1/p}. 
    \end{equation*}
\end{definition}
Let $\rho, \sigma \in \C^{d \times d}$ be quantum states. Our main results concern the sample complexity of learning quantum states in the two most common distance measures, \emph{trace distance} and \emph{fidelity}, which we define next. 

\begin{definition}[Trace distance] \label{def:td}
    The \emph{trace distance} between $\rho$ and $\sigma$ is
    \begin{equation*}
        \Dtr(\rho, \sigma) = \frac{1}{2} \norm{ \rho - \sigma}_1.
    \end{equation*}
    When $\rho = \ketbra{u}$ and $\sigma = \ketbra{v}$ are pure states, we have 
    \begin{equation*}
        \Dtr(\rho, \sigma) = \sqrt{ 1 - \left|\braket{u}{v}\right|^2}.
    \end{equation*}
\end{definition}

\begin{definition}[Fidelity] \label{def:fid}
    The \emph{fidelity} of $\rho$ and $\sigma$ is
    \begin{equation*}
        \Fid(\rho, \sigma) = \norm{ \sqrt{\rho} \sqrt{\sigma} }_1 = \tr \sqrt{ \sqrt{\rho} \sigma \sqrt{\rho} }.
    \end{equation*}
     When $\rho = \ketbra{u}$ and $\sigma = \ketbra{v}$ are pure states, we have $\Fid(\rho, \sigma) = \left| \braket{u}{v} \right|$. The \emph{infidelity} of $\rho$ and $\sigma$ is the quantity $1 - \Fid(\rho, \sigma)$.
\end{definition}

Note that we are using the ``square root'' convention for fidelity. Fidelity and trace distance are related by the Fuchs-van de Graaf inequalities. 
%(\cite[Section 9.2]{NC10})
\begin{lemma}[Fuchs-van de Graaf inequalities, \cite{NC10}]  \label{fuchs}
    We have the following pair of inequalities:
    \begin{equation*}
        1 - \Fid(\rho, \sigma) \leq \Dtr(\rho, \sigma) \leq \sqrt{ 1 - \Fid(\rho, \sigma)^2}.
    \end{equation*}
\end{lemma}

Fidelity is not a metric on quantum states, but it is closely related to Bures distance, which is a metric. 

\begin{definition}[Bures distance] \label{def:Bur}
    The \emph{Bures distance} between $\rho$ and $\sigma$ is defined by
    \begin{equation*}
    \DBur(\rho, \sigma) = \sqrt{2 ( 1 - \Fid(\rho, \sigma))}.
    \end{equation*}
\end{definition}

Trace distance and Bures distance are also related as in the next lemma, which can be proven straightforwardly using the Fuchs-van de Graaf inequalities. 

\begin{lemma} We have the following pair of inequalities:
    \begin{equation*}
        \frac{1}{2} \DBur(\rho, \sigma)^2 \leq \Dtr(\rho, \sigma) \leq \DBur(\rho, \sigma).
    \end{equation*}
\end{lemma}

In particular, the Bures distance between the two states is always at least as large as the trace distance, making Bures distance a generally more challenging metric to learn in. We will also find it useful to work with \emph{affinity}. 

\begin{definition}[Affinity]
The \emph{affinity} between $\rho$ and $\sigma$ is given by
\begin{equation*}
\Aff(\rho, \sigma) = \tr( \sqrt{\rho} \sqrt{\sigma} ).
\end{equation*}
\end{definition}

Affinity is also not a metric. In both our upper and lower bounds on projector tomography, we will consider the affinity between two rank-$r$ projector states. If $\rho = P/r$ and $\sigma = Q/r$, the affinity is:
\begin{equation}
    \Aff(\rho, \sigma) = \frac{1}{r} \tr(\sqrt{P} \sqrt{Q} ) = \frac{1}{r} \tr( P Q) = r \tr( \rho \sigma). \label{eq:aff_projectors}
\end{equation}

The following lemma, which relates affinity and fidelity, allows us to easily convert bounds on affinity to bounds on fidelity, and vice versa. 

\begin{lemma}[\cite{ANSV08}] \label{general_fidelity_affinity_inequalities}
    We have the following pair of inequalities:
    \begin{equation*}
        \Fid(\rho, \sigma)^2 \leq \Aff(\rho, \sigma) \leq \Fid(\rho, \sigma).
    \end{equation*}
\end{lemma}

\noindent
In a subsequent section, we will give a simple proof of these inequalities for the special case of two rank-$r$ projector states (see \Cref{Jordans_lemma_td_fid_aff}).

We have defined trace distance, fidelity, Bures distance and affinity when $\rho$ and $\sigma$ are mixed states. However, the definitions can be extended, via the same formulas, to more general classes of matrices: trace distance is defined for all matrices; fidelity and affinity for all pairs of PSD matrices; Bures distance for all PSD matrices with fidelity at most one.

\subsection{The Haar measure}

\begin{definition}[Haar measure]
The \emph{Haar measure} on $U(d)$ is the unique distribution with the following property: if $\bU$ is distributed according to the Haar measure, then for any fixed unitary $V \in U(d)$, both $V \bU$ and $\bU V$ are distributed according to the Haar measure as well. We say $\bU$ is \emph{Haar random} and write $\bU \sim \mu_H$. We refer to the property that $V \bU$ (resp.\ $\bU V$) is Haar random as \emph{left-invariance} (resp.\ \emph{right-invariance}).
\end{definition}

\begin{notation}
    When integrating with respect to the Haar measure, we will write $\dU$ for the integration~measure.
\end{notation}

The Haar measure induces the following distributions on vectors and projectors. 

\begin{definition}[Haar random pure states, Haar random projectors]
    A \emph{Haar random pure state} in $\C^d$ is a random pure state $\ket{\bu}$ distributed as $\bU \ket{v}$, where $\bU \sim \mu_H$ and $\ket{v}$ is any fixed pure state. A \emph{Haar random rank-$r$ projector} is a random rank-$r$ projector $\bP$ distributed as $\bU Q \bU^\dagger$, for any fixed rank-$r$ projector $Q$. We will abuse notation and write $\ket{\bu} \sim \mu_H$ and $\bP \sim \mu_H$ when the meaning is clear from context.
\end{definition}

% \begin{remark}
%     The distribution on rank-$r$ projectors \john{which distribution on rank-$r$ projectors?} is the unique unitarily invariant measure on such projectors. In the literature, this measure is also known as the \emph{Grassmannian measure}, e.g.\ in \cite{Hayden_2006}. 
% \end{remark}

\subsection{Projector tomography algorithms}

\emph{Tomography} is the problem of producing an estimator $\widehat{\brho}$ of a quantum state $\rho \in \C^d$, given access to some number $n$ of copies of this state. We require that for any input $\rho$, $\D(\rho, \widehat{\brho}) \leq \epsilon$, for some pre-specified distance $\D$ and allowed error $\epsilon$, with \emph{high probability}. By ``high probability'', we mean a large constant probability of success, which we will take to be $99\%$. This probability threshold is somewhat arbitrary (see \Cref{lem:boosting_success_probability}). We will say that a tomography algorithm \emph{learns} a quantum state in distance measure $\D$, given $n$ samples. We write $\widehat{\brho} \sim \calA(\rho)$ to denote the output of a tomography algorithm $\calA$ on input $\rho^{\otimes n}$.

In this paper, we focus primarily on a special case of tomography called \emph{rank-$r$ projector tomography}. In this special case, the input state is necessarily a \emph{rank-$r$ projector state}. A rank-$r$ projector state is any state of the form $P/r$, where $P$ is a rank-$r$ projector. The $r=1$ case is \emph{pure state tomography}. Moreover, for us $\D$ will always be either trace or Bures distance.

Before moving on, we note that the only assumption about the output $\widehat{\brho}$ that we will make is that $\D(\rho, \widehat{\brho})$ is always defined. For example, the output of a pure state tomography algorithm might be mixed, or not even a quantum state at all. However, in the next subsection, we describe various properties we can bestow on a generic algorithm, without much cost.

\subsubsection{Upgrading projector tomography algorithms}

In this section, we describe a couple useful \emph{upgrades}\footnote{We borrow this terminology from \cite{FO24}.} we can give to tomography algorithms. These upgrades endow algorithms with additional properties that make our analysis simpler, at either no cost, or a constant-factor loss in accuracy. 

Firstly, we observe that any algorithm that does not quite meet our threshold for high probability may be upgraded to one that does. This lemma is particularly convenient for situations in which we union bound over multiple steps which themselves only succeed with high probability. The following is an immediate corollary of \cite[Proposition 2.4]{HKOT23}. 

\begin{lemma} \label{lem:boosting_success_probability}
    Suppose $\calA$ is a tomography algorithm using $n$ copies to output an estimate which is $\epsilon$-close in a metric $\D$ with probability at least $51\%$. Then there exists an algorithm $\calA'$ using $O(n)$ copies that outputs an estimate which is $3\epsilon$-close in metric $\D$ with probability at least $99\%$. 
\end{lemma}

% Next, we note a convenient description of any learning algorithm. 

% \begin{remark} \label{POVM_labels}
% Any quantum algorithm that takes $n$ copies of an unknown quantum state $\rho$ as input and outputs a classical description of an estimate $\widehat{\rho}$ (possibly using arbitrary entangled measurements, ancillas, and adaptive protocols) can be modeled as a single POVM $\{M_{\sigma}\}_{\sigma}$ acting on $\rho^{\otimes n}$, where the outcome label $\sigma$ is the algorithm’s output.
% \end{remark}

Next, recall that, per our definition of projector tomography, there is no guarantee that a rank-$r$ projector tomography algorithm only outputs rank-$r$ projector states. However, any algorithm can be converted to one that \emph{does} only output rank-$r$ projector states,
with only a small loss in accuracy.

% Now we show projector tomography algorithms can be made to output projectors.
% \john{By the way, we might want to introduce some additional context here, since I think it is natural that the reader might simply assume that the tomography algorithm should always output a rank-$r$ projector state. So we should include some verbiage about how ``Note that as we have defined projector tomography in Section blah, although the unknown state rho is always a rank-$r$ projector, there is no guarantee that the tomography algorithm only outputs rank-$r$ projector states.''}

\begin{lemma} \label{WLOG_projector_output}
    Suppose $\mathcal{A}$ is an algorithm for rank-$r$ projector tomography which uses $n$ copies to output an estimate which, with high probability, is $\epsilon$-close in metric $\D$. Then there exists an algorithm for rank-$r$ projector tomography $\mathcal{A}'$ which uses $n$ copies to output a rank-$r$ projector state which, with high probability, is $2\epsilon$-close in metric $\D$. 
    % Suppose algorithm $\mathcal{A}$ solves the rank-$r$ projector tomography problem $(d, r, \D, \epsilon)$ using $n$ copies. There exists an algorithm $\mathcal{A}'$ solving the rank-$r$ projector tomography problem $(d, r, \D, 2\epsilon)$ using $n$ copies, with the special property that $\mathcal{A}'$ always outputs a rank-$r$ projector. 
\end{lemma}

\begin{proof} 
    Have $\mathcal{A}'$ run $\mathcal{A}$ on $\rho^{\otimes n}$ to generate an output $\widehat{\brho}$ (which is not necessarily a projector state). If there exists a rank-$r$ projector state $\widehat{\brho}'$ such that $\D(\widehat{\brho}, \widehat{\brho}') \leq \epsilon$, then output such a $\widehat{\brho}'$. Otherwise, output an arbitrary fixed rank-$r$ projector state. With high probability, $\mathcal{A}$ succeeds in generating an estimate $\widehat{\brho}$ with $\D(\rho, \widehat{\brho}) \leq \epsilon$, and in this case, $\mathcal{A}'$ necessarily succeeds in finding a nearby projector state (since $\rho$ itself is a candidate), and outputs a $\widehat{\brho}'$ such that 
    \begin{equation*}
        \D(\rho, \widehat{\brho}') \leq \D(\rho, \widehat{\brho}) + \D(\widehat{\brho}, \widehat{\brho}') \leq 2\epsilon. \qedhere
    \end{equation*}
\end{proof}

% \begin{remark}\label{fidelity_vs_Bures}
%     Infidelity is not a metric, but the closely-related Bures distance is. Since $\DBur = \sqrt{2 (1 - \Fid)}$, $\mathcal{A}$ solves projector tomography with $(d, r, \DBur, \epsilon^2/2)$, if and only if it solves projector tomography with $(d, r, 1-\F, \epsilon)$, so that $\mathcal{A}'$, as in the above lemma, solves projector tomography with $(d, r, 1 - \F, \sqrt{2}\epsilon)$, and outputs projectors. So we can assume the lemma is true also for learning in fidelity. 
% \end{remark}

\begin{remark} 
    The proof of \Cref{WLOG_projector_output} is non-constructive. One concrete method to ``round''  $\widehat{\brho}$ into a rank-$r$ projector state is via the following construction: let $\widehat{\bP}$ be the projector onto the eigenvectors corresponding to the $r$ largest eigenvalues of $\widehat{\brho}$, with ties settled arbitrarily, and output $\widehat{\brho}' = \widehat{\bP}/r$. It can be shown that this rounding method attains the same guarantee as in \Cref{WLOG_projector_output} for learning in either trace or Bures distance. Since we will not need any concrete rounding, we omit the proof. 
\end{remark}

\subsection{Jordan's lemma}
Jordan's lemma is a standard tool in quantum information theory for understanding the relationship between a pair of projectors. The lemma and its proof are well known; we include a proof for completeness, based on \cite{Reg06}. Jordan's lemma provides us with formulas for distance measures between projector states that will be important for our bootstrapping argument. 

% We then apply Jordan's lemma to rank-$r$ projector states. 

% \thilo{Would replace last sentence with this to provide bigger picture understanding of why Jordan's Lemma will be relevant: "Jordan's lemma provides us with explicit formulas for distance measures between projector states and will be crucial for our bootstrapping argument."}

\begin{lemma}[Jordan's lemma]\label{Jordans_lemma}
    Let $P$ and $Q$ be projectors on a finite-dimensional Hilbert space $\mathcal{H}$. There exists an orthogonal decomposition of $\mathcal{H}$ into one- and two-dimensional subspaces which are invariant under $P$ and $Q$. Moreover, inside each two-dimensional subspace, $P$ and $Q$ each act as a projector onto a one-dimensional subspace.
\end{lemma}

\begin{proof}
    Consider the operator $R \coloneq P + Q$. Since $R$ is Hermitian, it has an orthonormal eigenbasis $\{ \ket{u_i} \}$, with corresponding eigenvalues $\{\lambda_i\}$.
    
    If $P \ket{u_1} = \mu_1 \ket{u_1}$, then we have
    \begin{equation*}Q \ket{u_1} = (R - P) \ket{u_1} = (\lambda_1 - \mu_1) \ket{u_1},\end{equation*}
    so that $\ket{u_1}$ is an eigenvector of both $P$ and $Q$. We set $B \coloneq \mathrm{span}(\ket{u_1})$, and note that $B$ is a one-dimensional subspace invariant under $P$ and $Q$. 
    
    Otherwise, consider the two-dimensional subspace $B \coloneqq \mathrm{span}( \ket{u_1}, P \ket{u_1} )$. Then $B$ is invariant under $P$, since $P^2 = P$. Moreover, $P|_B$ is rank-1, since it maps any element of $B$ into $\mathrm{span}( P\ket{u_1} )$. So, $P|_B$ is projection onto this one-dimensional subspace of $B$. However, $B$ is also invariant under $Q$, since first
    \begin{equation*}Q \ket{u_1} = (R - P) \ket{u_1} = \lambda_1 \ket{u_1} - P \ket{u_1} \in B,
    \end{equation*}
    and second
    \begin{equation*}
        QP \ket{u_1} = Q(R-Q) \ket{u_1}  = Q ( \lambda_1 - Q) \ket{u_1} = (\lambda_1 - 1) Q \ket{u_1} \in B.
    \end{equation*} 
    Note that $Q|_B$ is rank-1 as well, since it maps $B$ onto $\mathrm{span}(Q \ket{u_1})$, and is therefore a projector onto this one-dimensional subspace of $B$. 

    In either case $B$ is a one- or two-dimensional subspace invariant under $P$ and $Q$. Moreover, $B$ is also invariant under $R$. Since $R$ is Hermitian, $B^\perp$ is also invariant under $R$. We may then recurse on $B^\perp$ to obtain the desired decomposition of $\mathcal{H}$. 
\end{proof}

\begin{notation}[Jordan block decomposition]
    Let $P$ and $Q$ be projectors on a finite-dimensional Hilbert space $\mathcal{H}$, and let $\mathcal{H} = \bigoplus_{i} B_i$ be a decomposition into one- and two-dimensional subspaces, each invariant under $P$ and $Q$, as in \Cref{Jordans_lemma}. We call such a decomposition a \emph{Jordan block decomposition}, and refer to each $B_i$ as a \emph{Jordan block}.
\end{notation}

\begin{remark}\label{Jordans_lemma_rank_r_remark}
    We saw in \Cref{Jordans_lemma} that, if $B_i$ is a $2 \times 2$ block, $P|_{B_i}$ and $Q|_{B_i}$ are rank-1 projectors. If $B_i$ is a $1 \times 1$ block, $P|_{B_i}$ and $Q|_{B_i}$ are each individually either the identity or zero on that block. In the special case where $P$ and $Q$ are both rank-$r$, there are an equal number of $1 \times 1$ blocks $B_i$ with $P|_{B_i} = 1$ and $Q|_{B_i} = 0$ as there are blocks $B_j$ with $Q|_{B_j} = 1$ and $P|_{B_j} = 0$. In this case, by merging pairs of $1 \times 1$ blocks (one block of each kind) into a single $2 \times 2$ block, we can assume every block $B_i$ in which $R|_{B_i} \neq 0$ contains two states $\ket{u_i}$ and $\ket{v_i}$, such that $P|_{B_i} = \ketbra{u_i}$ and $Q|_{B_i} = \ketbra{v_i}$. That is, we can assume that there are exactly $r$ blocks in which $P$ and $Q$ both act nontrivially as projectors onto one-dimensional subspaces, and $P$ and $Q$ are zero outside of these $r$ blocks. 
\end{remark}

\begin{definition}[Jordan vectors]
    Let $P$ and $Q$ be rank-$r$ projectors on a finite-dimensional Hilbert space $\mathcal{H}$, and take a Jordan block decomposition as in \Cref{Jordans_lemma_rank_r_remark}, so that $P = \sum_{i=1}^r \ketbra{u_i}$ and $Q = \sum_{i=1}^r \ketbra{v_i}$, with $\ket{u_i}$ and $\ket{v_i}$ in the $i$-th block $B_i$. We say $\{ \ket{u_i} \}$ and $\{\ket{v_i}\}$ are \emph{Jordan vectors} of $P$ and $Q$, respectively. 
\end{definition} 

% \jack{Thilo, if you prefer the non-intersecting corners version, could you make the text in the figure larger, and match to update the new notation?} \thilo{Done}

\begin{figure}[h]
\centering
\begin{adjustbox}{max width=0.95\textwidth}
\begin{tikzpicture}[every node/.style={anchor=center}, 
  jblock/.style={draw=gray, thick, rounded corners}]

% Matrix P
\matrix (m1) [matrix of nodes,
              nodes in empty cells,
              nodes={minimum width=3em, minimum height=3em, anchor=center},
              left delimiter={[},
              right delimiter={]}] at (0,0)
{
  $\ketbra{u_1}{u_1}$ & 0 & 0 & 0 \\
  0 & $\ketbra{u_2}{u_2}$ & 0 & 0 \\
  0 & 0 & & \\
  0 & 0 & & \\
};

% Matrix Q
\matrix (m2) [matrix of nodes,
              nodes in empty cells,
              nodes={minimum width=3em, minimum height=3em, anchor=center},
              left delimiter={[},
              right delimiter={]},
              right=4cm of m1] 
{
  $\ketbra{v_1}{v_1}$ & 0 & 0 & 0 \\
  0 & $\ketbra{v_2}{v_2}$ & 0 & 0 \\
  0 & 0 & & \\
  0 & 0 & & \\
};

% Jordan blocks P (using the non-overlapping jblock style)
\draw[jblock] (m1-1-1.north west) rectangle (m1-1-1.south east);
\draw[jblock] (m1-2-2.north west) rectangle (m1-2-2.south east);
\draw[jblock] (m1-3-3.north west) rectangle (m1-4-4.south east);
\node at ($(m1-3-3)!0.5!(m1-4-4)$) {$\ketbra{u_{3}}{u_{3}}$};

% Jordan blocks Q (using the non-overlapping jblock style)
\draw[jblock] (m2-1-1.north west) rectangle (m2-1-1.south east);
\draw[jblock] (m2-2-2.north west) rectangle (m2-2-2.south east);
\draw[jblock] (m2-3-3.north west) rectangle (m2-4-4.south east);
\node at ($(m2-3-3)!0.5!(m2-4-4)$) {$\ketbra{v_{3}}{v_{3}}$};

% Labels under matrices
\node[below=0.5cm of m1] {\( P = \sum_i \ketbra{u_i}{u_i} \)};
\node[below=0.5cm of m2] {\( Q = \sum_i \ketbra{v_i}{v_i} \)};
\end{tikzpicture}
\end{adjustbox}
\caption{An illustration of an example Jordan block decomposition. The two projectors $P$ and $Q$ can be simultaneously block-diagonalized, and each Jordan block is either $1 \times 1$ or $2 \times 2$. Within a given block, if $P$ (resp.\ $Q$) is nontrivial, then $P$ (resp.\ $Q$) acts as a projection onto a one-dimensional subspace}
\end{figure}

Given two rank-$r$ projector states, we can use Jordan's lemma to evaluate block-by-block quantities like trace distance, fidelity, and affinity.

\begin{lemma}\label{Jordans_lemma_td_fid_aff}
    Let $\rho = P/r$ and $\sigma = Q/r$ be rank-$r$ projector states, and let $\{ \ket{u_i}\}$ and $\{\ket{v_i}\}$ be Jordan vectors of $P$ and $Q$ respectively. Write $\omega_i = |\braket{u_i}{v_i}|$. Then we have:
    \begin{itemize}
         \item $\Dtr(\rho, \sigma) = \big(\sum_{i=1}^r \sqrt{1 - \omega_i^2}\big)/r$,
         \item $\Fid(\rho, \sigma) = \left(\sum_{i=1}^r \omega_i\right)/r$,
         \item $\Aff(\rho, \sigma) = \left(\sum_{i=1}^r \omega_i^2\right)/r$.
    \end{itemize}
\end{lemma}

\begin{proof}
    We evaluate each quantity block-by-block using the Jordan block decomposition. Note that each quantity we want to calculate is additive in the blocks of the Jordan decomposition. Thus, 
    \begin{equation*}
        \Dtr(\rho, \sigma) = \frac{1}{r} \cdot \Dtr(P, Q) = \frac{1}{r} \sum_{i=1}^{r} \Dtr( \ketbra{u_i}, \ketbra{v_i} ) = \frac{1}{r} \sum_{i=1}^{r} \sqrt{ 1 - \omega_i^2}.
    \end{equation*}
    Similarly, 
    \begin{equation*}
        \Fid(\rho, \sigma) = \frac{1}{r} \cdot \Fid(P, Q) = \frac{1}{r} \sum_{i=1}^{r} \Fid( \ketbra{u_i}, \ketbra{v_i} ) = \frac{1}{r} \sum_{i=1}^{r} \omega_i.
    \end{equation*}
    Finally, 
    \begin{equation*}
        \Aff(\rho, \sigma) = \frac{1}{r} \cdot \Aff(P, Q) = \frac{1}{r} \sum_{i=1}^{r} \Aff( \ketbra{u_i}, \ketbra{v_i} ) = \frac{1}{r} \sum_{i=1}^r \tr\big( \ketbra{u_i} \cdot \ketbra{v_i} \big) = \frac{1}{r} \sum_{i=1}^{r} \omega_i^2. \qedhere
    \end{equation*}
\end{proof}

Affinity is sometimes easier to analyze than fidelity. The following corollary, a special case of \Cref{general_fidelity_affinity_inequalities}, allows us to convert from bounds on one to bounds on the other.

\begin{corollary} \label{cor:aff-fid}
    Let $\rho = P/r$ and $\sigma = Q/r$ be rank-$r$ projector states. Then 
    \begin{equation*}\Aff(\rho, \sigma) \leq \Fid(\rho, \sigma) \leq \sqrt{\Aff(\rho, \sigma)}.\end{equation*}
    As a result, if $0 \leq \epsilon \leq 1$, then $\Aff(\rho, \sigma) \geq 1 - \epsilon$ implies $\Fid(\rho, \sigma) \geq 1 - \epsilon$, and $\Fid(\rho, \sigma) \geq 1 - \epsilon$ implies $\Aff(\rho, \sigma) \geq 1 - 2\epsilon$. 
\end{corollary}

\begin{proof}
    Let $\{ \ket{u_i}\}$ and $\{\ket{v_i}\}$ be Jordan vectors of $P$ and $Q$ respectively, and write $\omega_i = \left|\braket{u_i}{v_i}\right|$. By \Cref{Jordans_lemma_td_fid_aff}, 
    \begin{equation}\label{fid_aff_eq1}
        \Aff(\rho, \sigma) = \frac{1}{r} \sum_{i=1}^r \omega_i^2 \leq \frac{1}{r} \sum_{i=1}^r \omega_i = \Fid(\rho, \sigma),
    \end{equation}
    using $0 \leq \omega_i \leq 1$. Moreover,
    \begin{equation}\label{fid_aff_eq2}
        \Fid(\rho, \sigma)^2 = \frac{1}{r^2} \Big(\sum_{i=1}^r \omega_i\Big)^2 \leq \frac{1}{r^2} \Big(r \cdot \sum_{i=1}^r \omega_i^2 \Big) =  \frac{1}{r} \sum_{i=1}^r \omega_i^2 =  \Aff(\rho, \sigma),
    \end{equation}
    where the inequality is Cauchy-Schwarz.
\end{proof}

\subsection{L\'{e}vy's lemma}

L\'{e}vy's lemma is another classic tool in quantum information theory. Loosely speaking, it tells us that nice functions on high-dimensional spheres concentrate exponentially about their means. We will use it in our reduction from trace distance projector tomography to Bures distance projector tomography. 

\begin{lemma}[L\'{e}vy's lemma, \protect{\cite[Theorem 7.37]{Wat18}}] \label{levys}
    Let $f$ be a function from pure states in $\C^d$ to $\R$, and let $f$ be $L$-Lipschitz, meaning that $\big| f(\ket{u}) - f(\ket{v}) \big| \leq L \cdot \norm{ \, \ket{u} - \ket{v} \,}_2$ (where $\norm{\cdot}_2$ is the $\ell_2$-norm). Write $f_{\mathrm{avg}} \coloneq \E_{\ket{\bu} \sim \mu_H}\big[f(\ket{\bu})\big]$. Then, for some universal constant $C > 0$, and for any $\epsilon > 0$, 
    \begin{equation*}
        \Pr_{\ket{\bu} \sim \mu_H} \Big[ \big| f(\ket{\bu}) - f_{\mathrm{avg}} \big| > \epsilon \Big] \leq 3 \exp \Big( - \frac{C \epsilon^2 d}{L^2}\Big).
    \end{equation*}
\end{lemma}

In our proof, we will end up applying L\'{e}vy's lemma in the case where $f$ is  the expectation value of a projector, i.e.\ $f( \ket{u} ) = \matrixel{u}{P}{u}$, for some projector $P$. We will therefore need the following result.

    \begin{lemma} Let $P$ be a projector, and let $f$ be the function on pure states in $\C^d$ defined $f( \ket{u} ) \coloneq \matrixel{u}{P}{u}$. Then $f$ is $1$-Lipschitz. \label{measurementslipschitz}
    \end{lemma}

    \begin{proof} Trace distance has the alternate characterization:
    \begin{equation*}
        \Dtr(\rho, \sigma) = \max_Q \Big[\tr \big( Q (\rho - \sigma) \big)\Big],
    \end{equation*}
    where the maximization is over all projectors $Q$ (see, for example, \cite[Equation 9.62]{NC10}). Therefore,
    \begin{equation*}
        \big| f(\ket{u}) - f( \ket{v} ) \big| = \left| \matrixel{u}{P}{u} -  \matrixel{v}{P}{v}\right|  = \left|\tr \big( P \left( \ketbra{u} - \ketbra{v} \right) \big)\right| \leq \dtr\left( \ketbra{u}, \ketbra{v} \right).
    \end{equation*}
    The trace distance of two pure states is 
    \begin{equation*}
        \dtr\left(\ketbra{u}, \ketbra{v} \right) = \sqrt{ 1 - \left|\braket{u}{v}\right|^2} = \sqrt{  1 + \left|\braket{u}{v}\right|}\cdot  \sqrt{1 - \left|\braket{u}{v}\right|} \leq \sqrt{2} \cdot \sqrt{ 1 - \left|\braket{u}{v}\right|}. 
    \end{equation*}
    However, note that
    \begin{equation*}
        1 - \left| \braket{u}{v} \right|  \leq  1 - \mathrm{Re}( \braket{u}{v} ) = \frac{1}{2} \norm{ \,\ket{u} - \ket{v}\, }_2^2.
    \end{equation*}
    Combining everything, we conclude that $\left| f(\ket{u}) - f( \ket{v} ) \right| \leq \norm{ \, \ket{u} - \ket{v}\,}_2$. 
    \end{proof}

\subsection{The symmetric subspace}

In this section, we collect a few relevant and standard facts about the symmetric subspace, which we define for completeness. For much more on the symmetric subspace, including proofs, see \cite{Har13} or \cite{Mel24}. 

\begin{definition}[Symmetric subspace]
    The \emph{symmetric subspace}, denoted $\Sym$, is the subspace of $(\C^d)^{\otimes n}$ given by 
    \begin{equation*}
        \Sym = \mathrm{span}\big\{ \ket{u}^{\otimes n} \, : \, \ket{u} \in \C^d \big\}.
    \end{equation*}
    We denote the projector onto the symmetric subspace as $\projSym{n,d}$, though we will sometimes drop $n$ or $d$ when clear from context.
\end{definition}

\begin{lemma} \label{sym_subspace_permutation_invariant} The symmetric subspace is equal to the span of all permutation-invariant vectors in $(\C^d)^{\otimes n}$, i.e.\
\begin{equation*}
    \mathrm{Sym}^n(\C^d) = \mathrm{span}\big\{ \ket{\psi} \in (\C^d)^{\otimes n} \, : \, \P(\pi)\ket{\psi} = \ket{\psi} \text{ for all } \pi \in S_n \big\},
\end{equation*}
where $\mathcal{P}(\pi)$ is the representation of $\pi$ that permutes the $n$ tensor factors according to $\pi$, formally defined below in \Cref{prelim_P_and_Q}. 
\end{lemma}

% \jack{If we don't include John's optimality of PGM proof, we don't need type vectors}
% \begin{definition}[Type vectors]
%     Given a string $x = (x_1, \dots, x_n) \in [d]^n$, define its \emph{type} $\tau = (\tau_1, \dots, \tau_d)$ by the histogram of $x$, i.e.\ $\tau_i$ is the number of occurrences of $i$ in $x$, and write $\mathrm{type}(x) = \tau$. We associate $\tau$ with the set $\{ x \in [d]^n \, : \, \mathrm{type}(x) = \tau \}$. We also define the \emph{type vector} corresponding to $\tau$ as
%     \begin{equation*}
%         \ket{\tau} = \frac{1}{\sqrt{|\tau|}} \sum_{x \in \tau} \ket{x}. 
%     \end{equation*}
% \end{definition}

% \begin{lemma}
%     The type vectors form an orthonormal basis of $\Sym$. 
% \end{lemma}

% \jack{Back to things we need}
% \john{then let's get rid of 'em lol}

\begin{lemma} \label{dimension_sym_subspace}
The dimension of the symmetric subspace is
\begin{equation*}
    \dim( \Sym ) = \tr ( \projSym{n,d} ) = \binom{n + d -1}{n}.
\end{equation*}
\end{lemma}

\begin{lemma}\label{proj_sym_subspace}
The average over $n$-fold products of Haar random states is proportional to the projector onto the symmetric subspace:
\begin{equation*}
    \E_{\ket{\bu} \sim \mu_H} \big[\ketbra{\bu}^{\otimes n}\big] = \frac{1}{\binom{d+n-1}{n}} \cdot \projSym{n,d}.
\end{equation*}
\end{lemma}

\subsection{Representation theory}

Algorithms for tomography often use representation theory to utilize the symmetry of the input state $\rho^{\otimes n}$. In this section, we review the representation theory necessary for our results. Our coverage is based primarily on \cite[Chapter 2]{Wri16}. Other sources we draw on include \cite{Sag01,GW09,Ful97,Har05}. 

\subsubsection{Basics}

This section collects general definitions and results we will need. Our goal here is mainly to establish notation; detailed exposition, including proofs, can be found e.g.\ in \cite[Chapter 1]{Sag01}.

Let $U(V)$ denote the group of all unitary operators on a complex vector space $V$. 

\begin{definition}[Representations]
    Let $G$ be a group. A \emph{complex, unitary, finite-dimensional representation}, of $G$ is a pair $(\mu,V)$, where $V$ is a finite-dimensional complex vector space, and $\mu: G \to U(V)$ is a group homomorphism. The \emph{dimension} of the representation, written $\dim(\mu)$, is the dimension of $V$.
\end{definition}

Since we will not consider more general representations, we will refer to complex, unitary, finite-dimensional representations simply as \emph{representations}. We will also abbreviate a representation $(\mu,V)$ either as $\mu$ or $V$, when the meaning is clear from context.

\begin{definition}[Characters]
    Let $\mu$ be a representation of a group $G$. The \emph{character} of $\mu$ is the map $\chi_\mu: G \to \C$ given by $\chi_\mu(g) \coloneq \tr(\mu(g))$. 
\end{definition}

\begin{definition}[Intertwining operators]
    Let $(\mu_1,V_1)$ and $(\mu_2,V_2)$ be representations of a group $G$. An \emph{intertwining map}, or \emph{intertwiner}, between $\mu_1$ and $\mu_2$ is a map $T: V_1 \to V_2$ such that $T \cdot \mu_1(g) = \mu_2(g) \cdot T$, for all $g \in G$. 
\end{definition}

\begin{definition}[Isomorphic representations]
    Let $\mu_1$ and $\mu_2$ be representations of a group $G$. Then $\mu_1$ and $\mu_2$ are \emph{isomorphic representations}, or \emph{equivalent}, if there exists an invertible intertwining operator between $\mu_1$ and $\mu_2$. We write $\mu_1 \cong \mu_2$. 
\end{definition}

\begin{definition}[Irreducible representations]
    Let $(\mu, V)$ be a representation of a group $G$. A subspace $W \subseteq V$ is an \emph{invariant subspace} of $V$ if $\mu(g) \cdot W \subseteq W$ for all $g \in G$. An invariant subspace is \emph{trivial} if $W = \{0\}$ or $W = V$. If $V$ has a nontrivial invariant subspace, it is called \emph{reducible}, and otherwise is called \emph{irreducible}. An irreducible representation is also called an \emph{irrep}. The set of equivalence classes of irreps of $G$ will be written $\widehat{G}$. 
\end{definition}

We can fix a representative $\widehat{\mu}_i$ from each class of irreps, and identify $\widehat{G}$ with $\{ \widehat{\mu}_i\}$. 

\begin{lemma}[Schur's lemma] \label{lem:Schur's_lemma}
    Let $(\mu_1,V_1)$ and $(\mu_2,V_2)$ be irreducible representations of a group $G$, with an intertwining operator $T: V_1 \to V_2$. 
    \begin{itemize}
        \item If $\mu_1$ and $\mu_2$ are non-isomorphic, then $T = 0$.
        \item If $\mu_1 = \mu_2$, then $T = c \cdot I_{V_1}$, for some constant $c \in \C$. 
    \end{itemize} 
\end{lemma}

The following result can be proven by Schur's lemma. 

\begin{corollary} \label{cor:unitary_isomorphism}
    Let $(\mu_1, V_1)$ and $(\mu_2,V_2)$ be isomorphic representations of a group $G$. Then there exists a \emph{unitary} $U: V_1 \to V_2$ which intertwines $\mu_1$ and $\mu_2$. That is, for all $g \in G$, 
    \begin{equation*}
        U \cdot \mu_1(g) \cdot U^\dagger = \mu_2(g).
    \end{equation*}
\end{corollary}

\begin{definition}[Direct sum of representations]
    Let $(\mu_1, V_1)$, $\dots$, $(\mu_k, V_k)$ be representations of a group $G$. The \emph{direct sum} of $\mu_1$, $\dots$, $\mu_k$ is the representation $(\mu, V)$, where $V \coloneq \bigoplus_{i=1}^{k} V_i$, and $\mu(g) \coloneq \bigoplus_{i=1}^{k} \mu_i(g) = \sum_{i=1}^k \ketbra{i} \otimes \mu_i(g)$. Representations may occur more than once, and sometimes it will be convenient to allow them to occur zero times; in this case, we may also write $\mu(g) = \bigoplus_{i=1}^k m_i \cdot \mu_i(g)$, where the $\{m_i\}$ are nonnegative integers, and $m_i$ is the \emph{multiplicity} of $\mu_i$. 
\end{definition}

\begin{definition}[Complete reducibility]
    Let $\mu$ be a representation of $G$. Then $\mu$ is \emph{completely reducible} if 
    \begin{equation*}
        \mu \cong \bigoplus_{\widehat{\mu}_i \in \widehat{G}} m_i \cdot \widehat{\mu}_i,
    \end{equation*}
    for some nonnegative integers $\{m_i\}$.
\end{definition}

Every finite-dimensional, unitary representation is completely reducible.

\subsubsection{Partitions and Young diagrams}

The representation theory of the symmetric and unitary groups turns out to be connected to partitions. We now give a brief overview of relevant partition-related concepts. 

\begin{definition}[Partitions]
Let $n$ be a positive integer. A \emph{partition} of $n$ is a finite list $\lambda = (\lambda_1, \dots, \lambda_m)$ of nonnegative integers such that $\lambda_1 \geq \dots \geq \lambda_m \geq 0$ and $\lambda_1 + \dots + \lambda_m = n$. We write $\lambda \vdash n$ to denote that $\lambda$ is a partition of $n$. The \emph{length} of $\lambda$, written $\ell(\lambda)$, is the largest index $i$ such that $\lambda_i > 0$. The \emph{size} of $\lambda$ is $|\lambda| = n$. 
\end{definition}

Partitions can be represented pictorially with Young diagrams.

\begin{definition}[Young diagrams] 
Let $\lambda \vdash n$. A \emph{Young diagram} of shape $\lambda$ is a diagram consisting of $n$ boxes, arranged into $\ell(\lambda)$ left-justified rows, such that the $i$-th row contains $\lambda_i$ boxes. We will also refer to boxes interchangeably as \emph{cells}.
\end{definition}

We identify partitions with their Young diagrams, and use the two interchangeably. 

% \begin{figure}
%     \centering
%     \begin{ytableau}
%           ~ & ~ & ~ & ~ & *(gray!30)~ & *(gray!30)~  \\
%           ~ & ~ & ~ & ~ & *(gray!30)~ & *(gray!30)~ \\
%           ~ & ~ & *(gray!30)~ & *(gray!30)~ & *(gray!30)~ & *(gray!30)~  \\
%           ~ &  *(gray!30)~ & *(gray!30)~ & *(gray!30)~ & *(gray!30)~ & *(gray!30)~ \\ 
%            *(gray!30)~ & *(gray!30)~ & *(gray!30)~ & *(gray!30)~ & *(gray!30)~ & *(gray!30)~
%     \end{ytableau} 
% \end{figure}

\begin{figure}[h]
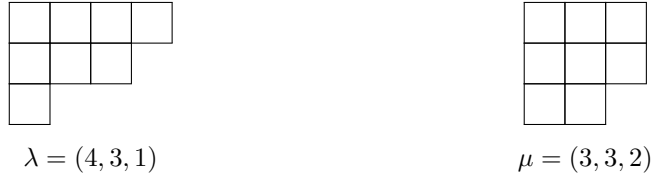

    \centering
    \begin{minipage}{0.35\textwidth}
        \centering
        \begin{ytableau}
              ~ & ~ & ~ & ~ \\
              ~ & ~ & ~ \\
              ~ 
        \end{ytableau}

        \vspace{0.7em} % adjust this spacing as needed
        $\lambda = (4,3,1)$
    \end{minipage}
    \hspace{1.6em}
    \begin{minipage}{0.35\textwidth}
        \centering
        \begin{ytableau}
            ~ & ~ & ~ \\ 
            ~ & ~ & ~ \\
            ~ & ~
        \end{ytableau}

        \vspace{0.7em}
        $\mu = (3,3,2)$
    \end{minipage}

    \caption{Two partitions of $n = 8$, and their Young diagrams. Left: $\lambda = (4,3,1)$. Right: $\mu = (3,3,2)$.}
    \label{fig:example_Young_diagrams}
\end{figure}

% \begin{figure}[h]
% \centering
% \begin{tikzpicture}[scale=0.6, every node/.style={font=\small}]

% % Partition (4,3,1)
% \begin{scope}[shift={(0,0)}]
%   \foreach \x/\y in {0/0,1/0,2/0,3/0,0/1,1/1,2/1,0/2}{
%     \draw[thick] (\x,-\y) rectangle ++(1,-1);
%   }
%   \node at (1.8,-3.6) {$\lambda=(4,3,1)$};
% \end{scope}

% % Partition (3,3,2)
% \begin{scope}[shift={(7,0)}]
%   \foreach \x/\y in {0/0,1/0,2/0,0/1,1/1,2/1,0/2,1/2}{
%     \draw[thick] (\x,-\y) rectangle ++(1,-1);
%   }
%   \node at (1.5,-3.6) {$\mu=(3,2,2,1)$};
% \end{scope}

% \end{tikzpicture}
% \caption{Two partitions of $n=8$ and their Young diagrams. Left: $\lambda=(4,3,1)$. Right: $\mu=(3,3,2)$.}
% \label{fig:partitions}
% \end{figure}

\begin{notation}[Additional notation] \label{not:additional}
We define the following notation:
\begin{itemize}
     \item The box in the $i$-th row and $j$-th column is denoted by $(i,j)$. 
     \item The \emph{content} of $(i,j)$ is $\content(i,j) \coloneqq j - i$.
     \item The \emph{hook length} of $(i,j)$ in $\lambda$, $\hook_\lambda(i,j)$, is the number of boxes $(k,\ell)$ in $\lambda$ such that $k = i$ and $\ell \geq j$ or $\ell = j$ and $k \geq i$. Informally, the hook length is the number of boxes either directly to the right of $(i,j)$, or directly below, including $(i,j)$ itself. 
     \item Let $\mu$ be a partition. Then $\mu$ \emph{contains} $\lambda$ if $\ell(\mu) \geq \ell(\lambda)$ and for each $i \in [\ell(\lambda)]$, we have $\mu_i \geq \lambda_i$. When $\mu$ contains $\lambda$, we write $\lambda \subseteq \mu$, and we write $\mu \setminus \lambda$ for the set of boxes in $\mu$ but not in $\lambda$, where we identify boxes in $\lambda$ and $\mu$ with the same label $(i,j)$. Alternatively, if we view partitions as subsets of $\N \times \N$, containment of partitions is just set containment.
     \item If $\mu$ can be obtained from adding a single box in row $i$ to $\lambda$, then we write $\mu = \lambda+e_i$. We also write $\mu = \lambda + k \cdot e_i$ if $\mu$ can be obtained by adding $k$ boxes to the $i$-th row.
\end{itemize}
\end{notation}

\begin{figure}[h!]
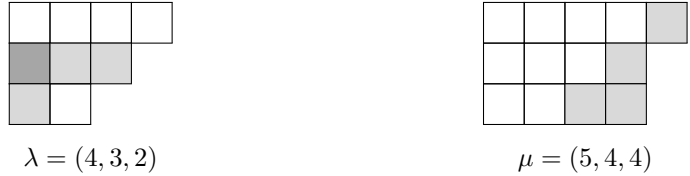

    \centering
    \begin{minipage}{0.35\textwidth}
        \centering
        \begin{ytableau}
              ~ & ~ & ~ & ~ \\
              *(gray!70)~ & *(gray!30)~ & *(gray!30)~ \\
              *(gray!30)~ & ~
        \end{ytableau}

        \vspace{0.7em} % adjust this spacing as needed
        $\lambda = (4,3,2)$
    \end{minipage}
    \hspace{1.6em}
    \begin{minipage}{0.35\textwidth}
        \centering
        \begin{ytableau}
            ~ & ~ & ~ & ~ & *(gray!30)~\\ 
            ~ & ~ & ~ & *(gray!30)~\\
            ~ & ~ & *(gray!30)~ & *(gray!30)~
        \end{ytableau}

        \vspace{0.7em}
        $\mu = (5,4,4)$
    \end{minipage}

    \caption{Illustration of \Cref{not:additional}. Left: a partition $\lambda = (4,3,1)$. The box $(2,1)$ is shaded in dark gray, and the remaining boxes of $\hook_\lambda(2,1)$ are shaded light gray. The content of $(2,1)$ is $1-2 = -1$, and its hook length is $4$. Right: a partition $\mu = (5,4,4)$. We have $\lambda \subseteq \mu$, and the boxes of $\mu \setminus \lambda$ are shaded.}
\end{figure}

\begin{definition}[Standard Young tableaux]
    A \emph{standard Young tableau} (SYT) of shape $\lambda \vdash n$ is a bijective labeling  of the boxes of $\lambda$ with the numbers in $[n]$, such that the labels are strictly increasing rightwards along rows, and downwards along columns. The set of all SYTs of shape $\lambda$ is denoted $\SYT(\lambda)$. 
\end{definition}

\begin{definition}[Semistandard Young tableaux]
Fix a positive integer $d$. A \emph{semistandard Young tableau} (SSYT) of shape $\lambda$ and alphabet $[d]$ is a labelling of the boxes of $\lambda$, such that the labels are \emph{weakly} increasing rightwards along rows, and \emph{strictly} increasing downwards along columns. The set of all SSYTs of shape $\lambda$ and alphabet $d$ is denoted $\SSYT(\lambda,d)$. 
\end{definition}

\begin{figure}[h!]
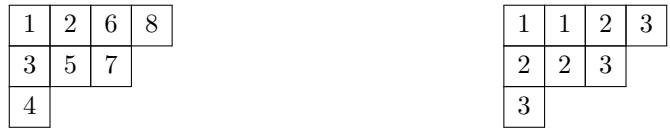

    \centering
    \begin{minipage}{0.35\textwidth}
        \centering
        \begin{ytableau}
              1 & 2 & 6 & 8\\
              3 & 5 & 7 \\
              4 
        \end{ytableau}

        \vspace{0.7em} % adjust this spacing as needed
    \end{minipage}
    \hspace{1.6em}
    \begin{minipage}{0.35\textwidth}
        \centering
        \begin{ytableau}
            1 & 1 & 2 & 3 \\ 
            2 & 2 & 3 \\
            3 
        \end{ytableau}

        \vspace{0.7em}
    \end{minipage}

    \caption{Examples of tableaux of shape $\lambda = (4,3,1)$. Left: an SYT. Right: an SSYT for $d \geq 3$.}
    \label{fig:example_Young_tableaux}
\end{figure}

% \begin{figure}[h]
% \centering
% \begin{tikzpicture}[scale=0.55, every node/.style={font=\small}]

% % SYT for lambda=(4,3,1)
% \begin{scope}[shift={(0,0)}]
%   % shape (4,3,1)
%   \foreach \x/\y in {0/0,1/0,2/0,3/0,0/1,1/1,2/1,0/2}{
%     \draw[thick] (\x,-\y) rectangle ++(1,-1);
%   }
%   % entries in larger font
%   \node at (0.5,-0.5) {\normalsize 1};
%   \node at (1.5,-0.5) {\normalsize 2};
%   \node at (2.5,-0.5) {\normalsize 6};
%   \node at (3.5,-0.5) {\normalsize 8};
%   \node at (0.5,-1.5) {\normalsize 3};
%   \node at (1.5,-1.5) {\normalsize 5};
%   \node at (2.5,-1.5) {\normalsize 7};
%   \node at (0.5,-2.5) {\normalsize 4};

%   \node at (1.8,-4.0) {SYT of shape $(4,3,1)$};
% \end{scope}

% % SSYT for lambda=(4,3,1)
% \begin{scope}[shift={(9,0)}]
%   % shape (4,3,1)
%   \foreach \x/\y in {0/0,1/0,2/0,3/0,0/1,1/1,2/1,0/2}{
%     \draw[thick] (\x,-\y) rectangle ++(1,-1);
%   }
%   % entries in larger font (alphabet [3])
%   \node at (0.5,-0.5) {\normalsize 1};
%   \node at (1.5,-0.5) {\normalsize 1};
%   \node at (2.5,-0.5) {\normalsize 2};
%   \node at (3.5,-0.5) {\normalsize 3};
%   \node at (0.5,-1.5) {\normalsize 2};
%   \node at (1.5,-1.5) {\normalsize 2};
%   \node at (2.5,-1.5) {\normalsize 3};
%   \node at (0.5,-2.5) {\normalsize 3};

%   \node at (1.8,-4.0) {SSYT of shape $(4,3,1)$};
% \end{scope}

% \end{tikzpicture}
% \caption{Examples of tableaux of shape $(4,3,1)$. Left: a standard Young tableau (SYT). Right: a semistandard Young tableau (SSYT) for $d=3$.}
% \label{fig:tableaux}
% \end{figure}

\subsubsection{Irreducible representations of $S_n$ and $U(d)$}

We now turn to the representation theory of the symmetric and unitary groups. We begin with descriptions of the two groups' irreps, first considering $S_n$. 

\begin{theorem} The irreducible representations of $S_n$ are in bijection with partitions $\lambda$ such that $\lambda \vdash n$.  
\end{theorem}

\begin{definition}
    The irrep of $S_n$ corresponding to $\lambda \vdash n$ is denoted $(\kappa_\lambda, \specht_\lambda)$, where $\specht_\lambda$ is called the \emph{Specht module}. We abbreviate $\dim(\specht_\lambda)$ as $\dim(\lambda)$. 
\end{definition}

\begin{theorem}
    There is a basis of $\specht_\lambda$ for which each basis element is bijectively associated with an SYT of shape $\lambda$, so that $\dim(\lambda) = | \SYT(\lambda)|$.
\end{theorem}

For $U(d)$ we focus only on \emph{polynomial} irreps.

\begin{definition}[Polynomial representations]
    Let $(\mu,V)$ be a representation of a matrix group $G$. Then $\mu$ is a \emph{polynomial representation} if we can pick a basis for $V$ such that the entries of the matrix $\mu(g)$ are polynomials in the entries of $g \in G$.\footnote{Note that if the entries are polynomials in some basis, then they are polynomials in any basis. That is, a representation being polynomial is a basis-independent notion.} 
\end{definition}

\begin{theorem}
The polynomial irreps of $U(d)$ are in bijection with partitions $\lambda$ such that $\ell(\lambda) \leq d$.  
\end{theorem}

\begin{definition}
    The irrep corresponding to $\lambda$, a partition with $\ell(\lambda) \leq d$, is denoted $(\nu_\lambda, V^d_\lambda)$, where $V^d_\lambda$ is called the \emph{Schur module}.
\end{definition} 

\begin{theorem}
There is a basis of $V^d_\lambda$ for which each basis element is bijectively associated with an SSYT of shape $\lambda$ and alphabet $[d]$, so that $\dim(V^d_\lambda) = | \SSYT(\lambda,d)|$. 
\end{theorem}

\begin{theorem}[Hook-content formula, \protect{\cite[Theorem 7.21.2]{Sta99}}]\label{thm:hook_content_formula}
We have
\begin{equation*}
        |\SSYT(\lambda,d)| =  \prod_{(i,j) \in \lambda} \frac{d + \content(i,j)}{\hook_\lambda(i,j)}.
\end{equation*}
\end{theorem}

\begin{example}[Defining representation]\label{ex:defining_rep}
    The defining representation of $U(d)$ is the representation $(\mu, V)$ such that $V = \C^d$ and $\mu(U) = U$. Since no subspace is fixed by \emph{all} unitaries, the defining representation is irreducible, and it turns out that $\mu \cong \nu_{\lambda}$, for $\lambda = (1)$.
\end{example}

\begin{remark}
    Since $\nu_\lambda$ is a polynomial representation, $\nu_\lambda(M)$ is defined for \emph{any} matrix $M$. 
\end{remark}

For the most part, we will not need any specific knowledge about any of these representations. We will, however, use the following fact:

\begin{theorem}[Littlewood-Richardson rule, \protect{\cite[Corollary 8.3.2(c)]{Ful97}}]\label{prelim_littlewood_richardson_rule}
Let $\lambda$ and $\mu$ be partitions of length at most $d$. Then $\nu_\lambda \otimes \nu_\mu$ defines a polynomial representation of $U(d)$, and decomposes as
\begin{equation*}\nu_\lambda \otimes \nu_\mu \cong \bigoplus_{\tau \, : \, \ell(\tau) \leq d} c^{\tau}_{\lambda \mu} \cdot \nu_{\tau},\end{equation*}
where the coefficients $\{c^\tau_{\lambda \mu}\}_\tau$ are nonnegative integers known as the \emph{Littlewood-Richardson coefficients}. We have $c^{\tau}_{\lambda \mu} = 0$ unless both of the following conditions are met:
\begin{itemize}
    \item $|\lambda| + |\mu| = |\tau|$.
    \item $\lambda \subseteq \tau$ and $\mu \subseteq \tau$. 
\end{itemize}
\end{theorem}

Much more can be said about the Littlewood-Richardson coefficients (e.g.\ see \cite[§5.1]{Ful97}). We will only need to know more in the following special case. 

\begin{corollary}[Pieri's rule]\label{cor:Pieri's_rule}
    In the special case where $\mu = (1)$, the Littlewood-Richardson coefficients are equal to $1$ if $\tau = \lambda+e_i$, for some $i$, and $0$ otherwise. That is, 
    \begin{equation*}
        \nu_\lambda \otimes \nu_{(1)} \cong \bigoplus_{i=1}^d \nu_{\lambda+e_i},
    \end{equation*}
    where the direct sum is understood to only iterate over the \emph{valid} partitions of the form $\lambda+e_i$. 
\end{corollary}

We will also need to know a little about the characters of the polynomial irreps of $U(d)$. We first need to define \emph{Schur polynomials}. 

\begin{definition}[Schur Polynomials] \label{def:schur_polys}
    Let $x_1, \dots, x_d$ be indeterminates. Given $\lambda \vdash n$, the \emph{Schur polynomial} $s_\lambda(x_1, \dots, x_d)$ is the degree-$n$ homogeneous polynomial given by $s_\lambda(x_1, \dots, x_d) = \sum_{T} x^T$, where the sum is over $T \in \SSYT(\lambda,d)$, and 
    \begin{equation*}
        x^T = \prod_{i=1}^d x_i^{w_T(i)}.
    \end{equation*}
    Here, $w_T(i)$ counts the number of boxes of $T$ filled with the number $i$. 
\end{definition}

\begin{example}
    If $\lambda = (2,1)$, and $d = 3$, then $\SSYT(\lambda, d)$ consists of the following tableaux. Underneath each tableau $T$ is the monomial $x^T$. 
    \begin{figure}[h!]
    \centering
    \[
    \begin{array}{cccccccc}
        \begin{array}{c}
            \begin{ytableau} 1 & 1 \\ 2 \end{ytableau} \\[2em]
            x_1^2x_2
        \end{array} &
        \begin{array}{c}
            \begin{ytableau} 1 & 2 \\ 2 \end{ytableau} \\[2em]
            x_1x_2^2
        \end{array} &
        \begin{array}{c}
            \begin{ytableau} 1 & 3 \\ 2 \end{ytableau} \\[2em]
            x_1x_2x_3
        \end{array} &
        \begin{array}{c}
            \begin{ytableau} 1 & 1 \\ 3 \end{ytableau} \\[2em]
            x_1^2x_3
        \end{array} &
        \begin{array}{c}
            \begin{ytableau} 1 & 2 \\ 3 \end{ytableau} \\[2em]
            x_1x_2x_3
        \end{array} &
        \begin{array}{c}
            \begin{ytableau} 1 & 3 \\ 3 \end{ytableau} \\[2em]
            x_1x_3^2
        \end{array} &
        \begin{array}{c}
            \begin{ytableau} 2 & 2 \\ 3 \end{ytableau} \\[2em]
            x_2^2x_3
        \end{array} &
        \begin{array}{c}
            \begin{ytableau} 2 & 3 \\ 3 \end{ytableau} \\[2em]
           x_2x_3^2
        \end{array}
    \end{array}
    \]
\end{figure}

    \noindent The corresponding Schur polynomial is the sum of all of these monomials:
    \begin{equation*}
        s_\lambda(x_1,x_2,x_3) = \Big(\sum_{i \neq j} x_i^2 x_j \Big) + 2 x_1 x_2 x_3.
    \end{equation*}
\end{example}

\begin{theorem}\label{slambda}
    Let $M$ be a diagonalizable matrix with eigenvalues $\alpha_1, \dots, \alpha_d$. Then $\chi_{\nu_\lambda}(M) = s_\lambda(\alpha_1, \dots, \alpha_d)$. 
\end{theorem}

\begin{remark}
    The character $\chi_{\nu_\lambda}(M)$ is defined for any $M$. Likewise, $s_\lambda$ can be continuously extended to all matrices, and then $s_\lambda(M) = \chi_{\nu_\lambda}(M)$. 
\end{remark}

\begin{remark}
     Since $s_\lambda$ is a degree-$|\lambda|$ homogeneous polynomial,  $s_\lambda(\alpha \cdot M ) = \alpha^{|\lambda|} \cdot s_\lambda(M)$ for all $\alpha \in \C$. 
\end{remark}

\begin{remark}
    We will often write $s_\lambda(1^r)$ as shorthand for $s_\lambda(1^r,0^{d-r})$, when $d$ is understood from context. For example, for a rank-$r$ projector $P \in \C^{d \times d}$, we have $s_\lambda(P) = s_\lambda(1^r)$. 
\end{remark}

\begin{remark} \label{rem:s_lambda_to_SSYTs}
    Note that $s_\lambda(1^d) = \tr( \nu_\lambda(I) ) = \tr\big( I_{\dim(V^d_\lambda)}\big) = \dim(V^d_\lambda)$. More generally, it is true that $s_\lambda(1^r0^{d-r})$ is the number of SSYTs of shape $\lambda$ using only labels in $[r]$, by \Cref{def:schur_polys}. This is because $x^T$ is $1$ if $T$ includes no boxes with a label in $\{r+1, \dots, d\}$, and $0$ otherwise. So, we have $s_\lambda(1^r) = |\SSYT(\lambda, r)|$. 
\end{remark}

We conclude this section with a straightforward calculation we will need for both our lower bound, and for showing that the PGM has optimal scaling. 

\begin{lemma}\label{lem:integral_over_unitary_register}
    Let $\lambda$ be a partition, and let $M \in \C^{d \times d}$. Then we have 
    \begin{equation*}
        \int_U \nu_{\lambda}( U M U^\dagger) \cdot \dU = \frac{s_\lambda(M)}{s_\lambda(1^d)} \cdot I_{\dim(V^d_\lambda)}
    \end{equation*}
\end{lemma}
\begin{proof}
    Let $T$ denote the integral. Then $T$ commutes with $\nu_\lambda(V)$ for any $V \in U(d)$:
    \begin{equation*}
        \nu_{\lambda}(V) \cdot T \coloneq \int_U \nu_{\lambda}(V U M U^\dagger) \cdot \dU = \int_{U'} \nu_{\lambda}(U' M U'^\dagger V) \cdot \mathrm{d}U' = T \cdot \nu_{\lambda}(V),
    \end{equation*}
    by defining $U' = VU$ and using the left-invariance of the Haar distribution. Thus, by Schur's lemma (\Cref{lem:Schur's_lemma}), $T$ is a multiple of the identity, i.e.\ $T = c \cdot I_{\dim(V^d_\lambda)}$. We can compute $c$ by taking traces. On the one hand:
    \begin{equation*}
        \tr(T) = \int_U \tr( \nu_\lambda(U) \nu_\lambda(M) \nu_\lambda(U)^\dagger ) \cdot \dU = \int_U s_\lambda(M) \cdot \dU = s_\lambda(M).
    \end{equation*}
    On the other hand, $\tr(T) = c \cdot \dim(V^d_\lambda)$. Thus $c = s_\lambda(M)/\dim(V^d_\lambda) = s_\lambda(M)/s_\lambda(1^d)$, and
    \begin{equation*}
        T = \int_U \nu_{\lambda}( U M U^\dagger) \cdot \dU = \frac{s_\lambda(M)}{s_\lambda(1^d)} \cdot I_{\dim(V^d_\lambda)}. \qedhere
    \end{equation*}
\end{proof}

From this lemma, the following corollary is immediate. 

\begin{corollary} \label{cor:irrep_Haar_random_projector}
    Let $\lambda$ be a partition, and $P \in \C^{d \times d}$ be a rank-$r$ projector. Then
    \begin{equation*}
        \E_{\bP \sim \mu_H} \big[ \nu_\tau(\bP) \big] = \frac{s_\lambda(1^r)}{s_\lambda(1^d)} \cdot I_{\dim(V^d_\lambda)}. 
    \end{equation*}
\end{corollary}

\subsubsection{Schur-Weyl duality}

In the problems we study, our algorithms are given as input $n$ copies of some unknown state, i.e.\ $\rho^{\otimes n}$. In this setting, there are two particularly important representations of the groups $S_n$ and $U(d)$ acting on the space $(\C^d)^{\otimes n}$. 

\begin{definition}\label{prelim_P_and_Q}
    The groups $S_n$ and $U(d)$ have the following natural representations acting on $(\C^d)^{\otimes n}$. First, $\P(\pi)$ acts by permuting the $n$ tensor factors according to $\pi$, i.e.\ for any standard basis element $\ket{i_1} \otimes \dots \otimes \ket{i_n}$, we have 
\begin{equation*}
    \P(\pi) \cdot \ket{i_1} \otimes \dots \otimes \ket{i_n} = \ket{i_{\pi^{-1}(1)}} \otimes \dots \otimes \ket{i_{\pi^{-1}(n)}}.
\end{equation*}
Next, $\Q(U)$ acts by operating on each tensor factor with $U$, i.e.\ as
\begin{equation*}
    \Q(U) \cdot \ket{i_1} \otimes \dots \otimes \ket{i_n} = U \ket{i_1} \otimes \dots \otimes U \ket{i_n}. 
\end{equation*}
\end{definition}

Since $\P(\pi)$ commutes with $\Q(U)$, for any $\pi \in S_n$ and $U \in U(d)$, the product of matrices $\P(\pi)\cdot \Q(U)$ defines a representation of the product of groups $S_n \times U(d)$. Schur-Weyl duality gives a nice decomposition of this representation into a sum over products of irreps of $S_n$ and $U(d)$ (which themselves are the irreps of $S_n \times U(d)$).

\begin{theorem}[Schur-Weyl duality] \label{prelim_schur_weyl_duality}
Consider the representations of $S_n$ and $U(d)$ described in \Cref{prelim_P_and_Q}. As a representation of $S_n \times U(d)$, we have the decomposition into irreps: 
\begin{equation*}
    (\C^d)^{\otimes n} \cong \bigoplus_{\substack{\lambda \vdash n \\ \ell(\lambda) \leq d}} \specht_\lambda \otimes V^d_\lambda. 
\end{equation*}
\end{theorem}

As a consequence of Schur-Weyl duality and \Cref{cor:unitary_isomorphism}, there then exists a fixed unitary, $\USW$, such that, for all $\pi \in S_n$ and $U \in U(d)$:
\begin{equation*}
    \USW \cdot \Big( \calP(\pi) \cdot \calQ(U) \Big) \cdot \USWdagger = \sum_{\substack{\lambda \vdash n \\ \ell(\lambda) \leq d}} \ketbra{\lambda} \otimes \kappa_\lambda(\pi) \otimes \nu_\lambda(U).
\end{equation*}
The map $\USW$ is called the \emph{Schur-Weyl transform}, or just the \emph{Schur transform}. To simplify notation, when we conjugate by a unitary which is clear from context, we will drop the unitaries and write a congruence. For example,
\begin{equation*}
    \calP(\pi) \cdot \calQ(U)  \cong \sum_{\substack{\lambda \vdash n \\ \ell(\lambda) \leq d}} \ketbra{\lambda} \otimes \kappa_\lambda(\pi) \otimes \nu_\lambda(U).
\end{equation*}

Since each $\nu_\lambda$ is a polynomial representation, we may extend $\Q$ to act on any matrix, so that the above equation holds if we replace $U$ to any matrix $M$, as remarked previously. Most usefully for us, we may apply Schur-Weyl duality to the state $\rho^{\otimes n} = \P(e) \cdot \Q(\rho)$ where $e$ is the identity permutation, obtaining the following fact.

% \begin{definition}[Schur basis]
%     The \emph{Schur basis} is the basis of $(\C^d)^{\otimes n}$ with basis elements $\{ \ket{\lambda} \otimes \ket{S} \otimes \ket{T} \}$. Here, $\lambda \vdash n$ and $\ell(\lambda) \leq d$, and, for fixed $\lambda$, we have $S \in \SYT(\lambda)$ and $T \in \SSYT(\lambda,d)$. 
% \end{definition}

\begin{corollary}\label{eq:schur_weyl_rho^n}
    There exists a fixed unitary change of basis $\USW$, and hence an allowed quantum mechanical transformation, that puts any input state $\rho^{\otimes n}$ into the following~form:
    \begin{align*}
        \rho^{\otimes n}  \cong \sum_{\substack{\lambda \vdash n \\ \ell(\lambda) \leq d}} \ketbra{\lambda} \otimes I_{\dim(\lambda)} \otimes \nu_\lambda(\rho). 
    \end{align*}
\end{corollary}

\subsubsection{Quantum learning algorithms from representation theory}

\begin{definition}[Weak Schur sampling]
Write $\Pi_\lambda$ for the projector such that
    \begin{equation*}
        \Pi_\lambda \cong \ketbra{\lambda} \otimes I_{\dim(\lambda)} \otimes I_{\dim(V^d_\lambda)}.
    \end{equation*}
    \emph{Weak Schur sampling} (WSS) refers to performing the projective measurement $\{ \Pi_\lambda \}_{\substack{\lambda \vdash n, \ell(\lambda) \leq d}}$, and induces a probability distribution on partitions $\lambda$, denoted $\mathrm{WSS}_n(\rho)$.
\end{definition}

Given a state $\rho \in \C^{d \times d}$ with spectrum $\alpha = (\alpha_1, \dots, \alpha_d)$, weak Schur sampling yields outcome $\blambda$ with probability
\begin{equation} \label{WSS_prob}
    \Pr_{\blambda \sim \mathrm{WSS}_n(\rho)} [ \blambda = \lambda] = \tr(\Pi_{\lambda} \cdot \rho^{\otimes n}) = \dim(\lambda) \cdot \tr( \nu_{\lambda}(\rho) ) = \dim(\lambda) \cdot s_{\lambda}(\alpha),
\end{equation}
using \Cref{eq:schur_weyl_rho^n}. Suppose the outcome $\lambda$ is obtained after weak Schur sampling from $\rho^{\otimes n}$. The resulting post-measurement state is 
\begin{equation}\label{WSS_state}
    \rho_\lambda \coloneq \frac{\Pi_\lambda \cdot \rho^{\otimes n} \cdot \Pi_\lambda}{\tr(\Pi_\lambda \cdot \rho^{\otimes n})} \cong \ketbra{\lambda} \otimes \frac{I_{\dim(\lambda)}}{\dim(\lambda)} \otimes \frac{\nu_\lambda(\rho)}{s_\lambda(\alpha)}.
\end{equation}

\begin{remark} \label{WSS_rank_r_comment}
    Recall that the boxes of an SSYT are strictly increasing as we move down a column. Therefore, any SSYT with more than $r$ rows necessarily has a box containing a number larger than $r$. So if $\ell(\lambda) > r$, then $s_\lambda(\alpha_1, \dots, \alpha_r, 0, \dots, 0) = 0$. This means that if we perform weak Schur sampling on a rank-$r$ state, we \emph{always} receive a Young diagram $\lambda$ with $\ell(\lambda) \leq r$. 
\end{remark}

For intuition: weak Schur sampling on $\rho^{\otimes n}$ returns a Young diagram whose row-lengths are proportional to a sorted list of the eigenvalues of $\rho$, in the asymptotic limit. That is, the empirical spectrum $\lambda/n \coloneqq (\lambda_1/n, \dots, \lambda_d/n)$ is close to $\alpha$ when $n$ is large \cite{ARS88, KW01, HM02, CM06, OW16, OW17a}.

We conclude this section by observing that quantum learning algorithms, promised inputs of the form $\rho^{\otimes n}$, are equivalent to algorithms which first perform weak Schur sampling.

\begin{lemma}\label{lem:WLOG_WSS}
Let $\rho \in \mathbb{C}^{d \times d}$ be an unknown quantum state. Any algorithm that takes as input $\rho^{\otimes n}$ is equivalent to another that begins by performing weak Schur sampling, and then, having received $\lambda \vdash n$, measures in $V^d_\lambda$. 
\end{lemma}
\begin{proof}
    Suppose an algorithm $\alg$ measures $\rho^{\otimes n}$ using a POVM $M = \{M_\sigma\}_\sigma$. By \Cref{eq:schur_weyl_rho^n}, $\rho^{\otimes n}$ is always a mixture of states with different values of $\lambda$, i.e.\ 
    \begin{equation*}
        \rho^{\otimes n} = \sum_{\substack{\lambda \vdash n \\ \ell(\lambda) \leq d}} \Pi_\lambda \cdot \rho^{\otimes n} \cdot  \Pi_\lambda = \sum_{\substack{\lambda \vdash n \\ \ell(\lambda) \leq d}} \Pr_{\blambda \sim \mathrm{WSS}_n(\rho)} [ \blambda = \lambda] \cdot \rho_\lambda,
    \end{equation*}
    where $\rho_\lambda$ is given by \Cref{WSS_state}. Therefore, the measurement outcome statistics are unaffected if we first perform weak Schur sampling to obtain a $\rho_\lambda$ with probability $\Pr[\lambda]$, and then perform $M$ on $\rho_\lambda$. 

    Note that, having obtained $\lambda$, the first two registers of the state $\rho_\lambda$ contain no quantum information, and may be regarded as ancilla registers prepared in fixed states. Therefore, measuring $\rho_\lambda$ with a POVM $M$ is equivalent to discarding these first two registers, and performing a measurement $M^{(\lambda)}$ which acts only on $V^d_\lambda$. Specifically, we have
    \begin{align*}
        \tr( M_\sigma \cdot \rho_\lambda ) & = \tr( \USW M_\sigma \USWdagger \cdot \USW \rho_\lambda \USWdagger) \\
        & = \sum_{\substack{ \mu \vdash n \\ \ell(\mu) \leq d}} \sum_{S \in \SYT(\mu)} \sum_{T \in \SSYT(\mu,d)} \bra{\mu, S, T} \Big( \USW M_\sigma \USWdagger \cdot \ketbra{\lambda} \otimes \frac{I_{\dim(\lambda)}}{\dim(\lambda)} \otimes \frac{\nu_\lambda(\rho)}{s_\lambda(\alpha)} \Big) \ket{\mu, S, T}\\
        & = \sum_{T \in \SSYT(\lambda,d)} \bra{T}  \Big(\Big( \frac{1}{\dim(\lambda)} \sum_{S \in \SYT(\lambda)} \bra{\lambda, S} \USW M_\sigma \USWdagger \ket{\lambda, S} \Big) \cdot \frac{\nu_\lambda(\rho)}{s_\lambda(\alpha)}\Big) \ket{T}.
    \end{align*}
    Thus, if we define $M^{(\lambda)}_\sigma$, an operator on $V^d_\lambda$, by 
    \begin{equation*}
        M^{(\lambda)}_{\sigma} \coloneq \frac{1}{\dim(\lambda)} \sum_{S \in \SYT(\lambda)} \bra{\lambda, S} \USW M_\sigma \USWdagger \ket{\lambda,S},
    \end{equation*}
    we have $\tr(M_\sigma \cdot \rho_\lambda) = \tr(M^{(\lambda)}_\sigma \cdot \frac{\nu_\lambda(\rho)}{s_\lambda(\alpha)} )$. Thus, after obtaining $\lambda$ from WSS, measuring with $M$ is equivalent to the measurement $M^{(\lambda)}$ on $V^d_\lambda$. 
    \end{proof}

\section{Lower bounds on learning in Bures distance}

In this section, we prove that $n = \Omega(rd/\epsilon^2)$ copies are necessary for rank-$r$ projector tomography in Bures distance. Our approach is to bound higher moments of the affinity between the true and estimated states. By studying a suitably chosen moment, and by relating affinity to Bures distance, we are able to rule out the possibility of algorithms that use too few samples.

\subsection{Warm up: the pure state case}
\label{sec:pure_states}
We begin by considering the special case of learning pure states, i.e.\ the $r=1$ case. In this case, we can use symmetric subspace techniques, and study the fidelity directly, rather than via affinity.

\begin{proposition}[A lower bound on learning pure states in Bures distance]\label{pure_state_lower_bound} Any pure state tomography algorithm learning to Bures distance $\epsilon > 0$ requires at least $n = \Omega( d / \epsilon^2 )$ samples, for $d \geq 2$ and $\epsilon \leq 1/\sqrt{48}$.
\end{proposition}

Note that this also proves a lower bound on learning to trace distance $\epsilon$ of $n = \Omega(d/\epsilon^2)$, for $d \geq 2$ and $\epsilon \leq 1/\sqrt{96}$. This is because for pure states, we have 
    \begin{equation*}\frac{1}{\sqrt{2}} \DBur \leq \Dtr \leq \DBur,\end{equation*}
using the pure state formulas for trace distance, fidelity and Bures distance, given in \Cref{def:td,def:fid,def:Bur}. 

We will use the following well-known bound, which appears in \cite[Section 2.1]{Har13}. 

\begin{lemma}[An upper bound on the $k$-th moment of fidelity]\label{pure_states_upper_bound_kth_moment}
    Suppose $\alg$ is an algorithm for pure state tomography that outputs pure states. Let $k$ be an arbitrary nonnegative integer. Then \begin{equation*}\E_{\ket{\bu} \sim \mu_H} \Big[\E_{\ket{\widehat{\bu}} \sim \alg(\bu)} \big[\left| \braket{\widehat{\bu}}{\bu} \right|^{2k}\big]\Big] \leq  \frac{\binom{d+n-1}{n}}{\binom{d+n+k-1}{n+k}}\end{equation*}
\end{lemma}

\begin{proof} Let $M = \{M_{\ket{\widehat{u}}}\}_{\ket{\widehat{u}}}$ be the measurement used by $\alg$, with POVM elements indexed by the corresponding output.  Since the input is necessarily in the symmetric subspace, we can assume $M$ is a POVM on $\Sym$. Given input $\ket{u}$, the $k$-th moment of the squared fidelity is
\begin{align*}
     \E_{\ket{\widehat{\bu}} \sim \alg(u)} \big[\left| \braket{\widehat{\bu}}{u} \right|^{2k}\big] & = \sum_{ \widehat{u} } \tr(M_{\ket{\widehat{u}}} \cdot \ketbra{u}^{\otimes n}) \cdot  \left| \braket{\widehat{u}}{u} \right|^{2k} \\
     & = \sum_{ \widehat{u} } \tr(M_{\ket{\widehat{u}}} \otimes \ketbra{\widehat{u}}^{\otimes k} \cdot \ketbra{u}^{\otimes (n+k)}) \\
     & = \tr\Big( \Big(\sum_{\widehat{u}} M_{\ket{\widehat{u}}} \otimes \ketbra{\widehat{u}}^{\otimes k} \Big) \cdot \ketbra{u}^{\otimes (n+k)} \Big). 
\end{align*}
Here, the sum over $\ket{\widehat{u}}$ is formal: if the POVM has finitely many elements, then this is a sum in the usual sense; if the POVM is continuous, it should be replaced with the appropriate integral. On a Haar random input, we then have
\begin{align}
    \E_{\ket{\bu} \sim \mu_H} \Big[\E_{\ket{\widehat{\bu}} \sim \alg(\bu)} \big[\left| \braket{\widehat{\bu}}{\bu} \right|^{2k}\big]\Big]  & = \tr\Big( \Big(\sum_{\widehat{u}} M_{\ket{\widehat{u}}} \otimes \ketbra{\widehat{u}}^{\otimes k} \Big) \cdot \E_{\ket{\bu}\sim\mu_H} \Big[\ketbra{\bu}^{\otimes (n+k)} \Big] \Big) \nonumber \\
    & = \frac{1}{\binom{d+n+k-1}{n+k}} \tr\Big( \Big(\sum_{\widehat{u}} M_{\ket{\widehat{u}}} \otimes \ketbra{\widehat{u}}^{\otimes k} \Big) \cdot \projSym{n+k} \Big),  \label{eq:exp_overlap}
\end{align}
using \Cref{proj_sym_subspace}. We now bound $\Pi^{(n+k)}_{\mathrm{sym}}$ in the PSD order as $\Pi^{(n+k)}_{\mathrm{sym}} \preceq I_{n+k}$ to obtain
\begin{align*}
    \tr\Big( \Big(\sum_{\widehat{u}} M_{\ket{\widehat{u}}} \otimes \ketbra{\widehat{u}}^{\otimes k} \Big) \cdot \projSym{n+k} \Big) & \leq \tr\Big( \Big(\sum_{\widehat{u}} M_{\ket{\widehat{u}}} \otimes \ketbra{\widehat{u}}^{\otimes k} \Big) \cdot I_{n+k} \Big) \\
    & = \tr\Big( \sum_{\widehat{u}} M_{\ket{\widehat{u}}}  \Big)  = \tr \Big( \projSym{n} \Big)  = \binom{d+n-1}{n}.
    \end{align*}
    In the last step we have used \Cref{dimension_sym_subspace}. Substituting back into \Cref{eq:exp_overlap} finishes the proof.
\end{proof}

    We now prove the main result of the subsection.

    \begin{proof}[Proof of \cref{pure_state_lower_bound}]
    
    From $\alg$, we can construct an algorithm $\alg'$ which uses no more samples and always outputs pure states, while learning to Bures distance $2\epsilon$ with high probability, by \cref{WLOG_projector_output}. Then, from \cref{pure_states_upper_bound_kth_moment}, we have
    \begin{align}
        \E_{\ket{\bu} \sim \mu_H} \Big[\E_{\ket*{\widehat{\bu}} \sim \alg'(\bu)} \big[\left| \braket*{\widehat{\bu}}{\bu} \right|^{2k}\big]\Big]  \leq  \frac{\binom{d+n-1}{n}}{\binom{d+n+k-1}{n+k}} & = \frac{(n+1)\dots(n+k)}{(n+d)\dots(n+d+k-1)} \nonumber \\ 
        & \leq \Big(\frac{n+k}{d+n+k-1}\Big)^{k} = \Big( 1 - \frac{d-1}{d+n+k-1}\Big)^{k}. \label{ps_fidelity_upper_bound}
    \end{align}
    % Since the \emph{worst-case} is no better than the \emph{average-case}, for each $k$ there exists a fixed state $\ket{u_k}$ such that
    % \begin{equation}\label{ps_fidelity_upper_bound}
    %     \E_{\ket*{\widehat{\bu}} \sim \alg(\ket{u_k})} \big[ \left|\braket*{\widehat{\bu}}{u_k} \right|^{2k} \big] \leq \Big( \frac{n+k}{d+n+k-1}\Big)^k = \Big(1 - \frac{d-1}{d+n+k-1}\Big)^k.
    % \end{equation}
    However, for any input $\ket{u}^{\otimes n}$, $\alg'$ succeeds in learning a state $\ket*{\widehat{\bu}}$ such that $\DBur(\widehat{\bu}, u) \leq 2\epsilon$ with probability $99\%$. In this case, we have
    \begin{equation*}
        \big|\braket{\widehat{\bu}}{u} \big| = \Fid\big(\widehat{\bu}, u\big) = 1 - \frac{1}{2} \DBur\big(\widehat{\bu}, u\big)^2 \geq 1 - 2\epsilon^2. 
    \end{equation*}
    Otherwise, we always have $\big|\braket*{\widehat{\bu}}{u} \big| \geq 0$ at least. Thus,
    \begin{equation}\label{ps_fidelity_lower_bound}
        \E_{\ket{\bu} \sim \mu_H} \Big[\E_{\ket*{\widehat{\bu}} \sim \alg'(\bu)} \big[\big| \braket*{\widehat{\bu}}{\bu} \big|^{2k}\big]\Big] \geq 0.99\cdot \Big(1-2\epsilon^2\Big)^{2k} + 0.01 \cdot 0 = 0.99\cdot\Big(1-2\epsilon^2\Big)^{2k}.
    \end{equation}
    Combining \cref{ps_fidelity_upper_bound} and \cref{ps_fidelity_lower_bound} gives
    \begin{equation*}
        0.99\cdot \Big(1-2\epsilon^2\Big)^{2k} \leq \E_{\ket{\bu} \sim \mu_H} \Big[\E_{\ket*{\widehat{\bu}} \sim \alg'(\bu)} \big[\big| \braket*{\widehat{\bu}}{\bu} \big|^{2k}\big]\Big] \leq \Big(1 - \frac{d-1}{d+n+k-1}\Big)^k.
    \end{equation*}
    We now apply the inequalities $1 + xy \leq (1+x)^y \leq e^{xy}$, which hold for $x \geq -1$, to loosen both bounds, obtaining
    \begin{equation*}\label{ps_k_inequality}
        0.99 \cdot \left(1 - 4k \epsilon^2\right) \leq \exp \Big(-\frac{k (d-1)}{d+n+k-1} \Big) 
    \end{equation*}
We now choose a particular $k$. If we take $k = \left\lfloor \frac{1}{16\epsilon^2} \right\rfloor$, then we have
    \begin{equation*}
        e^{-1/2} \leq 0.99 \cdot \left(1 - \frac{1}{4} \right) \leq 0.99 \cdot \left(1 - 4k\epsilon^2\right) \leq \exp \left( - \frac{k(d-1)}{d+n+k-1}\right).
    \end{equation*}
    Taking logarithms then gets us
    \begin{equation*}
        \frac{k(d-1)}{d+n+k-1} < \frac{1}{2},
    \end{equation*}
    which implies 
    \begin{equation*}
        n > 2k(d-1) - d - k +1 = (2k-1)(d-1) - k.
    \end{equation*}
    For $d \geq 2$, we have $d-1 \geq d/2$ and $k \leq k(d-1)$, so that
    \begin{equation*}
        n > (2k-1)(d-1) - k(d-1) = (k-1)(d-1) \geq \frac{1}{2} (k-1) d.
    \end{equation*}
    % which implies
    % \begin{equation*}
    %     n > (2k-1)(d-1) - k.
    % \end{equation*}
    % When $d \geq 2$, we have $k \leq (d-1) k$ and $d-1 \geq \frac{1}{2}d$, so that
    % \begin{equation*}n > (2k-1)(d-1) - k \geq (k-1)(d-1) \geq \frac{1}{2} (k-1)d.\end{equation*} 
    When $16\epsilon^2 \leq \frac{1}{3}$, we may also apply the inequality $\left \lfloor x \right \rfloor - 1 \geq x/2$, which holds for all $x \geq 3$. With $x = \frac{1}{16\epsilon^2}$, we get $k-1 \geq \frac{1}{32\epsilon^2}$. We conclude that
    \begin{equation*}n > \frac{d}{64\epsilon^2}. \qedhere \end{equation*}
\end{proof}

\subsection{The general case}
\newcommand{\muGr}{\mu_{\mathrm{Gr}}}
In this subsection we generalize the argument from the previous subsection to general $r$, to prove a lower bound on the sample complexity of learning rank-$r$ quantum states in Bures distance.

\begin{proposition}[A lower bound on learning rank-$r$ projector states in Bures distance]\label{projector_lower_bound}

Any rank-$r$ projector tomography algorithm learning to Bures distance $\epsilon > 0$ requires at least $n = \Omega(rd/\epsilon^2)$ samples, for $d \geq 2$, $r \leq d/2$, and $\epsilon \leq 1/80$. 
\end{proposition}

We note that the further restrictions in the proposition statement of the form $d \geq d_0$, $r \leq r_0$ and $\epsilon \leq \epsilon_0$ are necessary. This is because taking any of $d = 1$, $r = d$ or $\epsilon = \sqrt{2}$ renders the problem trivial. In the first two cases, there is a unique rank-$r$ projector to return: $I/d$. In the last case, returning any state suffices, since $\DBur(\rho,\sigma) \leq \sqrt{2}$ for all $\rho$ and $\sigma$. Therefore, no lower bounds can be proven without such restrictions. 

\begin{proof}[Proof of \Cref{projector_lower_bound}]

Fix any such algorithm $\calA$, and suppose it uses $n$ samples. By \Cref{WLOG_projector_output}, we can construct a new algorithm $\calA'$ which: uses no more samples, outputs rank-$r$ projector states, and learns to Bures distance $2\epsilon$ with high probability. By \Cref{lem:WLOG_WSS}, we can also assume $\mathcal{A}'$ begins by performing weak Schur sampling and then proceeding conditioned on the  outcome $\blambda \vdash n$. We describe its subsequent action by an algorithm  $\mathcal{A}'^{(\blambda)}$, which measures in $V^d_{\blambda}$. This POVM is written $M^{(\blambda)} = \{M^{(\blambda)}_Q\}_Q$, with measurement operators indexed by the rank-$r$ orthogonal projector corresponding to the output, i.e.\ $Q/r$. 

Let $\rho = P/r$ be an input state. By Schur-Weyl duality (\Cref{eq:schur_weyl_rho^n}), we have
\begin{equation*}
\rho^{\otimes n} \cong \sum_{\substack{\lambda \vdash n \\ \ell(\lambda) \leq d}} \ketbra{\lambda} \otimes I_{\dim(\lambda)} \otimes \nu_\lambda(\rho) = \frac{1}{r^n}\sum_{\substack{\lambda \vdash n \\ \ell(\lambda) \leq d}} \ketbra{\lambda} \otimes I_{\dim(\lambda)} \otimes \nu_\lambda(P).
\end{equation*}
Upon weak Schur sampling, we obtain a Young diagram $\lambda$ with probability
\begin{equation*}
\Pr_{\blambda \sim \mathrm{WSS}_n(\rho)}[\blambda = \lambda] = \dim(\lambda) \cdot s_\lambda(\rho) = \frac{1}{r^n} \cdot \dim(\lambda) \cdot s_{\lambda}(P),\end{equation*}
as in \Cref{WSS_prob}, and the post-measurement state is
\begin{equation*}
    \rho_\lambda \cong \ketbra{\lambda} \otimes \frac{I_{\dim(\lambda)}}{\dim(\lambda)} \otimes \frac{\nu_\lambda(\rho)}{s_\lambda(\rho)} = \ketbra{\lambda} \otimes \frac{I_{\dim(\lambda)}}{\dim(\lambda)} \otimes \frac{\nu_\lambda(P)}{s_\lambda(P)}.
\end{equation*}
as in \Cref{WSS_state}.

Now suppose WSS has occurred, and the fixed outcome $\lambda \vdash n$ has been obtained. The algorithm now measures in $V^d_\lambda$ with $M^{(\lambda)}$. We will write $\rho|_\lambda$ for the state in the $V^d_\lambda$ register, i.e.\ $\rho|_\lambda = \nu_\lambda(P)/s_\lambda(P)$. The $k$-th moment of affinity between $\rho$ and the output $\widehat{\brho} = \widehat{\bP}/r$ is 
% \begin{align}
%     \E_{\widehat{\brho} \sim \alg'^{(\lambda)}(\rho|_\lambda)} \Big[ \Aff ( \rho, \widehat{\brho})^k \Big]  & = \int_{\widehat{P}} \tr ( M_{\widehat{P}}^{(\lambda)} \cdot \rho_\lambda ) \cdot \Aff( \rho, \widehat{\rho})^k \cdot \mathrm{d}\widehat{P} = \E_{\widehat{\bP} \sim \mu_H} \Big[ \tr( M_{\widehat{\brho}}^{(\lambda)} \cdot \rho_\lambda) \cdot \Aff ( \rho, \widehat{\brho})^k \Big].
% \end{align}
\begin{align}
    \E_{\widehat{\brho} \sim \alg'^{(\lambda)}(\rho|_\lambda)} \Big[ \Aff ( \rho, \widehat{\brho})^k \Big] = \sum_{\widehat{P}} \tr( M^{(\lambda)}_{\widehat{P}} \cdot \rho|_\lambda) \cdot \Aff( \rho, \widehat{\rho})^k. \label{aff_moment_1}
\end{align}
As in the pure state case, the sum is formal, representing e.g.\ an integral in the continuous case. We can rewrite the affinity:
\begin{equation*}
    \Aff(\rho, \widehat{\rho})^k  = \tr( \rho^{1/2} \cdot \widehat{\rho}^{1/2})^k = \frac{1}{r^k} \tr ( P^{1/2} \cdot {\smash{\widehat{P}}}^{1/2})^k= \frac{1}{r^k} \tr( P \cdot \widehat{P})^k = \frac{1}{r^k} \tr( P^{\otimes k} \cdot {\smash{\widehat{P}}}^{\otimes k}).\end{equation*}
We can evaluate this trace in the Schur basis. Since for any $k$-fold operator, 
\begin{equation*}
    Q^{\otimes k} \cong \sum_{\substack{\mu \vdash k \\ \ell(\mu) \leq d}} \ketbra{\mu} \otimes I_{\dim(\mu)} \otimes \nu_\mu(Q),
\end{equation*}
we have
\begin{equation*}
    \Aff(\rho, \widehat{\rho})^k = \frac{1}{r^k} \tr( P^{\otimes k} \cdot {\smash{\widehat{P}}}^{\otimes k}) = \frac{1}{r^k} \sum_{\substack{\mu \vdash k \\ \ell(\mu) \leq d}} \dim(\mu) \cdot \tr( \nu_\mu(P) \cdot \nu_\mu(\widehat{P})).
\end{equation*}
Substituting this expression for the $k$-th moment of affinity back into \Cref{aff_moment_1}, we have
\begin{align*}
    \E_{\widehat{\brho} \sim \alg'^{(\lambda)}(\rho|_\lambda)} \Big[ \Aff ( \rho, \widehat{\brho})^k \Big] & = \frac{1}{r^{k}} \sum_{\substack{\mu \vdash k \\ \ell(\mu) \leq d}} \dim(\mu) \cdot \sum_{\widehat{P}} \tr( M_{\widehat{P}}^{(\lambda)} \cdot \rho|_\lambda ) \cdot \tr( \nu_\mu(P) \cdot \nu_\mu(\widehat{P}) ) \\
    & = \frac{1}{r^{k}} \sum_{\substack{\mu \vdash k \\ \ell(\mu) \leq d}} \dim(\mu) \cdot \tr\Big( \Big( \sum_{\widehat{P}} M^{(\lambda)}_{\widehat{P}} \otimes \nu_\mu(\widehat{P})\Big) \cdot \Big(\rho|_\lambda \otimes \nu_\mu(P)\Big)\Big) \\
    & = \frac{1}{r^{k} s_\lambda(1^r)} \sum_{\substack{\mu \vdash k \\ \ell(\mu) \leq d}} \dim(\mu) \cdot \tr\Big( \Big( \sum_{\widehat{P}} M^{(\lambda)}_{\widehat{P}} \otimes \nu_\mu(\widehat{P})\Big)  \cdot \Big(\nu_\lambda(P) \otimes \nu_\mu(P)\Big)\Big).
    % \frac{r^k}{s_\lambda(\rho)} \sum_{\substack{\mu \vdash k \\ \ell(\mu) \leq d}} \dim(\mu)  \E_{\widehat{\bP} \sim \mu_H} \Big[ \tr( M_{\widehat{\brho}}^{(\lambda)} \cdot \nu_\lambda(\rho) ) \cdot \tr( \nu_\mu(\widehat{\brho}) \cdot \nu_\mu(\rho) )\Big] \\
    % & = \frac{r^k}{s_\lambda(\rho)} \sum_{\substack{\mu \vdash k \\ \ell(\mu) \leq d}} \dim(\mu)  \tr( \E_{\widehat{\bP} \sim \mu_H} \Big[ M^{(\lambda)}_{\widehat{\brho}} \otimes \nu_\mu(\widehat{\brho})\Big] \cdot \nu_\lambda(\rho) \otimes \nu_\mu(\rho) ).
\end{align*}
In the last step, we have used $\rho|_\lambda = \nu_\lambda(P)/s_\lambda(P) = \nu_\lambda(P)/s_\lambda(1^r)$. 
On a Haar random projector state input $\brho = \bP/r$ where $\bP \sim \mu_H$, we have
\begin{align}\label{aff_moment_20}
    & \E_{\bP \sim \mu_H} \Big[\E_{\widehat{\brho} \sim \alg'^{(\lambda)}(\brho_\lambda)} \Big[ \Aff ( \brho, \widehat{\brho})^k \Big]\Big]  \nonumber \\
    & = \frac{1}{r^{k} s_\lambda(1^r)} \sum_{\substack{\mu \vdash k \\ \ell(\mu) \leq d}} \dim(\mu) \cdot \tr\Big( \Big( \sum_{\widehat{P}} M^{(\lambda)}_{\widehat{P}} \otimes \nu_\mu(\widehat{P})\Big)  \cdot \E_{\bP \sim \mu_H}\Big[\nu_\lambda(\bP) \otimes \nu_\mu(\bP)\Big]\Big). 
\end{align}

We would now like to understand the expectation on the right-hand side of \Cref{aff_moment_20}. Using the Littlewood-Richardson rule (\Cref{prelim_littlewood_richardson_rule}), we have
\begin{equation*}
    \nu_\lambda(\bP) \otimes \nu_\mu(\bP) \cong \sum_{\substack{\tau \vdash n + k\\\ell(\tau) \leq d}} \ketbra{\tau} \otimes I_{c^{\tau}_{\lambda \mu}} \otimes \nu_\tau(\bP).
\end{equation*}
Here, the congruence indicates equality up to conjugation by a unitary change-of-basis implied by \Cref{prelim_littlewood_richardson_rule} and \Cref{cor:unitary_isomorphism}. Therefore,
\begin{align}
    \E_{\bP \sim \mu_H} \Big[ \nu_\lambda(\brho) \otimes \nu_\mu(\brho)\Big] & \cong \sum_{\substack{\tau \vdash n + k\\\ell(\tau) \leq d}} \ketbra{\tau} \otimes I_{c^{\tau}_{\lambda \mu}} \otimes \E_{\bP \sim \mu_H} \big[ \nu_\tau(\bP) \big] \nonumber \\
    & = \sum_{\substack{\tau \vdash n + k\\\ell(\tau) \leq d}} \frac{s_\tau(1^r)}{s_\tau(1^d)} \cdot \ketbra{\tau} \otimes I_{c^{\tau}_{\lambda \mu}} \otimes  I_{\dim(V^d_\tau)}, \label{aff_moment_3}
\end{align}
where in the last step we have used \Cref{cor:irrep_Haar_random_projector}. Now, by \Cref{rem:s_lambda_to_SSYTs} and the hook-content formula (\Cref{thm:hook_content_formula}), for any irrep $\sigma$ we have 
\begin{equation*}
    \frac{s_\sigma(1^r)}{s_\sigma(1^d)} = \frac{|\SSYT(\lambda,r)|}{|\SSYT(\lambda,d)|} = \left(\prod_{(i,j) \in \sigma} \frac{r + \content(i,j)}{\hook_\sigma(i,j)}\right) \cdot \left(\prod_{(i,j) \in \sigma} \frac{d + \content(i,j)}{\hook_\sigma(i,j)}\right)^{-1} = \prod_{(i,j) \in \sigma} \frac{r + \content(i,j)}{d + \content(i,j)}.
\end{equation*}
For any $\tau$ with $c^{\tau}_{\lambda \mu}$ nonzero, $\tau$ contains $\lambda$, and hence
\begin{equation}\label{eq:product}
   \Bigg(\frac{s_\tau(1^r)}{s_\tau(1^d)}\Bigg)\cdot \Bigg(\frac{s_\lambda(1^r)}{s_\lambda(1^d)}\Bigg)^{-1} = \Bigg(\prod_{(i,j) \in \tau} \frac{r + \content(i,j)}{d + \content(i,j)}\Bigg) \cdot \Bigg(\prod_{(i,j) \in \lambda} \frac{r + \content(i,j)}{d + \content(i,j)}\Bigg)^{-1}  = \prod_{(i,j) \in \tau \setminus \lambda} \frac{r + \content(i,j)}{d + \content(i,j)}.
\end{equation}
Moreover, we must have $|\tau \setminus \lambda| = k$ for $c^\tau_{\lambda \mu}$ to be nonzero. 

How large can this product be, if $\lambda \subseteq \tau$ and $|\tau \setminus \lambda| = k$? Firstly, in order to maximize an individual term in the product, we should choose $\content(i,j) = j-i$ as large as possible, since $r \leq d$. Next, to maximize the content of a new box, we should always put that box into the first row, since this allows for both $i$ to be minimal and $j$ to be maximal. Lastly, we can view $\tau$ as constructed by first adding some boxes to the first row, then some to the second row, and so on. Consider the last box inserted in this process. If it were not in the first row, we could increase the product by inserting it instead into the first row. Therefore, starting with $\lambda$, we maximize the product by inserting $k$ boxes into the first row. See \Cref{fig:maximize_product} for an intuitive picture. 

So, we can bound the product in \Cref{eq:product} by choosing $\tau = \tau^* \coloneq \lambda+k\cdot e_1$\footnote{Note that this bound is not necessarily tight for our application, since, for example, we have ignored the further constraint that $\mu \subseteq \tau$ for $c^{\tau}_{\lambda \mu} \neq 0$. For example, if $\lambda = (1)$, $\mu = (1,1)$, then $\tau^* = (3)$, but in this case $\mu \not\subseteq \tau$. The bound will suffice for our purposes however.}, and we have
\begin{equation*}
    \Bigg(\frac{s_\tau(1^r)}{s_\tau(1^d)}\Bigg) \cdot \Bigg(\frac{s_\lambda(1^r)}{s_\lambda(1^d)}\Bigg)^{-1} \leq \Bigg(\frac{s_{\tau^*}(1^r)}{s_{\tau^*}(1^d)}\Bigg) \cdot \Bigg(\frac{s_\lambda(1^r)}{s_\lambda(1^d)}\Bigg)^{-1} = \prod_{i=1}^k \frac{r + (\lambda_1 + i - 1)}{d + (\lambda_1 + i - 1)}.
\end{equation*}
This gives the following bound on the ratios appearing in \Cref{aff_moment_3}:
\begin{equation*}
    \frac{s_\tau(1^r)}{s_\tau(1^d)} \leq \frac{s_{\tau^*}(1^r)}{s_{\tau^*}(1^d)} =  \frac{s_\lambda(1^r)}{s_\lambda(1^d)} \cdot \prod_{i=1}^k \frac{r + \lambda_1 + i - 1}{d + \lambda_1 + i - 1}.
\end{equation*}
\begin{figure}[t]
    \centering
    % \begin{minipage}{0.25\textwidth}
    %     \centering
    %     \begin{ytableau}
    %           ~ & ~ & ~ & ~ \\
    %           ~ & ~ & ~ \\
    %           ~ 
    %     \end{ytableau}

    %     \vspace{0.7em} % adjust this spacing as needed
    %     $\lambda = (4,3,1)$ \\ \phantom{with possibilities for an additional cell}
    % \end{minipage}
    % \hspace{1.6em}
    \begin{minipage}{0.35\textwidth}
        \centering
        \begin{ytableau}
              ~ & ~ & ~ & ~ & *(gray!30) 4\\
              ~ & ~ & ~ & *(gray!30) 2\\
              ~ & *(gray!30)-1 
        \end{ytableau}

        \vspace{0.7em} % adjust this spacing as needed
        {$\lambda = (4,3,1)$ with possibilities for an additional cell, labeled by contents}
    \end{minipage}
    \hspace{1.6em}
    \begin{minipage}{0.35\textwidth}
        \centering
        \begin{ytableau}
              ~ & ~ & ~ & ~ & ~ & *(gray!30) 5\\
              ~ & ~ & ~ & *(gray!30) 2\\
              ~ & *(gray!30)-1 
        \end{ytableau}

        \vspace{0.7em}
        $\lambda+e_1 = (5,3,1)$ with possibilities for an additional cell, labeled by contents
    \end{minipage}

    \caption{For fixed $\lambda$, the product   $\prod_{(i,j) \in \tau \setminus \lambda} \frac{r + \content(i,j)}{d + \content(i,j)}$ is maximized by the choice $\tau = \lambda + k \cdot e_1$, subject to the constraints $\lambda \subseteq \tau$ and $|\tau \setminus \lambda| = k$. We illustrate the reasoning here with an example. Take $d=3$. Left: $\lambda = (4,3,1)$, together with additional, shaded boxes, which represent boxes we \emph{could} add to $\lambda$. The shaded boxes are labeled with their contents. To maximize the content of a new box, we should add it to the first row. Having done so, we obtain $\lambda+e_1$. Right: $\lambda+e_1 = (5,3,1)$, again with possibilities for the next box to-be-added shaded, and labeled by contents. Since content increases to the right, the maximum content of a new box will always be in the first row.} \label{fig:maximize_product}
\end{figure}
So, from this inequality and \Cref{aff_moment_3}, we can give the following bound in the PSD order:
\begin{align*}
\E_{\bP \sim \mu_H} \Big[ \nu_\lambda(\bP) \otimes \nu_\mu(\bP)\Big] & \cong \sum_{\substack{\tau \vdash n + k\\\ell(\tau) \leq d}} \frac{s_{\tau}(1^r)}{s_{\tau}(1^d)} \cdot \ketbra{\tau} \otimes I_{c^{\tau}_{\lambda \mu}} \otimes  I_{\dim(V^d_\tau)} \\
& \preceq \frac{s_{\tau^*}(1^r)}{s_{\tau^*}(1^d)} \sum_{\substack{\tau \vdash n + k\\\ell(\tau) \leq d}}   \ketbra{\tau} \otimes I_{c^{\tau}_{\lambda \mu}} \otimes  I_{\dim(V^d_\tau)}  \\
& \cong  \frac{s_{\tau^*}(1^r)}{s_{\tau^*}(1^d)} \cdot I_{\dim(V^d_\lambda)} \otimes I_{\dim(V^d_\mu)}, 
\end{align*}
so that 
\begin{equation*}
    \E_{\bP \sim \mu_H} \Big[ \nu_\lambda(\bP) \otimes \nu_\mu(\bP)\Big] \preceq \Bigg( \frac{s_\lambda(1^r)}{s_\lambda(1^d)} \cdot \prod_{i=1}^k \frac{r + \lambda_1 + i - 1}{d + \lambda_1 + i - 1} \Bigg) \cdot I_{\dim(V^d_\lambda)} \otimes I_{\dim(V^d_\mu)}.
\end{equation*}
Substituting this back into \Cref{aff_moment_20}:
\begin{equation}\label{aff_moment_21}
     \E_{\bP \sim \mu_H} \Big[\E_{\widehat{\brho} \sim \alg'^{(\lambda)}(\rho|_\lambda)} \Big[ \Aff ( \rho, \widehat{\brho})^k \Big]\Big]  \leq 
 \frac{1}{r^{k} s_\lambda(1^d)} \cdot \left( \prod_{i=1}^k \frac{r + \lambda_1 + i - 1}{d + \lambda_1 + i - 1} \right)\cdot \sum_{\substack{\mu \vdash k \\ \ell(\mu) \leq d}} \dim(\mu) \cdot \tr\Big( \sum_{\widehat{P}} M^{(\lambda)}_{\widehat{P}} \otimes \nu_\mu(\widehat{P})\Big).
\end{equation}
Now we use the fact that $\tr(\nu_\mu(\widehat{P})) = s_\mu(1^r)$ for any $\widehat{P}$, so that
\begin{equation*}
    \tr\Big( \sum_{\widehat{P}} M^{(\lambda)}_{\widehat{P}} \otimes \nu_\mu(\widehat{P})\Big) = s_\mu(1^r) \cdot \tr\Big(\sum_{\widehat{P}} M^{(\lambda)}_{\widehat{P}} \Big) = s_\mu(1^r) \cdot \dim(V^d_\lambda) = s_\mu(1^r) \cdot s_\lambda(1^d).
\end{equation*}
The second-last step holds since $M^{(\lambda)}$ is a POVM on $V^d_\lambda$. Therefore,
\begin{equation}
    \sum_{\substack{\mu \vdash k \\ \ell(\mu) \leq d}} \dim(\mu) \cdot \tr\Big( \sum_{\widehat{P}} M^{(\lambda)}_{\widehat{P}} \otimes \nu_\mu(\widehat{P})\Big)  = s_\lambda(1^d) \cdot\Big( \sum_{\substack{\mu \vdash k \\ \ell(\mu) \leq d}} \dim(\mu) \cdot s_\mu(1^r)\Big) = s_\lambda(1^d) \cdot r^k.\label{sum_over_mu} 
\end{equation}
% However, since $M^{(\lambda)}_{Q} \preceq I_{\dim(V^d_\lambda)}$ for any $Q$, we have
% \begin{align}
%      \sum_{\substack{\mu \vdash k \\ \ell(\mu) \leq d}} \dim(\mu) \cdot \tr \Big(M^{(\lambda)}_{\widehat{\bP}} \otimes  \nu_\mu(\widehat{\bP})\Big)  & \leq  \sum_{\substack{\mu \vdash k \\\ell(\mu) \leq d}} \dim(\mu) \cdot \tr \Big(I_{\dim(V^d_\lambda)} \otimes  \nu_\mu(\widehat{\bP})\Big) \nonumber \\
%     & =  \sum_{\substack{\mu \vdash k \\ \ell(\mu) \leq d}} \dim(\mu) \cdot s_\lambda(1^d) \cdot s_\mu(1^r) \nonumber \\
%     & = \Big( \sum_{\substack{\mu \vdash k \\ \ell(\mu) \leq d}} \dim(\mu) \cdot s_\mu(1^r)\Big)\cdot s_\lambda(1^d) \nonumber \\
%     & = r^k s_\lambda(1^d). 
% \end{align}
The last equality can be seen as follows. Weak Schur sampling on $k$ copies of a fixed projector state $\rho = Q/r \in \C^{d \times d}$, yields $\bmu \vdash k$ with probability $\dim(\bmu) \cdot s_{\bmu}(\sigma) = \dim(\bmu) \cdot s_{\bmu}(1^r)/r^k$ (by \Cref{WSS_prob}). Thus
\begin{equation*}
    1 = \sum_{\substack{\mu \vdash k \\ \ell(\mu) \leq d}} \Pr_{\bmu \sim \mathrm{WSS}_k(\brho)} [ \bmu = \mu] = \frac{1}{r^k} \sum_{\substack{\mu \vdash k \\ \ell(\mu) \leq d}} \dim(\mu) \cdot s_\mu(1^r).
\end{equation*}
Then, substituting \Cref{sum_over_mu} into \Cref{aff_moment_21}, we finally obtain the bound:
\begin{equation} \label{eq:bound_given_lambda}
    \E_{\bP \sim \mu_H} \Big[ \E_{\widehat{\brho} \sim \alg'^{(\lambda)}(\brho_\lambda)} \Big[ \Aff ( \brho, \widehat{\brho})^k \Big] \Big] \leq \prod_{i=1}^k \frac{r+\lambda_1+i-1}{d+\lambda_1 + i - 1} \leq \left( \frac{r+ \lambda_1 + k - 1}{d + \lambda_1 + k - 1}\right)^k = \left( 1 - \frac{d-r}{d+\lambda_1 + k - 1}\right)^k.
\end{equation}

This bound we have just derived applies when, upon weak Schur sampling, we obtain $\lambda$. Therefore, to get a bound on the $k$-th moment of affinity, we should average over all possible Young diagrams we can obtain from WSS. This gives:
\begin{align}
    \E_{\bP \sim \mu_H} \Big[ \E_{\widehat{\brho} \sim \alg'(\brho)} \Big[ \Aff ( \brho, \widehat{\brho})^k \Big] \Big] & = \E_{\bP \sim \mu_H} \Bigg[ \sum_{\substack{\lambda \vdash n \\ \ell(\lambda) \leq d }} \Pr_{\blambda \sim \mathrm{WSS}_n(\brho)} [ \blambda = \lambda] \cdot \E_{\widehat{\brho} \sim \mathcal{A}'^{(\lambda)}(\rho|_\lambda)} \Big[ \Aff(\rho, \widehat{\brho})^k\Big]  \Bigg] \nonumber \\
    & =  \sum_{\substack{\lambda \vdash n \\ \ell(\lambda) \leq d }} \Pr_{\blambda \sim \mathrm{WSS}_n(\rho)} [ \blambda = \lambda] \cdot \bigg(\E_{\bP \sim \mu_H} \Big[ \E_{\widehat{\brho} \sim \alg'^{(\lambda)}(\brho_\lambda)} \Big[ \Aff ( \brho, \widehat{\brho})^k \Big] \Big]\bigg) \nonumber \\
    & \leq  \sum_{\substack{\lambda \vdash n \\ \ell(\lambda) \leq d }} \Pr_{\blambda \sim \mathrm{WSS}_n(\rho)} [ \blambda = \lambda] \cdot\left( 1 - \frac{d-r}{d+\lambda_1 + k - 1}\right)^k, \label{mixed_state_expectation_bound}
    % & = \frac{1}{r^n} \sum_{\substack{\lambda \vdash n \\ \ell(\lambda) \leq d }} \dim(\lambda) \cdot s_\lambda(1^r) \cdot\left( 1 - \frac{d-r}{d+\lambda_1 + k - 1}\right)^k.
\end{align}
where $\rho$ is any fixed rank-$r$ projector state. In the second step, we have used the fact that WSS probabilities depend only on the spectrum of $\rho$, which is the same for any rank-$r$ projector state.

We now proceed by showing that for any $\brho$, with high probability, we have $\blambda_1 \leq C n/r$. 
To do so, we use two previously known results on WSS statistics in an off-the-shelf manner. Firstly, Theorem 5.2 of \cite{OW16} states:
\begin{equation*}
    \E_{\blambda \sim \mathrm{WSS}_n(\brho)} [\blambda_1]  \leq \frac{n}{r} + 2 \sqrt{n}.
\end{equation*}
Secondly, Proposition 4.8 of \cite{OW17a} proves the concentration bound:
\begin{equation*}
    \Pr_{\blambda \sim \mathrm{WSS}_n(\brho)} \Big[ \big| \blambda_1 - \E_{\blambda \sim \mathrm{WSS}_n(\brho)}[ \blambda_1]\big| \geq t \Big] \leq 2 \exp \left(- \frac{t^2}{8n} \right).
\end{equation*}
Combining these gives
\begin{equation}\label{aff_moment_8}
\Pr_{\blambda \sim \mathrm{WSS}_n(\brho)} \Big[  \blambda_1 \geq \frac{n}{r} + (C+2)\sqrt{n} \Big] \leq 2 \exp \left(- \frac{C^2}{8} \right).
\end{equation}
Choosing $C = 7$ makes this probability smaller than $1\%$. Assume for now that $n \geq 81 r^2$, so that $9\sqrt{n} \leq \frac{n}{r}$. Then from \Cref{aff_moment_8} we have
\begin{equation}\label{aff_moment_9}
\Pr_{\blambda \sim \mathrm{WSS}_n(\brho)} \Big[  \blambda_1 \geq \frac{2n}{r} \Big] \leq 0.01.
\end{equation}
So, with probability at least $99\%$, we have $\blambda_1 \leq 2n/r$, and in this case
\begin{equation*}
    1 - \frac{d-r}{d+\blambda_1 + k - 1} \leq 1 - \frac{d-r}{d + \frac{2n}{r} + k - 1}.
\end{equation*}
In the event $\blambda_1 > 2n/r$, occurring with only at most $1\%$ probability, we will use instead the trivial bound
\begin{equation*}
    1 - \frac{d-r}{d+\blambda_1 + k - 1} \leq 1. 
\end{equation*}
Therefore, we can bound the expectation in \Cref{mixed_state_expectation_bound} as
\begin{align}\label{aff_moment_10}
\E_{\bP \sim \mu_H} \Big[ \E_{\widehat{\brho} \sim \alg'(\brho)} \Big[ \Aff ( \brho, \widehat{\brho})^k \Big] \Big] \nonumber 
&  \leq 0.99 \cdot \left( 1 - \frac{d-r}{d +\frac{2n}{r} + k - 1}\right)^k + 0.01 \cdot 1 \\ & \leq \left( 1 - \frac{d-r}{d +\frac{2n}{r} + k - 1}\right)^k + 0.01. 
\end{align}

We now proceed similarly to the pure state case. Recall $\mathcal{A}'$ produces a rank-$r$ projector state $\widehat{\brho}$ such that $\DBur(\brho, \widehat{\brho}) \leq 2\epsilon$ with probability at least $99\%$. In this case, we have $\Fid(\brho, \widehat{\brho}) \geq 1 - 2 \epsilon^2$, and hence $\Aff(\brho, \widehat{\brho}) \geq 1 - 4\epsilon^2$, by \Cref{cor:aff-fid}. In the remaining case, occuring with probability at most $1\%$, we always at least have the bound $\Aff(\brho, \widehat{\brho}) \geq 0$. Hence,
\begin{equation}\label{aff_moment_11}
\E_{\bP \sim \mu_H} \Big[ \E_{\widehat{\brho} \sim \alg(\brho)} \Big[ \Aff ( \brho, \widehat{\brho})^k \Big] \Big] \geq 0.99 \cdot \left(1 - 4\epsilon^2\right)^k + 0.01 \cdot 0 = 0.99 \cdot \left(1 - 4\epsilon^2\right)^k. 
\end{equation}
Combining \Cref{aff_moment_10} and \Cref{aff_moment_11} gets us:
\begin{equation*}
    \left(1 - 4\epsilon^2\right)^k - 0.01 \leq \left( 1 - \frac{d-r}{d+ \frac{2n}{r}  + k - 1}\right)^k. 
\end{equation*}
We now apply the inequalities $1 + xy \leq (1+x)^y \leq e^{xy}$, valid for $x \geq -1$, to get
\begin{equation*}
    \left(1 - 4k\epsilon^2\right) -0.01 \leq \exp \left( - \frac{k(d-r)}{d + \frac{2n}{r} + k - 1}\right). 
\end{equation*}
If we now choose $k = \left \lfloor \frac{1}{16\epsilon^2} \right \rfloor$, then the left-hand side can be further lower-bounded as:
\begin{equation*}
e^{-1/2} < 0.75 - 0.01 = (1 - 4\epsilon^2/16\epsilon^2)-0.01 \leq \left( 1 - 4k\epsilon^2 \right) - 0.01.
\end{equation*}
Taking logarithms then gets us
\begin{equation*}
    \frac{k(d-r)}{d+ \frac{2n}{r} + k - 1} < \frac{1}{2},
\end{equation*}
which implies
\begin{equation*}
    \frac{2n}{r} > 2k(d-r) - d - k + 1.
\end{equation*}
For $r \leq d/2$, we then have
\begin{equation*}
    \frac{2n}{r} > kd - d - k + 1 = (k-1)(d-1).
\end{equation*}
If $d \geq 2$, then $d-1 \geq d/2$. Moreover, for $16\epsilon^2 \leq 1/3$, we may apply the inequality $\lfloor x \rfloor - 1 \geq x/2$ with $x = 1/16\epsilon^2$, since the inequality holds for all $x \geq 3$. This gives us $k-1 \geq 1/32\epsilon^2$. Substituting both of these, we arrive at the bound:
\begin{equation}\label{aff_moment_12}
    n > \frac{r}{2} \cdot \frac{1}{32\epsilon^2} \cdot \frac{d}{2} = \frac{r d}{128\epsilon^2}.
\end{equation}

To conclude, we circle back to our assumption that $n \geq 81r^2$. We have shown so far that $n \geq 81r^2$ implies $n \geq rd/128\epsilon^2$, which is at least $100r^2$ if we impose the further restriction that $\epsilon \leq 1/80$:
\begin{equation*}
    \frac{rd}{128\epsilon^2} \geq \frac{r^2}{64\epsilon^2} \geq 100r^2.
\end{equation*}
But this implies that no algorithm can succeed with $n < 81r^2$ either, since if $\mathcal{A}$ could learn rank-$r$ projectors for such $r$ and $\epsilon$, using $n < 81r^2$, then certainly $\mathcal{A}$ could also solve the problem using $n$ samples with $n \in (81r^2, 100r^2)$, simply by ignoring the extra copies. This establishes the lower bound in \Cref{aff_moment_12} without this extra assumption on $n$, and completes our proof. 
\end{proof}

\section{Bootstrapping from trace distance learning to Bures distance learning}

\newcommand{\proj}{\mathrm{proj}}
\newcommand{\Spann}{\mathrm{span}}
\newcommand{\WAS}{S_{\mathrm{Align}}}
\newcommand{\ProjWAS}{\Pi_{\mathrm{Align}}}

In this section, we prove our bootstrapping result. 

\begin{proposition} \label{prop_bootstrapping}Let $\mathcal{A}$ be an algorithm for rank-$r$ projector tomography that, when given $n$ samples of $\rho$, returns a rank-$r$ projector state $\widehat{\brho}$ such that $\Dtr(\widehat{\brho}, \rho) \leq \epsilon$ with probability at least $99\%$. Then there exists an algorithm $\mathcal{A}'$ for rank-$r$ projector tomography that takes $n' = 2n + O(r^2/\epsilon^2)$ samples of $\rho$, and returns a $\widehat{\brho}$ such that $\DBur(\widehat{\brho}, \rho) \leq O(\epsilon)$ with probability at least $95\%$, for $r > r_0$, and $\epsilon < \epsilon_0$, where $r_0$ and $\epsilon_0$ are constants. 
\end{proposition}

We now describe how the new algorithm $\mathcal{A}'$ is constructed from $\mathcal{A}$.

\begin{definition}[The bootstrapped algorithm]\label{def_bootstrapped_alg}
Let $\mathcal{A}$ be an algorithm for rank-$r$ projector tomography that, when given $n$ samples of $\rho = P/r$, returns a rank-$r$ projector state $\widehat{\brho}$ such that $\Dtr(\widehat{\brho}, \rho) \leq \epsilon$ with probability $99\%$. The \emph{bootstrapped algorithm} $\mathcal{A}'$ is defined as follows. 

On input $\rho^{\otimes n'}$, where $n' = 2n + O(r^2/\epsilon^2)$:
\begin{enumerate}
    \item Pick a random unitary $\bU$. Give $n$ copies of $\bU \rho \bU^{\dagger}$ to $\alg$ and let $\bQ/r$ be its output. Write $\widehat{\bP}_1 = \bU^{\dagger} \bQ \bU$.
    \item Repeat this process a second time to construct $\widehat{\bP}_2$.
    \item Let $\bR$ be the projector onto $\Spann\{\widehat{\bP}_1, \widehat{\bP}_2\}$.
    \item Take $O(r^2/\epsilon^2)$ copies of $\rho$ and measure each of them with $\{\bR, \overline{\bR}\}$. Discard the post-measurement states corresponding to the outcome $\overline{\bR}$.
    \item The remaining post-measurement states $\rho|_{\bR}$ live inside $\bR$, which is a subspace of dimension at most $2r$.
    Using the Bures distance tomography algorithm of \cite{PSW25}, compute an estimate $\widehat{\brho}$ of $\rho|_{\bR}$ with Bures distance error $\epsilon$ using only $O(r^2/\epsilon^2)$ copies of $\rho|_{\bR}$.
    \item Output $\widehat{\brho}$ as the estimate for $\rho$.
\end{enumerate}
\end{definition}

\subsection{Proof overview}

In this subsection, we describe the main steps of our proof of \Cref{prop_bootstrapping}. In subsequent subsections, we fill in the technical details formally. We will see that each step occurs with high probability, given previous steps. We will assume that all previous steps have succeeded during our proof, and then address the success probability at the very end. 

\paragraph{Step 1.} We start by showing that because $\calA$ learns in trace distance, each of the $\widehat{\bP}_i$ must have a large-rank subprojector\footnote{By a \emph{subprojector} of $\Pi$, we mean a projector onto a subspace of $\supp(\Pi)$.} which is approximately ``aligned'' with $P$. We formalize this notion with the following definition. In this definition, $\alpha$ is a sufficiently small constant we will specify at a later stage in the proof.

% In the first step, we identify large-rank subprojectors $\bA_i$ of $\widehat{\bP}_i$, and $\bB_i$ of $P$, such that $\bA_i$ and $\bB_i$ are approximately equal. 

% In the first step, we identify subprojectors\footnote{By a \emph{subprojector} of $\Pi$, we mean a projector onto a subspace of $\supp(\Pi)$.} $\bA_1$ of $\widehat{\bP}_1$ and $\bA_2$ of $\widehat{\bP}_2$ with two key properties: 
% \begin{enumerate}
%     \item[(i)] $\rank(\bA_i) \geq (1- \alpha) \cdot \rank(\widehat{\bP}_i) = (1-\alpha) \cdot r$, where $\alpha$ is a sufficiently small constant we will fix later.
%     \item[(ii)]$\supp(\bA_i)$ is a \emph{good} subspace with respect to $P$, in the following sense.
% \end{enumerate}

% \begin{definition}[Good subspaces]
%     Let $\Pi$ be a projector on a vector space $V$. Then a subspace $W \subseteq V$ is a \emph{good subspace with respect to $\Pi$} if for all $\ket{w} \in W$, we have $\bra{w} \Pi \ket{w} \geq 1 - \epsilon^2/\alpha^2$. 
% \end{definition}

% To construct the $\bA_i$, we make a further definition. 

%We start with the following definition. 

\begin{definition}[Well-aligned subspaces]
    Let $\Pi_1$ and $\Pi_2$ be rank-$r$ projectors. Then the \emph{well-aligned subspace of $\Pi_1$ with respect to $\Pi_2$}, denoted $\WAS(\Pi_1 \mid \Pi_2)$, is defined as follows. Take a Jordan decomposition of $\Pi_1 = \sum_{i} \ketbra{u_i}$ and $\Pi_2 = \sum_i \ketbra{v_i}$, where $\ket{u_i}$ and $\ket{v_i}$ are the Jordan vectors in the $i$-th block. Write $\omega_i = | \braket{u_i}{v_i} |$. Then 
    \begin{equation*}
        \WAS(\Pi_1 \mid \Pi_2) = \Spann \{ \ket{u_i} \, \vert \, \omega_i^2 \geq 1 - \epsilon^2/\alpha^2 \}.
    \end{equation*}
    We will denote the projector onto the well-aligned subspace as $\ProjWAS(\Pi_1 \mid\Pi_2)$. 
\end{definition}

We now define 
\begin{equation}
    \bA_i \coloneq \ProjWAS(\widehat{\bP}_i\mid P), \qquad \bB_i \coloneq \ProjWAS(P\mid\widehat{\bP}_i). 
\end{equation}
Intuitively speaking, $\bA_i$ projects onto a subspace of $\supp(\widehat{\bP}_i)$ whose vectors have high overlap with $\supp(P)$, and $\bB_i$ projects onto a subspace of $\supp(P)$ whose vectors have high overlap with $\supp(\widehat{\bP}_i)$. Informally, $\bA_i$ is an approximate copy of $\bB_i$, and importantly, $\bA_i$ sits inside $\widehat{\bP}_i$, while $\bB_i$ sits inside $P$. That is, $\bA_i$ is our large-rank subprojector approximately ``aligned'' with $P$, and in particular, it is approximately $\bB_i$. 

In the first step of the proof, we formalize the above claims, and prove that these projectors have rank at least $(1-\alpha)\cdot \rank(P) = (1-\alpha) \cdot r$. We also show that the distributions of $\bA_i$ and $\bB_i$ are invariant under conjugation by unitaries of the form $U_P \oplus U_{\overline{P}}$. 

\paragraph{Step 2.} Next, we show that the two projectors $\bB_1$ and $\bB_2$, with high probability, ``cover'' $P$ in a robust sense. Not only does $\bB_1 + \bB_2$ have full rank, but the orthogonal complement of $\bB_1$ is approximately contained in $\bB_2$, and vice versa. In particular, in this second step we formally show that, with high probability, $\bB_1$ and $\bB_2$ \emph{robustly cover} $P$ in the following sense.

\begin{definition}
    Let $\Pi$ be a projector, and let $\Pi_1$ and $\Pi_2$ be subprojectors of $\Pi$. Then $\Pi_1$ and $\Pi_2$ \emph{robustly cover} $\Pi$ if two conditions are satisfied:
    \begin{itemize}
        \item Firstly, $\rank(\Pi_1 + \Pi_2) = \rank(\Pi)$.
        \item Take any Jordan block decomposition of $\Pi_{1}$ and $\Pi_{2}$. In a $2 \times 2$ block $B$, let $\Pi_{i}|_{B} = \ketbra*{w_{i,B}}$. The second condition is: $|\braket{w_{1,B}}{w_{2,B}}|^2 \leq 0.1$ for all such blocks $B$. 
    \end{itemize}
\end{definition}

Roughly speaking, this means that there is a complete ``copy" of $P$ inside $\supp(\bB_1 + \bB_2)$.
In particular, the second bullet says that the Jordan vectors of $\Pi_1$ and $\Pi_2$ form an almost orthogonal basis for $P$.
Intuitively, $\bB_1$ and $\bB_2$ should satisfy this definition because $\supp(\bB_1)$ and $\supp(\bB_2)$ are random subspaces of high rank, and therefore ``cover'' the entire space they sit in.

\paragraph{Step 3.} Now we use the fact that if $\bB_1$ and $\bB_2$ robustly cover $P$ to construct a basis of $P$ that we can ``lift'' to a set of $r$ nearby vectors contained in $\supp(\bA_1 + \bA_2)$. Informally, because there is a complete ``copy" of $P$ inside $\supp(\bB_1 + \bB_2)$, there is also an approximate ``copy'' of $P$ inside $\supp(\bA_1 + \bA_2)$. 
Looking ahead, our ultimate goal is to show that $\rho = P/r$ is roughly contained in the projector $\bR$ onto the subspace $\Spann\{\widehat{\bP}_1, \widehat{\bP}_2\}$, and showing that $P$ is roughly contained inside $\supp(\bA_1 + \bA_2)$ would suffice to show this, as $\supp(\bA_1 + \bA_2) \subseteq \Spann\{\widehat{\bP}_1, \widehat{\bP}_2\}$.

Consider a Jordan decomposition of $\bB_1$ and $\bB_2$, where we regard these as projectors in the space $\supp(P)$. Since $\bB_1 + \bB_2$ has full rank, any $1 \times 1$ block in the decomposition is fixed by one of $\bB_1$ or $\bB_2$. Thus, there are three types of blocks: $\calB_1$, $1 \times 1$ blocks which are fixed by $\bB_1$; $\calB_2$, $1 \times 1$ blocks which are not fixed by $\bB_1$, but are fixed by $\bB_2$; and $\calB_{12}$, the $2\times 2$ blocks. 

For $B \in \calB_1$, $\bB_1|_{B} = \ketbra{\bu_B}$ for some vector $\ket{\bu_B}$. Similarly, for $B \in \calB_2$, $\bB_2|_{B} = \ketbra{\bv_B}$ for some vector $\ket{\bv_B}$. Now consider $B \in \calB_{12}$. In this block, $\bB_1|_{B} = \ketbra{\bw_{1,B}}$ and $\bB_2|_{B} = \ketbra{\bw_{2,B}}$, with $| \braket*{\bw_{1,B}}{\bw_{2,B}}| \leq 0.1$. The vectors $\ket{\bw_{1,B}}$ and $\ket{\bw_{2,B}}$ are linearly independent, and span $B$, but are not necessarily orthonormal. However, if we define as the vector
\begin{equation*}
    \ket*{\bw_{1,B}^\perp} \propto \ket{\bw_{2,B}} - \braket*{\bw_{1,B}}{\bw_{2,B}} \cdot \ket{\bw_{1,B}},
\end{equation*}
then $\{ \ket{\bw_{1,B}}, \ket*{\bw_{1,B}^\perp}\}$ is an orthonormal basis for $B$. Thus 
\begin{equation*}
    \calO_P \coloneq \big\{ \ket{\bu_B} \big\}_{B \in \calB_1} \cup \big\{  \ket{\bv_B} \big\}_{B \in \calB_2} \cup \big\{ \ket{\bw_{1,B}}, \ket*{\bw_{1,B}^\perp} \big\}_{B \in \calB_{12}},
\end{equation*}
is an orthonormal basis for $\supp(P)$.

We now describe how to lift these basis vectors to a new set of vectors in $\supp(\bA_1 + \bA_2)$. The idea is, roughly speaking, to lift the vectors $\{\ket{\bu_B}\}$ and $\{\ket{\bw_{1,B}}\}$ to preimages under $P$ in $\bA_1$, and then normalize, obtaining $\{\ket{\widetilde{\bu}_B}\}$ and $\{\ket{\widetilde{\bw}_{1,B}}\}$ respectively. Likewise the vectors $\{\ket{\bv_B}\}$ and $\{\ket{\bw_{2,B}}\}$ are lifted to preimages in $\bA_2$ and normalized, giving $\{\ket{\widetilde{\bv}_B}\}$ and $\{\ket{\widetilde{\bw}_{2,B}}\}$ respectively. We now explain why this can be done. Take any $\ket{\bu_B} \in \supp(\bB_1)$ as an example. Since $\bA_1$ and $\bB_1$ are defined via the Jordan vectors of $\widehat{\bP}_1$ and $P$ respectively which are closely aligned, there is a vector $\ket{\widetilde{\bu}_B} \in \supp(\bA_1)$ such that $P \ket{\widetilde{\bu}_B} \propto \ket{\bu_B}$. In particular, if we write the sufficiently-aligned Jordan vectors of $P$ and $\widehat{\bP}_1$  as $\{ \ket{\bu_i}\}$ and $\{ \ket*{\widetilde{\bu}_i}\}$ respectively, and if
\begin{equation*}
    \ket{\bu_B} = \sum_{i} \braket{\bu_i}{\bu_B} \cdot \ket{\bu_i},
\end{equation*}
then the unnormalized vector
\begin{equation*}
    \sum_{i} \braket{\bu_i}{\bu_B} \cdot \frac{1}{\braket*{\bu_i}{\widetilde{\bu}_i}} \cdot \ket*{\widetilde{\bu}_i}
\end{equation*}
is mapped by $P$ to $\ket{\bu_B}$, since $P \ket{\widetilde{\bu}_i} = \ketbra{\bu_i} \cdot \ket*{\widetilde{\bu}_i} = \braket*{\bu_i}{\widetilde{\bu}_i} \cdot \ket{\bu_i}$. Normalizing gives us $\ket{\widetilde{\bu}_B}$ such that $P\ket{\widetilde{\bu}_B} \propto \ket{\bu_B}$. The construction is analogous for the aforementioned cases.
% We now describe how to lift these vectors to a set of vectors in $\supp(\bA_1 + \bA_2)$. For the vectors $\{ \ket{\bu_B} \}_{B \in \calB_1}$, we associate $\ket{\bu_B} \in \supp(\bB_1)$ with the orthonormal vector defined by $\ket{\widetilde{\bu}_B} \propto \widehat{\bP}_1 \ket{\bu_B}$. This is the orthonormal vector in $\hat{\bP}_1$ most aligned with $\ket{\bu_B}$. For the vectors in $\{  \ket{\bv_B} \}_{B \in \calB_2}$ and $\{ \ket{\bw_{i,B}}_{B \in \calB_{12}}$, we do the analogous construction: $\ket{\widetilde{\bv}_B} \propto \widehat{\bP}_2 \ket{\bv_B}$ and $\ket{\widetilde{\bw}_{i,B}} \propto \widehat{\bP}_i \ket{\bw_{i,B}}$. 

Finally, for the vectors in $\{  \ket*{\bw_{1,B}^\perp} \}_{B \in \calB_{12}}$, we lift each of its constituent vectors separately, i.e.\ as
\begin{equation*}
    \ket*{\widetilde{\bw}_{1,B}^\perp} \propto \ket{\widetilde{\bw}_{2,B}} - \braket*{\bw_{1,B}}{\bw_{2,B}} \cdot \ket{\widetilde{\bw}_{1,B}}.
\end{equation*}

\begin{definition}[The lift of $\calO_P$]
    The \emph{lift} of the basis $\calO_P$ is the set of vectors
    \begin{equation*}
    \widetilde{\calO}_P \coloneq \big\{ \ket{\widetilde{\bu}_B} \big\}_{B \in \calB_1} \cup \big\{  \ket{\widetilde{\bv}_B} \big\}_{B \in \calB_2} \cup \big\{ \ket{\widetilde{\bw}_{1,B}}, \ket*{\widetilde{\bw}_{1,B}^\perp} \big\}_{B \in \calB_{12}},
\end{equation*}
\end{definition}

We show that each of the $r$ vectors in the lift has overlap $1 - O(\epsilon^2)$ with $P$. Since these vectors sit within $\bR$, we are then able to use this to show that $\tr(\bR \cdot P/r) \geq 1 - O(\epsilon^2)$. 

\paragraph{Step 4.} Lastly, we use the inequality $\tr( \bR \cdot \rho) \geq 1- O(\epsilon^2)$ to conclude the main result. There are two main implications of the inequality:
\begin{itemize}
    \item Measuring $O(r^2/\epsilon^2)$ copies of $\rho$ will, with high probability, leave us with $O(r^2/\epsilon^2)$ copies of $\rho|_{\bR}$. 
    \item The state $\rho|_{\bR}$ is $O(\epsilon)$-close to $\rho$ in Bures distance. 
\end{itemize}
With these facts established, the correctness of the algorithm follows readily from the Bures distance tomography algorithm of \cite{PSW25}. This algorithm requires $O(r^2/\epsilon^2)$ copies of $\rho|_{\bR}$, a state in a subspace of dimension $\rank(\bR) \leq \rank(\widehat{\bP}_1)+\rank(\widehat{\bP}_2) = 2r$, to produce an estimate $\widehat{\brho}$ such that $\DBur(\widehat{\brho}, \rho|_{\bR}) \leq O(\epsilon)$ with high probability. Then $\DBur(\widehat{\brho}, \rho) \leq O(\epsilon)$ by the triangle inequality.

\subsection{Step 1: properties of the projectors $\bA_i$ and $\bB_i$}

\begin{lemma}\label{lem:properties_of_bA_and_bB}
    The projectors $\bA_i$ and $\bB_i$ have the following properties, with high probability:
    \begin{itemize}
        \item[(i)] $\rank(\bA_i) = \rank(\bB_i) \geq (1 - \alpha) \cdot r$. 
        \item[(ii)] For all $\ket{v} \in \supp(\bA_i)$, we have $\bra{v} P \ket{v} \geq 1 - \epsilon^2/\alpha^2$. 
        % \item[(iii)] Similarly, for all $\ket{v} \in \supp(\bB_i)$, we have $\bra{v} \widehat{\bP}_i \ket{v} \geq 1 - \epsilon^2/\alpha^2$. \jack{dont think we need this one}
    \end{itemize}
\end{lemma}

\begin{proof}
    First, we observe that each $\widehat{\bP}_i$ is a good estimate for $P$: for both $i = 1$ and $2$, we have 
\begin{equation*}
        \dtr(P/r, \widehat{\bP}_i/r) = \dtr(\bU_i (P/r) \bU_i^\dagger, \bU_i (\widehat{\bP}_i/r) \bU_i^\dagger) = \dtr( \bU_i \rho \bU_i^\dagger, \bQ_i/r) \leq \epsilon,
\end{equation*}
with high probability. Next, take a Jordan decomposition of $P$ and $\widehat{\bP}_i$, and write $\ket*{\bu_j} \in P$ and $\ket*{\widetilde{\bu}_j} \in \widehat{\bP}_i$ for the Jordan vectors in the $j$-th nonzero block. Let $\bomega_j \coloneq \left| \braket*{\bu_j}{\widetilde{\bu}_j} \right|$, and $\bepsilon_j \coloneq (1 - \bomega_j^2)^{1/2}$. By \Cref{Jordans_lemma_td_fid_aff}, we have
        \begin{equation*}
            \dtr(P/r, \widehat{\bP}_i/r)  = \frac{1}{r} \sum_{j=1}^r \sqrt{1 - \bomega_j^2} = \frac{1}{r} \sum_{j=1}^r \bepsilon_j.
        \end{equation*}
Suppose, for sake of contradiction, that strictly fewer than $(1-\alpha) \cdot r$ of these blocks have $\bomega_j^2 \geq 1 - \epsilon^2/\alpha^2$, or equivalently, that strictly more than $\alpha \cdot r$ blocks have $\bepsilon_j > \epsilon/\alpha$. Then 
        \begin{equation*}
            \epsilon \geq \Dtr(P/r, \widehat{\bP}_i/r) = \frac{1}{r} \sum_{j=1}^{r} \bepsilon_j > \alpha \cdot  (\epsilon/\alpha) + (1-\alpha) \cdot 0 = \epsilon, 
        \end{equation*}
        which is our contradiction. Thus, there are at least $(1-\alpha) \cdot r$ Jordan blocks for which $\bomega_j^2 \geq 1 - \epsilon^2/\alpha^2$, which implies that $\dim( \WAS(P\mid  \widehat{\bP}_i)) = \dim( \WAS(\widehat{\bP}_i\mid P)) \geq (1-\alpha) \cdot r$. Since these dimensions are also the ranks of $\bA_i$ and $\bB_i$ respectively, we have property (i). 

    Write $J \coloneq \{ j \, : \, \bomega_j^2 \geq 1 - \epsilon^2/\alpha^2\}$. Then $\bA_i = \sum_{j \in J} \ketbra*{\widetilde{\bu}_j}$. For any unit vector $\ket{v} \in \supp(\bA_i)$, we can write $\ket{v} = \sum_{j \in J} \beta_j \ket*{\widetilde{\bu}_j}$, and we have
        \begin{equation*}
            \matrixel{v}{P}{v}
            = \bra{v} \Big(\sum_{j \in J} \ketbra{\bu_j}\Big) \ket{v}
            = \sum_{j \in J} |\beta_j|^2 \cdot \bomega_j^2 \geq \left(1 - \epsilon^2/\alpha^2 \right) \sum_{j \in J} |\beta_j|^2 = 1 -  \epsilon^2/\alpha^2. 
        \end{equation*}
        This proves item (ii). 
% \newpage
% \jack{Past text}

%         Note that $\matrixel{\bphi_j}{P}{\bphi_j} = \bra*{\bphi_j} \cdot \ketbra*{\bpsi_j} \cdot \ket*{\bphi_j} = \bomega_j^2$. Suppose, for sake of contradiction, that strictly fewer than $(1-\alpha)r$ of these blocks have $\bomega_j^2 \geq 1 - \epsilon^2/\alpha^2$, or equivalently, that strictly more than $\alpha \cdot r$ blocks have $\bepsilon_j > \epsilon/\alpha$. Then 
%         \begin{equation*}
%             \epsilon \geq \Dtr(P/r, \widehat{\bP}_i/r) = \frac{1}{r} \sum_{j=1}^{r} \bepsilon_j > \alpha \cdot  (\epsilon/\alpha) + (1-\alpha) \cdot 0 = \epsilon, 
%         \end{equation*}
%         which is our contradiction. So, at least $(1-\alpha)r$ of these Jordan blocks have $\bomega_j^2 \geq 1 - \epsilon^2/\alpha^2$. Define $J \coloneq \{j \, : \, \bomega_j^2 \geq 1 - \epsilon^2/\alpha^2\}$, and $\bA_j \coloneq \sum_{j \in J} \ketbra*{\bphi_j}$. Then $\rank(\bA_i) \geq (1-\alpha)r$, as the Jordan vectors $\ket*{\bphi_j}$ are orthogonal, and for any $\ket*{\varphi} = \sum_{j \in J} \beta_j \ket*{\bphi_j}$, we have
%         \begin{equation*}
%             \matrixel{\varphi}{P}{\varphi} = \sum_{j \in J} |\beta_j|^2 \cdot \bomega_j^2 \geq \left(1 - \epsilon^2/\alpha^2 \right) \sum_{j \in J} |\beta_j|^2 = 1 -  \epsilon^2/\alpha^2. \qedhere
%         \end{equation*}
\end{proof}

\begin{lemma}
     The distribution of $\bB_i$ is invariant under conjugation by unitaries of the form $U_P \oplus U_{\overline{P}}$. 
\end{lemma}

\begin{proof}
    First, we show the distribution of $\widehat{\bP}_i$ is invariant under such unitaries. Fix $W = U_P \oplus U_{\overline{P}}$. Note that with $\bU_i$ a Haar random unitary, $\bV_i \coloneq \bU_i W$ is Haar random too. Furthermore, since $W^\dagger P W = P$, we have
        \begin{equation*}
            \bU_i \rho \bU_i^\dagger  = (\bV_i W^\dagger)  \cdot \rho \cdot (W  \bV_i^\dagger) = \bV_i \cdot (W^\dagger \rho W) \cdot \bV_i^\dagger = \bV_i \rho \bV_i^\dagger.
        \end{equation*}
        Since both $\bU_i$ and $\bV_i$ are Haar random, $\bU_i^\dagger \mathcal{A}( \bU_i \rho \bU_i^\dagger) \bU_i$ and $\bV_i^\dagger \mathcal{A}(\bV_i \rho \bV_i^\dagger) \bV_i$ are identically distributed. Therefore, $\widehat{\bP}_i = \bU_i^\dagger \mathcal{A}( \bU_i \rho \bU_i^\dagger ) \bU_i$ is identically distributed to \begin{equation*}\bV_i^\dagger \mathcal{A}( \bV_i \rho \bV_i^\dagger ) \bV_i = W^\dagger \left(\bU_i^\dagger \mathcal{A} (\bU_i \rho \bU_i^\dagger) \bU_i \right) W = W^\dagger \widehat{\bP}_i W.\end{equation*}
        Thus, the claim holds for $\widehat{\bP}_i$. 

        Now we turn to $\bB_i$. Note that
        \begin{equation*}
            P \widehat{\bP}_i P = \sum_{j} \ketbra{\bu_j} \cdot \ketbra{\widetilde{\bu}_j} \cdot \ketbra{\bu_j} = \sum_j | \braket{\bu_j}{\widetilde{\bu}_j}|^2 \cdot \ketbra{\bu_j} = \sum_j |\bomega_j|^2 \cdot \ketbra{\bu_j},
        \end{equation*}
        whereas $\bB_i$ is $\sum_{j \in J} \ketbra{\bu_j}$, where $J = \{j \, : \, \bomega_j^2 \geq 1 - \epsilon^2/\alpha^2 \}$. Thus, $\bB_i$ can be formed from $P \widehat{\bP}_i P$, by taking the projector onto the eigenspaces with large enough eigenvalue. Write $f$ for this operation, so that $\bB_i = f( P \widehat{\bP}_i P)$. But $P \widehat{\bP}_i P$ is identically distributed to $P (W^\dagger \widehat{\bP}_i W) P = W^\dagger (P \widehat{\bP}_i P) W$, since $W$ and $W^\dagger$ both commute with $P$. Thus, $\bB_i = f( P \widehat{\bP}_i P)$ is identically distributed to $f( W^\dagger (P \widehat{\bP}_i P) W ) = W^\dagger f(P \widehat{\bP}_i P) W = W^\dagger \bB_i W$, where the first step holds since $f$ commutes with conjugation by a unitary. 
\end{proof}

\begin{corollary}
    Conditioned on $\rank(\bB_i) = r'$, $\bB_i$ is a Haar random rank-$r'$ projector on $\supp(P)$. 
\end{corollary}

\begin{proof}
    The distribution of $\bB_i$ (regarded as a subprojector on $\supp(P)$) is invariant under conjugation by unitaries $U_P$, by the previous lemma. That is, $\bB_i$ is identically distributed to $\bU_P \cdot \bB_i \cdot \bU_{P}^\dagger$ for $\bU_P \sim \mu_H(\supp(P))$, the Haar measure on $\supp(P)$. We can imagine that we draw $\bB$ first, and then $\bU_P$. If $\bB_i$ has fixed rank $r'$, then we end up with a Haar random rank-$r'$ projector on $\supp(P)$.
\end{proof}

% \begin{corollary}
%     Let $\bC_i$ be a Haar random projector on $\supp(\bB_i)$ of rank $r' \coloneq \lceil (1-\alpha) \cdot r \rceil$. Then $\bC_i$ is a Haar random rank-$r'$ projector on $\supp(P)$. 
% \end{corollary}

% \begin{proof}
%     Note that $\supp(\bB_i)$ is distributed identically to $\supp( U_P \bB_i U_P^\dagger) = U_P \cdot \supp(\bB_i) \cdot U_P^\dagger$. Since the Haar distribution on rank-$r$ projectors is unitarily inva

% \end{proof}

\subsection{Step 2: $\bB_1$ and $\bB_2$ robustly cover $P$}

\begin{lemma} \label{lem:lemma_using_levys}
    For $\alpha$ a sufficiently small constant, and for all $r$ sufficiently large, we have $\matrixel{u}{\bB_2}{u} \geq 0.9$ for all vectors $\ket*{u} \in \supp(\overline{\bB}_1)$ with high probability. Here $\overline{\bB}_1$ is the complement of $\bB_1$ in $\supp(P)$. 
\end{lemma}

\begin{proof}
    Let $\mu_H(P)$ denote the Haar measure on $\supp(P)$. First, note that for any fixed projector $F_{r'}$ of rank $r'$ at least $(1-\alpha) \cdot r$ on $\supp(P)$, we have
    \begin{equation*}\E_{\ket*{\bu} \sim \mu_H(P)} \big[\matrixel{\bu }{F_{r'}}{\bu}\big] = \tr \Big(  \E_{\ket*{\bu} \sim \mu_H(P)} \big[\ketbra*{\bu}\big] \cdot F_{r'} \Big) = \tr \left( \frac{P}{r} \cdot F_{r'} \right) \geq  1 - \alpha,\end{equation*}
    where we have used \Cref{proj_sym_subspace} in the second step.
    We can apply L\'{e}vy's lemma (\Cref{levys}) and \Cref{measurementslipschitz} to the function $f: \supp(P) \to \R$ given by $f(\ket*{u}) = \matrixel{u}{F_{r'}}{u}$, to get
    \begin{align*}\Pr_{\ket*{\bu}\sim \mu_H(P)} \Big[ \matrixel{\bu}{F_{r'}}{\bu} < 1- \alpha - \beta \Big] & \leq \Pr_{\ket*{\bu}\sim \mu_H(P)} \Big[ \Big| \matrixel{\bu}{F_{r'}}{\bu} - \E_{\ket*{\bu} \in P} \big[\matrixel{\bu }{F_{r'}}{\bu}\big] \Big| > \beta \Big] \\
    &\leq C_1 \exp \left( -C_2 \beta^2 r \right),\end{align*}
    for some constants $C_1$ and $C_2$, and any $\beta$. 
    
    Since $\bB_1$ and $\bB_2$ are independently sampled, and since each has a distribution invariant under conjugation by unitaries $U_P$, we can regard $\bB_1$ as a fixed projector, and $\bB_2$ as a random projector. We condition on the rank of each projector being at least $(1-\alpha)\cdot r$. Further conditioned on $\rank(B_i) = r'$, we can view $\bB_2$ as a Haar-random rotation of a fixed projector $F_{r'}$. For any fixed $\ket{u}$ in $P$, and any $\beta$, we have
    \begin{align*}
        \Pr_{\bB_2} \Big[ \matrixel{u}{\bB_2}{u} < 1 - \alpha - \beta \Big] & = \sum_{r' = \lceil (1-\alpha) \cdot r\rceil}^{r} \Pr_{\bB_2} \Big[ \rank(\bB_2) = r'\Big] \cdot \Pr_{\bB_2} \Big[\matrixel{u}{\bB_2}{u} < 1 - \alpha - \beta \, \Big| \, \rank(\bB_2) = r' \Big] \\
        & = \sum_{r' = \lceil (1-\alpha) \cdot r\rceil}^{r} \Pr_{\bB_2} \Big[ \rank(\bB_2) = r'\Big] \cdot \Pr_{\bU \sim \mu_H(P)} \Big[\matrixel{u}{\bU F_{r'} \bU^\dagger}{u} < 1 - \alpha - \beta \Big] \\
        & = \sum_{r' = \lceil (1-\alpha) \cdot r\rceil}^{r} \Pr_{\bB_2} \Big[ \rank(\bB_2) = r'\Big] \cdot \Pr_{\ket{\bu} \sim \mu_H(P)} \Big[\matrixel{\bu}{F_{r'}} {\bu} < 1 - \alpha - \beta \Big] \\
        & \leq \sum_{r' = \lceil (1-\alpha) \cdot r\rceil}^{r} \Pr_{\bB_2} \Big[ \rank(\bB_2) = r'\Big] \cdot \Big(C_1 \exp \left( -C_2 \beta^2 r \right)\Big) \\
        & = C_1 \exp \left( -C_2 \beta^2 r \right).
        % \Pr_{\bU \sim \mu_H(P)} \bigg[ \matrixel{\varphi}{\bU^\dagger F \bU}{\varphi} < 1 - \alpha - \beta \bigg] \\
        % & = \Pr_{\ket*{\bphi} \in \mu_H(P)} \bigg[ \matrixel{\bphi}{F}{\bphi}  < 1 - \alpha - \beta \bigg] \\
        % & \leq C_1 \exp \left( - C_2 \beta^2 r \right).
    \end{align*}
    
    We apply this bound to all vectors $\ket{u_i} \in N_\gamma$. Here, $N_\gamma$ is a fixed net of mesh $\gamma$ for $\supp(\overline{\bB}_1)$, where $\gamma$ is a sufficiently small constant we pick later. By a net of mesh $\gamma$, we mean a set of states $\{ \ket{u_i} \}$ such that for all $\ket{u} \in \supp(\overline{\bB}_1)$, there exists a $\ket{u_i} \in N_\gamma$ such that $\Dtr(\ketbra{u}, \ketbra{u_i}) \leq \gamma$. Since $\rank(\overline{\bB}_1) \leq \alpha r$, we can take $|N_\gamma| \leq (5/2\gamma)^{2\alpha r}$ by \cite[Lemma II.4]{HLSW04}. Then, by a union bound,
    \begin{align*}\Pr_{\bB_2} \bigg[ \exists \ket*{u_i} \in N_\gamma \, : \,  \matrixel{u_i}{\bB_2}{u_i} < 1 - \alpha - \beta \bigg] & \leq |N_\gamma| \cdot C_1 \exp \left( - C_2 \beta^2 r \right) \\
    & \leq C_1 \exp \left( (C_3 \alpha - C_2 \beta^2)r \right),
    \end{align*}
    where we are writing $C_3 = 2\ln(5/2\gamma)$.
    If we fix $\beta$ and $\gamma$, so that $1 - \beta - \gamma > 0.95$ then provided $\alpha$ is a sufficiently small constant, we have (i) that $1-\alpha-\beta-\gamma > 0.9$, and (ii) that the exponent $C_3\alpha - C_2 \beta^2$ is a negative constant. Hence for $r > r_0$ with $r_0$ some sufficiently large constant, this probability is less than $0.01$. Finally, since $N_\gamma$ has mesh $\gamma$, and since $f$ has Lipschitz number $1$ (\Cref{measurementslipschitz}), with high probability over $\bB_2$, $\matrixel{u}{\bB_2}{u} \geq 1 - \alpha - \beta - \gamma > 0.9$ for all $\ket*{\varphi} \in \overline{\bB}_1$. This is because for any $\ket{u}$ there exists a $\ket{u_i} \in N_\gamma$ such that
    \begin{equation*}
        \left|\matrixel{u}{\bB_2}{u} - \matrixel{u_i}{\bB_2}{u_i} \right| \leq \Dtr( \ketbra{u}, \ketbra{u_i}) \leq \gamma. \qedhere
    \end{equation*}
\end{proof}

For the remainder of the proof, we assume $\alpha$ is a sufficiently small constant, and $r$ sufficiently large, for \Cref{lem:lemma_using_levys}. 

\begin{corollary} \label{cor:robust_covering}
    The projectors $\bB_1$ and $\bB_2$ robustly cover $P$ with high probability. 
\end{corollary}

\begin{proof}
    For $\ket{u} \in \supp(P)$, we have $\bra{u} (\bB_1 + \bB_2) \ket{u} = \matrixel{u}{\bB_1}{u} + \matrixel{u}{\bB_2}{u}$. By \Cref{lem:lemma_using_levys}, if $\matrixel{u}{\bB_1}{u}=0$, then $\matrixel{u}{\bB_2}{u} \geq 0.9$, so $\bra{u} (\bB_1 + \bB_2) \ket{u} \neq 0$ for any $\ket{u} \in \supp(P)$. Thus, $\bB_1 + \bB_2$, an operator with support on $\supp(P)$, is full-rank in $\supp(P)$. That is, $\rank(\bB_1 + \bB_2) = \rank(P)$. 

    Now take a Jordan block decomposition of $\bB_1$ and $\bB_2$, and consider a $2 \times 2$ block in the decomposition, $B$. Suppose $\bB_1|_{B} = \ketbra{\bu}$ and $\bB_2|_{B} = \ketbra{\bv}$, and write $\bomega_B = |\braket{\bu}{\bv}|$. Let $\ket*{\bu^\perp}$ be a vector in $B$ such that $\braket*{\bu}{\bu^\perp} = 0$. Then $\matrixel*{\bu^\perp}{\bB_1}{\bu^\perp} = 0$ so that \Cref{lem:lemma_using_levys} implies $\matrixel*{\bu^\perp}{\bB_2}{\bu^\perp} = | \braket*{\bu^\perp}{\bv} |^2 \geq 0.9$. Thus, \begin{equation*}\bomega_B^2 = | \braket*{\bu}{\bv} |^2 = 1 - | \braket*{\bu^\perp}{\bv} |^2 \leq 0.1.\end{equation*}
    Since $B$ is arbitrary, $\bB_1$ and $\bB_2$ robustly cover $P$. 
\end{proof}

\subsection{Step 3: lifting $\calO_P$ to $\widetilde{\calO}_P$}

\begin{lemma}
    We have
    % \begin{equation*}
    %     \matrixel*{\widetilde{\psi}}{P}{\widetilde{\psi}} \geq 1 - O(\epsilon^2),
    % \end{equation*}
    \begin{equation*}
        \big| \braket*{\psi}{\widetilde{\psi}} \big|^2 \geq 1 - 3\epsilon^2/\alpha^2,
    \end{equation*}
    for each matching pair $\ket{\psi} \in \calO_P$ and $\ket*{\widetilde{\psi}} \in \widetilde{\calO}_P$ (e.g.\ $\ket{\bu_B}$ and $\ket{\widetilde{\bu}_B}$). 
\end{lemma}

\begin{proof}
    We have a couple cases:
\begin{enumerate}
    \item[(i)] If $\ket{\psi} \in \{ \ket{\bu_B} \}_{B \in \calB_1} \cup \{ \ket{\bv_B} \}_{B \in \calB_2} \cup \{ \ket{\bw_{1,B}} \}_{B \in \calB_{12}}$, then 
    \begin{equation*}
         \big| \braket*{\psi}{\widetilde{\psi}} \big|^2 = \big| \matrixel*{\psi}{P}{\widetilde{\psi}} \big|^2 = \frac{\big|\matrixel*{\widetilde{\psi}}{P}{\widetilde{\psi}}\big|^2}{\norm*{ P \ket*{\widetilde{\psi}}}^2} = \matrixel*{\widetilde{\psi}}{P}{\widetilde{\psi}} \geq 1 - \epsilon^2/\alpha^2,
    \end{equation*}
    where in the first step we have used $P\ket{\psi} = \ket{\psi}$, in the second we have used $P \ket*{\widetilde{\psi}}/ \norm*{ P \ket*{\widetilde{\psi}}} = \ket{\psi}$, and in the last step we have used \Cref{lem:properties_of_bA_and_bB}, item (ii). 
    % If $\ket*{{\psi}} = \ket*{{\bu}_B}$ for some block $B \in \calB_1$ or $\ket*{\widetilde{\psi}} = \ket*{\widetilde{\bw}_{1,B}}$ for some block $B \in \calB_{12}$, then $\ket*{\widetilde{\psi}_B} \in \supp(\bA_1)$, we have
    % % \begin{equation*}
    % %     \matrixel*{\widetilde{\bu}_B}{P}{\widetilde{\bu}_B} \geq  \braket{\widetilde{\bu}_B}{{\bu_B}}\cdot \braket{{\bu_B}}{\widetilde{\bu}_B} = \frac{\matrixel{{\bu}_B} {\widehat{\bP}_1}{\bu_B}\cdot \matrixel{{\bu}_B} {\widehat{\bP}_1}{\bu_B} }{\bra{{\bu}_B} \widehat{\bP}_1  \ket{{\bu}_B}} = \matrixel{\bu_B}{\widehat{\bP}_1}{\bu_B} \geq 1 - \epsilon^2/\alpha^2,
    % % \end{equation*}
    % % where in the first step we use $\ketbra{\bu_B} \preceq P$, in the second we use $\ket*{\widetilde{\bu}_B} \propto \widehat{\bP}_1 \ket{\bu_B}$, and in the last step we use \Cref{lem:properties_of_bA_and_bB}, item (ii). 
    % \begin{equation*}
    %     \matrixel*{\widetilde{\bu}_B}{P}{\widetilde{\bu}_B} \geq 1 - \epsilon^2/\alpha^2,
    % \end{equation*}
    % by \Cref{lem:properties_of_bA_and_bB}, item (ii) directly.

    % \item[(ii)] If $\ket*{\widetilde{\psi}} = \ket{\widetilde{\bv}_B}$, for some block $B \in \calB_2$, then the same reasoning as (i) holds, except now $\ket*{\widetilde{\psi}} \in \supp(\bA_2)$. 
    \item[(ii)] The last case is if $\ket*{{\psi}} = \ket*{{\bw}^\perp_{1,B}}$, for some block $B \in \calB_{12}$. This case is more involved, as $\ket*{\widetilde{\bw}^\perp_{1,B}}$ has a more complicated construction. We have
    \begin{equation*}
        \ket*{{\bw}^\perp_{1,B}} \propto \ket{\bw_{2,B}} - \braket{\bw_{1,B}}{\bw_{2,B}} \cdot  \ket{\bw_{1,B}} \eqcolon \ket{\bx},
    \end{equation*}
    and
    \begin{equation*}
        \ket*{\widetilde{\bw}^\perp_{1,B}} \propto \ket{\widetilde{\bw}_{2,B}} - \braket{{\bw}_{1,B}}{{\bw}_{2,B}} \cdot \ket{\widetilde{\bw}_{1,B}} \eqcolon \ket*{\widetilde{\bx}},
    \end{equation*} 
    where $\ket{\bx}$ and $\ket*{\widetilde{\bx}}$ are unnormalized vectors.
    We start by writing:
    \begin{equation}
        % \matrixel*{\widetilde{\bw}^\perp_{1,B}}{P}{\widetilde{\bw}^\perp_{1,B}} \geq \braket*{\widetilde{\bw}^\perp_{B,1}}{{\bw}^\perp_{1,B}} \cdot \braket*{{\bw}^\perp_{1,B}}{\widetilde{\bw}^\perp_{B,1}} = 
        \big|\braket*{\widetilde{\bw}^\perp_{1,B}}{{\bw}^\perp_{1,B}}\big|^2 = \frac{\big| \braket*{\widetilde{\bx}}{\bx}\big|^2}{\braket*{\widetilde{\bx}} \cdot \braket*{\bx}}. \label{eq:ratio_start}
    \end{equation}
    We first consider the numerator, expanding $\braket*{\widetilde{\bx}}{\bx}$ as 
    \begin{equation*}
        \braket*{\widetilde{\bx}}{\bx}  = \braket*{\widetilde{\bw}_{2,B}}{\bw_{2,B}} + \bomega_B^2 \cdot \braket*{\widetilde{\bw}_{1,B}}{\bw_{1,B}} - \braket{\bw_{1,B}}{\bw_{2,B}} \cdot \braket*{\widetilde{\bw}_{2,B}}{\bw_{1,B}}- \braket{\bw_{2,B}}{\bw_{1,B}} \cdot \braket*{\widetilde{\bw}_{1,B}}{\bw_{2,B}},
    \end{equation*}
    where $\bomega_B^2 = | \braket{\bw_{1,B}}{\bw_{2,B}}|^2$. Now we use the fact that 
    \begin{equation}
        \braket*{\widetilde{\bw}_{i,B}}{\bw_{j,B}} = \matrixel*{\widetilde{\bw}_{i,B}}{P}{\bw_{j,B}} =  \norm*{ P \ket*{\widetilde{\bw}_{i,B}}} \cdot \braket{\bw_{i,B}}{\bw_{j,B}}. \label{eq:norm_fact}
    \end{equation}
    This gives:
    \begin{align*}
        \braket*{\widetilde{\bx}}{\bx}  & = \norm*{ P \ket*{\widetilde{\bw}_{2,B}}} + \bomega_B^2 \cdot \norm*{ P \ket*{\widetilde{\bw}_{1,B}}}  - \bomega_B^2 \cdot \norm*{ P \ket*{\widetilde{\bw}_{2,B}}} - \bomega_B^2 \cdot \norm*{ P \ket*{\widetilde{\bw}_{1,B}}} \nonumber \\
        & = (1 - \bomega_B^2) \cdot \norm*{ P \ket*{\widetilde{\bw}_{2,B}}}. 
    \end{align*}
    We can bound this via
    \begin{equation}
        \norm*{ P \ket*{\widetilde{\bw}_{i,B}}}^2 = \matrixel*{\widetilde{\bw}_{i,B}}{P}{\widetilde{\bw}_{i,B}} \geq 1 - \epsilon^2/\alpha^2, \label{eq:norm_P_bound}
    \end{equation}
    using \Cref{lem:properties_of_bA_and_bB}, item (ii). Thus, the numerator can be lower bounded as
    \begin{equation}
        |\braket*{\widetilde{\bx}}{\bx}|^2 \geq (1 - \bomega_B^2)^2 \cdot \big(1 - \epsilon^2/\alpha^2 \big). \label{eq:numerator}
    \end{equation}
    
    Now we consider the denominator. First,
    \begin{equation}
        \braket{\bx}{\bx} = \braket{\bw_{2,B}} + \bomega_B^2 \cdot \braket{\bw_{1,B}} - 2 \bomega_B^2 = 1 - \bomega_B^2. \label{eq:denominator_1}
    \end{equation}
    Next, 
    \begin{align}
        \braket*{\widetilde{\bx}} & = \braket*{\widetilde{\bw}_{2,B}} + \bomega_B^2 \cdot \braket*{\widetilde{\bw}_{1,B}} - \braket{\bw_{1,B}}{\bw_{2,B}} \cdot \braket*{\widetilde{\bw}_{2,B}}{\widetilde{\bw}_{1,B}} - \braket{\bw_{2,B}}{\bw_{1,B}} \cdot \braket*{\widetilde{\bw}_{1,B}}{\widetilde{\bw}_{2,B}} \nonumber \\
        & = 1 + \bomega_B^2 - 2 \Re \Big[\braket{\bw_{1,B}}{\bw_{2,B}} \cdot \braket*{\widetilde{\bw}_{2,B}}{\widetilde{\bw}_{1,B}} \Big]. \label{eq:denom_intermediate}
    \end{align}
    The last term can be bounded as:
    \begin{align*}
        & \Re \Big[\braket{\bw_{1,B}}{\bw_{2,B}} \cdot \braket*{\widetilde{\bw}_{2,B}}{\widetilde{\bw}_{1,B}} \Big] \\
        & = \Re \Big[\braket{\bw_{1,B}}{\bw_{2,B}} \cdot \matrixel*{\widetilde{\bw}_{2,B}}{(P + \overline{P})}{\widetilde{\bw}_{1,B}} \Big] \\
        & = \Re \Big[\braket{\bw_{1,B}}{\bw_{2,B}} \cdot \matrixel*{\widetilde{\bw}_{2,B}}{P}{\widetilde{\bw}_{1,B}} \Big] + \Re \Big[\braket{\bw_{1,B}}{\bw_{2,B}} \cdot \matrixel*{\widetilde{\bw}_{2,B}}{\overline{P}}{\widetilde{\bw}_{1,B}} \Big] \\
        & = \bomega_B^2 \cdot \norm*{ P \ket*{\widetilde{\bw}_{1,B}}} \cdot \norm*{ P \ket*{\widetilde{\bw}_{2,B}}} + \Re \Big[\braket{\bw_{1,B}}{\bw_{2,B}} \cdot \matrixel*{\widetilde{\bw}_{2,B}}{\overline{P}}{\widetilde{\bw}_{1,B}} \Big]\\
        & \geq \bomega_B^2 \cdot \norm*{ P \ket*{\widetilde{\bw}_{1,B}}} \cdot \norm*{ P \ket*{\widetilde{\bw}_{2,B}}} - | \braket{\bw_{1,B}}{\bw_{2,B}}| \cdot \norm*{ \overline{P} \ket*{\widetilde{\bw}_{2,B}}} \cdot \norm*{ \overline{P} \ket*{\widetilde{\bw}_{1,B}}} \tag{Cauchy-Schwarz},
    \end{align*}
    where the second-to-last step is because $P\ket{\widehat{\bw}_{i, B}} = \Vert P \ket{\widehat{\bw}_{i, B}}\Vert \cdot \ket{\bw_{i, B}}$.
    We can bound this further using \Cref{eq:norm_P_bound}, and 
    \begin{equation*}
        \norm*{ \overline{P} \ket*{\widetilde{\bw}_{B,i}}}^2 = \matrixel*{\widetilde{\bw}_{B,i}}{\overline{P}}{\widetilde{\bw}_{B,i}} = 1 - \matrixel*{\widetilde{\bw}_{B,i}}{P}{\widetilde{\bw}_{B,i}} \leq \epsilon^2/\alpha^2. 
    \end{equation*}
    Thus, 
    \begin{equation*}
        \Re \Big[\braket{\bw_{1,B}}{\bw_{2,B}} \cdot \braket*{\widetilde{\bw}_{2,B}}{\widetilde{\bw}_{1,B}} \Big] \geq \bomega_B^2 \cdot \big( 1 - \epsilon^2/ \alpha^2 \big) - \bomega_B \cdot \big( \epsilon^2/\alpha^2 \big) \geq \bomega_B^2 - 2 \bomega_B \cdot \epsilon^2/\alpha^2.
    \end{equation*}
    Substituting this back into \Cref{eq:denom_intermediate} gives:
    \begin{equation*}
        \braket*{\widetilde{\bx}} \leq 1 + \bomega_B^2 - 2 \big( \bomega_B^2 - 2 \bomega_B \cdot \epsilon^2/\alpha^2\big) = (1 - \bomega_B^2) \cdot \Big( 1 + \frac{4\bomega_B}{1 - \bomega_B^2} \cdot \epsilon^2/\alpha^2\Big) \leq (1 - \bomega_B^2)\cdot \big( 1 + 2 \epsilon^2/\alpha^2),
    \end{equation*}
    using $\bomega_B^2 \leq 0.1$ in the last step. From this and \Cref{eq:denominator_1}, the denominator can be upper bounded as:
    \begin{equation*}
        \braket{\bx} \cdot \braket*{\widetilde{\bx}} \leq (1 - \bomega_B^2)^2 \cdot \big( 1 +2 \epsilon^2/\alpha^2\big).
    \end{equation*}
    Finally, from this and \Cref{eq:ratio_start,eq:numerator,} gives:
    \begin{equation*}
          \big|\braket*{\widetilde{\bw}^\perp_{1,B}}{{\bw}^\perp_{1,B}}\big|^2 =  \frac{\big| \braket*{\widetilde{\bx}}{\bx}\big|^2}{\braket*{\widetilde{\bx}} \cdot \braket*{\bx}} \geq \frac{(1-\bomega_B^2)^2 \cdot \big( 1 - \epsilon^2/\alpha^2 \big)}{(1- \bomega_B^2)^2 \cdot \big( 1 + 2 \epsilon^2/\alpha^2 \big)} \geq \big( 1- \epsilon^2/\alpha^2 \big) \cdot \big(1 - 2\epsilon^2/\alpha^2 \big) \geq 1 - 3\epsilon^2/\alpha^2.
    \end{equation*}
\end{enumerate}
So, in all cases we have $|\braket*{{\psi}}{\widetilde{\psi}}|^2 \geq 1 - 3\epsilon^2/\alpha^2$. 
\end{proof}

\begin{corollary} \label{cor:overlap_of_P_and_R}
    We have $\tr(\bR \cdot \rho) \geq 1 - O(\epsilon^2)$.
\end{corollary}

\begin{proof}
Recall
    \begin{equation*}
    \calO_P \coloneq \big\{ \ket{\bu_B} \big\}_{B \in \calB_1} \cup \big\{  \ket{\bv_B} \big\}_{B \in \calB_2} \cup \big\{ \ket{\bw_{1,B}}, \ket*{\bw_{1,B}^\perp} \big\}_{B \in \calB_{12}}
\end{equation*}
forms an orthonormal basis for $\supp(P)$, and that each vector in the lift $\widetilde{\calO}_P$,
\begin{equation*}
    \widetilde{\calO}_P \coloneq \big\{ \ket{\widetilde{\bu}_B} \big\}_{B \in \calB_1} \cup \big\{  \ket{\widetilde{\bv}_B} \big\}_{B \in \calB_2} \cup \big\{ \ket{\widetilde{\bw}_{1,B}}, \ket*{\widetilde{\bw}_{1,B}^\perp} \big\}_{B \in \calB_{12}},
\end{equation*}
is in $\supp(\bA_1 + \bA_2) \subseteq \supp(\bR)$. So, we have 
\begin{align*}
    \tr(\bR \cdot P) & = \sum_{B \in \calB_1} \matrixel*{\bu_B}{\bR}{\bu_B} + \sum_{B \in \calB_2}\matrixel*{\bv_B}{\bR}{\bv_B} + \sum_{B \in \calB_{12}} \matrixel*{\bw_{1,B}}{\bR}{\bw_{1,B}}+ \sum_{B \in \calB_{12}}\matrixel*{\bw_{2,B}}{\bR}{\bw_{2,B}} \\
    & \geq \sum_{B \in \calB_1} |\braket*{\bu_B}{\widetilde{\bu}_B}|^2 + \sum_{B \in \calB_2}|\braket*{\bu_B}{\widetilde{\bv}_B}|^2 + \sum_{B \in \calB_{12}} |\braket*{\bw_{1,B}}{\widetilde{\bw}_{1,B}}|^2+ \sum_{B \in \calB_{12}}|\braket*{\bw_{2,B}}{\widetilde{\bw}_{2,B}}|^2 \\
    & \geq r \cdot \big( 1 - 3\epsilon^2/\alpha^2 \big),
\end{align*}
where in the last step, we have used that there are $r$ terms across all four sums, since they index vectors forming an orthonormal basis of $\supp(P)$ which has dimension $r$, and that each term is at least $(1 - 3\epsilon^2/\alpha^2)$ by the previous lemma. The result follows since $\rho = P/r$. 
\end{proof}

\subsection{Step 4: Bures distance learning in $\bR$}

\begin{lemma}
    We have the following:
    \begin{itemize}
        \item[(i)] With high probability, measuring $O(r^2/\epsilon^2)$ copies with $\{ \bR, \overline{\bR} \}$ leaves us with $O(r^2/\epsilon^2)$ copies of $\rho|_{\bR}$. 
        \item[(ii)] $\DBur(\rho, \rho|_{\bR}) \leq O(\epsilon)$. 
    \end{itemize}
\end{lemma}

\begin{proof}
\begin{itemize}
    \item[(i)] From \Cref{cor:overlap_of_P_and_R}, we have that the probability of obtaining a copy of $\rho_{\bR}$ is
    \begin{equation*}
        \tr(R \cdot \rho) \geq 1 - O(\epsilon^2).
    \end{equation*}
    In particular, this is with high probability for sufficiently small $\epsilon$. Thus, $O(r^2/\epsilon^2)$ copies of $\rho$ suffice to obtain $O(r^2/\epsilon^2)$ copies of $\rho|_{\bR}$, by Markov's inequality.
    \item[(ii)] Note that
    \begin{equation*}
        \rho|_{\bR} = \frac{\bR \rho \bR}{\tr(\rho \bR)} = \frac{\bR P \bR}{\tr(P \bR)},
    \end{equation*}
    so that
    \begin{equation*}
        \Fid(\rho, \rho|_{\bR}) = \tr \sqrt{ \sqrt{\rho} \cdot \rho|_{\bR} \cdot \sqrt{\rho} } = \tr \sqrt{ \frac{ P \cdot (\bR P \bR) \cdot  P}{r \tr( P \bR)}} = \frac{\tr(P \bR P)}{\sqrt{ r \tr(P \bR)}} = \sqrt{ \frac{\tr(P\bR)}{r}}.
    \end{equation*}
    Thus, by \Cref{cor:overlap_of_P_and_R}, 
    \begin{equation*}
        \Fid(\rho, \rho|_{\bR}) \geq \sqrt{1 - O(\epsilon^2)} \geq 1 - O(\epsilon^2). 
    \end{equation*}
    Converting this to Bures distance, we get:
    \begin{equation*}
        \DBur(\rho, \rho|_{\bR}) = \sqrt{2 ( 1 - \Fid(\rho, \rho|_{\bR}))} \leq O(\epsilon). \qedhere
    \end{equation*}
\end{itemize}    
\end{proof}

\begin{lemma}
    The bootstrapped algorithm succeeds with probability at least $95\%$ given $2n + O(r^2/\epsilon^2)$ samples. 
\end{lemma}

\begin{proof}
    There are only a small number of probabilistic steps that must succeed for the algorithm to return a good estimate. Each succeeds with high probability, perhaps conditioned on the previous steps. These are:
    \begin{itemize}
        \item[(i)] We need $\Dtr(\rho, \widehat{\bP}_1/r) \leq \epsilon$.
        \item[(ii)] We need $\Dtr(\rho, \widehat{\bP}_2/r) \leq \epsilon$.
        \item[(iii)] We need $\bB_1$ and $\bB_2$ to robustly cover $P$, which holds with high probability by \Cref{cor:robust_covering}.
        \item[(iv)] We need $O(r^2/\epsilon^2)$ samples to prepare $O(r^2/\epsilon^2)$ copies of $\rho|_{\bR}$, and this succeeds with high probability.
        \item[(v)] Finally, the Bures distance learning algorithm of \cite{PSW25} succeeds with high probability.
    \end{itemize}
    Union bounding over these five events gives us the claimed success probability of $95\%$.
\end{proof}

\section{Lower bounds on learning in trace distance}

In this section, we use our bootstrapping algorithm to conclude lower bounds on the sample complexity required to learn a rank-$r$ projector in trace distance.

\begin{theorem}[A lower bound on learning rank-$r$ projector states in trace distance]\label{thm:lower_bound_td_projectors}
    Any rank-$r$ projector tomography algorithm learning to trace distance $\epsilon > 0$ requires at least $n = \Omega(rd/\epsilon^2)$ samples, for $r \in [r_0, c_1\cdot d]$, and $\epsilon < \epsilon_0$, where $r_0$ is a sufficiently large constant, and $c_1$ and $\epsilon_0$ are sufficiently small constant. 
\end{theorem}

\begin{proof}
    Suppose an algorithm $\calA$ existed that could learn rank-$r$ projectors to trace distance $\epsilon$ with probability at least $99\%$, for such $r$ and $\epsilon$, using $n \leq c_2 rd/\epsilon^2$ samples, for some constant $c_2$ we pick later. Then, by \Cref{prop_bootstrapping}, we could bootstrap it to an algorithm $\calA'$ that learns to Bures distance $O(\epsilon)$ using $n \leq 2c_2 rd/\epsilon^2 + c_3 r^2/\epsilon^2 \leq (2c_2 + c_1c_3) rd/\epsilon^2$ samples, with probability $95\%$, for some constant $c_3$ (arising from the Bures distance learning algorithm of \cite{PSW25}). By \Cref{lem:boosting_success_probability}, we can convert this to an algorithm learning to $O(\epsilon)$ with $99\%$ probability, using $O( (2c2 + c_1c_3) rd/\epsilon^2)$ copies. For both $c_1$ and $c_2$ sufficiently small, this contradicts \Cref{projector_lower_bound}, which requires that $n \geq rd/128\epsilon^2$. Therefore, having chosen such a $c_2$ sufficiently small, we must have $n = \Omega(rd/\epsilon^2)$. 
\end{proof}

This result further implies the following lower bound on learning generic rank-$r$ mixed states.

\begin{theorem}[A lower bound on learning rank-$r$ mixed states in trace distance]\label{thm:lower_bound_td}
    Given a rank-$r$ mixed state $\rho \in \C^{d \times d}$, $n = \Omega(rd/\epsilon^2)$ copies are required to estimate it to trace distance error $\epsilon > 0$, for sufficiently small $\epsilon$, and $d >1$. 
\end{theorem}

\begin{proof}
    We combine the following observations:
    \begin{itemize}
        \item[(i)] For $r \in [r_0, c \cdot d]$ and $d$ large enough so that $r_0 < c \cdot d$, we can directly appeal to \Cref{thm:lower_bound_td_projectors} since any general rank-$r$ tomography algorithm is also a rank-$r$ projector tomography algorithm.
        \item[(ii)] For $r \in [c \cdot d,d]$, with $d$ large enough so that $r_0 < c \cdot d$, we can use rank-$(c\dot d)$ projectors instead to obtain a lower bound of $n = \Omega(d^2/\epsilon^2) = \Omega(rd/\epsilon^2)$. 
        \item[(iii)] If $r < r_0$ and $d > 1$, then we can obtain a lower bound of $n = \Omega(d/\epsilon^2)$ from the pure state Bures distance learning lower bound, \Cref{pure_state_lower_bound}. This is because for pure states, we have 
        \begin{equation*}\frac{1}{\sqrt{2}} \DBur \leq \Dtr \leq \DBur,\end{equation*}
        using the pure state formulas for trace distance, fidelity and Bures distance, given in \Cref{def:td,def:fid,def:Bur}. Thus, learning to $\epsilon$ trace distance is equivalent to learning to $O(\epsilon)$ Bures distance, and we conclude an $\Omega(d/\epsilon^2) = \Omega(dr/\epsilon^2)$ lower bound in this case. 
        % \item[(iv)] Finally, if $c\cdot d \leq r_0$, then $d$ is constant, and hence $r$ is constant too. In this case, we have a $\Omega(1/\epsilon^2) =  \Omega(rd/\epsilon^2)$  from the classical lower bound on the number of samples need to learn a $p$-biased coin to error $\epsilon$~\cite{Can20}.
    \end{itemize}
    Thus, for sufficiently small $\epsilon$ to cover all three cases, we have a lower bound of $\Omega(rd/\epsilon^2)$. 
\end{proof}

\section{The pretty good measurement}
\newcommand{\dP}{\mathrm{d}P}
\renewcommand{\Sym}{\mathrm{Sym}}
\newcommand{\dx}{\mathrm{d}x}

This section treats the Pretty Good Measurement (PGM), a natural measurement which we show has optimal sample complexity for the problem of projector tomography. 

% We also show that, as a curiosity, in the pure state setting with a Haar random prior, the PGM is the optimal algorithm among all proper tomography algorithms. This suggests that the PGM may also be optimal for proper rank-$r$ projector tomography, though we leave this question to future work. 

\begin{definition}[The Pretty Good Measurement] \label{def:pgm}
    Let $\{ \rho_i \}_{i=1}^m \subseteq \mathbb{C}^{d \times d}$ be a finite set of density matrices, and let $\{ \alpha_i \}_{i=1}^m$ be a probability distribution on $[m]$. We define the average state as $S \coloneq \sum_{i=1}^m \alpha_i \rho_i$. Then the \emph{Pretty Good Measurement} associated with this ensemble is the POVM $\{ M_i \}_{i=1}^m$ defined by
\begin{align*}
    M_i := S^{-1/2} \cdot  \alpha_i \rho_i \cdot S^{-1/2},
\end{align*}
where $S^{-1/2}$ denotes the \emph{Moore–Penrose pseudoinverse} of $S^{1/2}$, i.e.\ the inverse restricted to the support of~$S$.
\end{definition} 

The PGM's name derives from the fact that it is \emph{pretty good} at the problem of \emph{state discrimination}. State discrimination is the following task: we are given a single copy of $\rho_{\bi} \in \{ \rho_i \}_{i =1}^m$, where $\bi$ is generated according to a known distribution $\alpha$, and we are asked to identify $\bi$. Let $P_{\mathrm{PGM}}$ be the probability that the PGM succeeds at this task, and let $P_{\mathrm{OPT}}$ be the optimal success probability, maximized over all possible~POVMs. 

\begin{theorem}[\cite{BK02}]
    $P_{\mathrm{PGM}} \geq P_{\mathrm{OPT}}^2$. In particular, if $P_{\mathrm{OPT}} \geq 1- \delta$, then $P_{\mathrm{PGM}} \geq 1 - 2\delta$. 
\end{theorem}

We can reformulate (proper) projector tomography as a continuous version of state discrimination: we are given a state of the form $(P/r)^{\otimes n}$, for $P$ a rank-$r$ projector, and asked to identify our input state among all such $n$-fold projector states. Thus, a continuous version of the PGM is a natural measurement to consider for projector tomography. We parameterize the input state via $P = U Q U^\dagger/r$, for some fixed rank-$r$ projector $Q$, and $U \in U(d)$. Then, measurement operators can be defined as
\begin{equation*}M_{U}  = S^{-1/2} \cdot (U Q U^\dagger/r) ^{\otimes n} \cdot S^{-1/2} \cdot \dU = \frac{1}{r^n} S^{-1/2} \cdot (U Q U^\dagger) ^{\otimes n} \cdot S^{-1/2} \cdot \dU,\end{equation*}
with 
\begin{align*}
    S  = \int_U (U Q U^\dagger/r)^{\otimes n} \cdot \dU &  \cong \frac{1}{r^n} \sum_{\substack{\lambda \vdash n \\ \ell(\lambda) \leq d}} \ketbra{\lambda} \otimes I_{\dim(\lambda)} \otimes \int_U \nu_\lambda (U Q U^\dagger) \cdot \dU \\
    & =\frac{1}{r^n} \sum_{\substack{\lambda \vdash n \\ \ell(\lambda) \leq d}} \frac{s_{\lambda}(1^r)}{s_\lambda(1^d)} \cdot \ketbra{\lambda} \otimes I_{\dim(\lambda)} \otimes I_{\dim(V^d_\lambda)} \cdot \dU, 
\end{align*}
using \Cref{lem:integral_over_unitary_register} in the last step. Substituting $S$ back into $M_U$ gives
\begin{equation}\label{PGM_operators}
    M_U \cong \sum_\lambda \frac{s_\lambda(1^d)}{s_\lambda(1^r)} \cdot \ketbra{\lambda} \otimes I_{\dim(\lambda)} \otimes  \nu_\lambda (U Q U^\dagger) \cdot \dU.
\end{equation}

\begin{definition}[The PGM for projector tomography]\label{PGM_def}
    The \emph{Pretty Good Measurement for rank-$r$ projector tomography} has POVM elements $\{M_U\}_{U \in U(d)}$, with $M_U$ given by \Cref{PGM_operators}. Upon measuring $\bU$, the PGM outputs $\bU (Q/r) \bU^\dagger$. 
\end{definition}

\begin{remark}
    In light of \Cref{lem:WLOG_WSS}, measuring with the PGM is equivalent to applying the following two~steps:
    \begin{enumerate}
        \item WSS to obtain $\blambda \vdash n$. The post-measurement state, written in the Schur basis, is
        \begin{equation*}
            \ketbra{\blambda} \otimes \frac{I_{\dim(\blambda)}}{\dim(\blambda)} \otimes \frac{\nu_{\blambda}(P)}{s_{\blambda}(1^r)}.
        \end{equation*}
        \item Within $V^d_{\blambda}$, measure with the POVM $M^{(\blambda)} = \{M^{(\blambda)}_U\}_{U \in U(d)}$ with operators
        \begin{equation*}
            M^{(\blambda)}_U  = \frac{s_{\blambda}(1^d)}{s_{\blambda}(1^r)} \cdot \nu_{\blambda} ( U Q U^\dagger ) \cdot \dU.
        \end{equation*}
    \end{enumerate}
    The second step is itself a continuous PGM, where we attempt to identify $\nu_\lambda(P)/s_\lambda(1^r)$, among all possible states in $V^d_\lambda$ of that form. 
\end{remark}

\subsection{Sample-optimality of the PGM for learning projectors}

In this subsection, we prove the following result, which also implies the upper bound in \Cref{thm:projector}.

\begin{proposition}[The PGM achieves optimal sample complexity for projector tomography]
The PGM defined in \Cref{PGM_def} can learn a rank-$r$ projector state to within Bures distance $\epsilon$ with high probability, using $n = O(rd/\epsilon^2)$ copies. 
\end{proposition}

\begin{proof}
    We start by studying the expected affinity, with the aim of lower-bounding the expected fidelity between the output $\widehat{\brho} = \bU Q \bU^\dagger/r$ and the input $\rho = P/r$:
    \begin{align*}
        \E_{\widehat{\brho} \sim \mathrm{PGM}} \Big[\Aff( \widehat{\brho}, \rho)\Big] & =  \E_{\widehat{\brho} \sim \mathrm{PGM}} \Big[ r \cdot \tr ( \widehat{\brho} \cdot \rho) \Big] \tag{\Cref{eq:aff_projectors}}\\ & = r \int_U \tr\left( M_U \cdot \left(\frac{P}{r}\right)^{\otimes n}\right) \cdot \tr( \frac{U Q U^\dagger }{r} \cdot \frac{P}{r} ) \\
    & = \frac{1}{r^{n+1}} \int_U \tr \left( M_U \cdot P^{\otimes n} \right) \cdot \tr \left( U Q U^\dagger \cdot P \right).
\end{align*}
With \Cref{PGM_operators}, we re-express the first factor from the integrand using the Schur basis:
\begin{align*}
    \tr\left( M_U \cdot P^{\otimes n}\right) & = \tr\Big( \sum_{\substack{\lambda \vdash n \\ \ell(\lambda) \leq d}} \frac{s_\lambda(1^d)}{s_\lambda(1^r)} \cdot \ketbra{\lambda} \otimes I_{\dim(\lambda)} \otimes \nu_\lambda( U Q U^\dagger \cdot P) \Big) \cdot \dU \\ & = \sum_{\substack{\lambda \vdash n \\ \ell(\lambda) \leq d}} \frac{s_\lambda(1^d)}{s_\lambda(1^r)} \cdot \dim(\lambda) \cdot s_\lambda(UQ U^\dagger \cdot P) \cdot \dU.
\end{align*}
Plugging this back into the integral, and using (i) the fact that the unitary irrep corresponding to $\lambda = (1)$ is the defining representation (\Cref{ex:defining_rep}), and (ii) Pieri's rule (\Cref{cor:Pieri's_rule}), we get
\begin{align*}
    \int_U \tr \left( M_U \cdot P^{\otimes n} \right) \cdot \tr \left( U Q U^\dagger \cdot P \right) \cdot \dU & = \sum_{\substack{\lambda \vdash n \\ \ell(\lambda) \leq d}} \frac{s_\lambda(1^d)}{s_\lambda(1^r)} \cdot \dim(\lambda) \cdot \int_U s_\lambda(UQU^\dagger \cdot P) \cdot s_{(1)}(U QU^\dagger \cdot P) \cdot \dU \\
    & = \sum_{\substack{\lambda \vdash n \\ \ell(\lambda) \leq d}} \frac{s_\lambda(1^d)}{s_\lambda(1^r)} \cdot \dim(\lambda) \cdot \tr \Big( \int_U \nu _\lambda(UQU^\dagger \cdot P) \otimes \nu_{(1)}(U QU^\dagger \cdot P) \cdot \dU \Big) \\
    & = \sum_{\substack{\lambda \vdash n \\ \ell(\lambda) \leq d}} \frac{s_\lambda(1^d)}{s_\lambda(1^r)} \cdot \dim(\lambda) \cdot \sum_{i=1}^d \tr \Big( \int_U \nu _{\lambda+e_i}(UQU^\dagger \cdot P) \cdot \dU \Big).
\end{align*}
From \Cref{lem:integral_over_unitary_register}, we have
\begin{equation*}
    \int_U \nu _{\lambda+e_i}(UQU^\dagger \cdot P) \cdot \dU  = \Big(\int_U \nu_{\lambda+e_i}(UQU^\dagger) \cdot \dU \Big) \cdot \nu_{\lambda+e_i}(P) = \frac{s_{\lambda+e_i}(1^r)}{s_{\lambda+e_i}(1^d)} \cdot \nu_{\lambda+e_i}(P).
\end{equation*}
Plugging this back into the trace, and combining all of our steps so far, we obtain
\begin{align*}
    \E_{\widehat{\brho} \sim \mathrm{PGM}} \Big[\Aff( \widehat{\brho}, \rho)\Big] & = \frac{1}{r^{n+1}} \sum_{\substack{\lambda \vdash n \\ \ell(\lambda) \leq d}} \frac{s_{\lambda}(1^d)}{s_{\lambda}(1^r)} \cdot \dim(\lambda) \cdot \sum_{i=1}^{d} \frac{s_{\lambda+e_i}(1^r)}{s_{\lambda+e_i}(1^d)} \cdot s_{\lambda+e_i}(1^r).
\end{align*}
At this point, it will be useful to reorganize the expression, and reintroduce factors of $r$ into the arguments of the Schur polynomials, so that we can interpret the terms in this sum as an expectation over $\blambda \sim \mathrm{WSS}_n(\rho)$. To do so, we use the relation $s_\lambda(\alpha) = s_{\lambda}(1^r/r) = s_{\lambda}(1^r)/r^{|\lambda|}$. We also use the notation $\Phi_\lambda(\alpha) = s_\lambda(\alpha)/s_\lambda(1^d)$. We get
\begin{align*}
    \E_{\widehat{\brho} \sim \mathrm{PGM}} \Big[\Aff( \widehat{\brho}, \rho)\Big] & =  r \sum_{\substack{\lambda \vdash n \\ \ell(\lambda) \leq d}} \dim(\lambda) \cdot \frac{s_\lambda(1^r)}{r^n} \cdot \sum_{i=1}^{d} \frac{s_{\lambda}(1^d)}{s_{\lambda+e_i}(1^d)} \cdot \left( \frac{1}{r} \cdot \frac{s_{\lambda+e_i}(1^r)}{s_{\lambda}(1^r)}\right)^2 \\
    & =  r\sum_{\substack{\lambda \vdash n \\ \ell(\lambda) \leq d}} \dim(\lambda) \cdot s_{\lambda}(\alpha) \cdot \sum_{i=1}^{d} \frac{s_{\lambda}(1^d)}{s_{\lambda+e_i}(1^d)} \cdot  \left(\frac{s_{\lambda+e_i}(\alpha)}{s_\lambda(\alpha)}\right)^2 \\
    & = r\sum_{\substack{\lambda \vdash n \\ \ell(\lambda) \leq d}} \dim(\lambda) \cdot s_\lambda( \alpha) \cdot \sum_{i=1}^{d}\frac{\Phi_{\lambda+e_i}(\alpha)}{\Phi_\lambda(\alpha)} \cdot \frac{s_{\lambda+e_i}(\alpha)}{s_\lambda(\alpha)}.
\end{align*}
We now use results proven in \cite{OW15b} and \cite{OW16} in an off-the-shelf manner to lower bound this quantity. Specifically, \cite{OW15b} shows that for any sorted probability distributions $\alpha$ on $[d]$:
\begin{equation*}
    \sum_{i=1}^{d}\frac{\Phi_{\lambda+e_i}(\alpha)}{\Phi_\lambda(\alpha)} \cdot \frac{s_{\lambda+e_i}(\alpha)}{s_\lambda(\alpha)} \geq \sum_{i=1}^d \frac{\Phi_{\lambda+e_i}(\alpha)}{\Phi_\lambda(\alpha)} \cdot \frac{\lambda_i}{n},
\end{equation*}
Meanwhile, Eq. (20) in Section 4.2 of \cite{OW16} shows that, again for such $\alpha$:
\begin{equation*}
    \sum_{\lambda} \dim(\lambda) \cdot s_\lambda(\alpha) \cdot\sum_{i=1}^d \frac{\Phi_{\lambda+e_i}(\alpha)}{\Phi_\lambda(\alpha)} \cdot \lambda_i  = \E_{\blambda \sim \mathrm{SW}^n(\alpha)} \Big[\sum_{i=1}^d \frac{\Phi_{\blambda+e_i}(\alpha)}{\Phi_{\blambda}(\alpha)} \cdot \blambda_i\Big] \geq n \cdot \norm{\alpha}_2^2 - \frac{3}{2} d.
\end{equation*}
Here, $\norm{ \cdot}_2$ is the $\ell_2$ norm, i.e.\ $\norm{\alpha}^2_2 =\sum_{i=1}^d |\alpha_i|^2$. Applying both of these to our special case of $\alpha = (1^r)/r$, which has $\norm{\alpha}_2^2 = r \cdot 1/r^2 = 1/r$, and using \Cref{cor:aff-fid} to convert our affinity bound to a fidelity bound, we get:
\begin{equation*}\label{tomographylowerbound}
    \E_{\widehat{\brho} \sim \mathrm{PGM}} \Big[\Fid( \widehat{\brho}, \rho)\Big] \geq \E_{\widehat{\brho} \sim \mathrm{PGM}} \Big[\Aff( \widehat{\brho}, \rho)\Big]  \geq \frac{r}{n}\left( n \cdot \norm{(1^r)/r}_2^2 - \frac{3}{2}d \right) = 1 - \frac{3rd}{2n}.
\end{equation*}
By Markov's inequality, taking $n = O(rd/\epsilon^2)$ therefore suffices to produce $\widehat{\brho}$ such that 
\begin{equation*}
\Pr_{\widehat{\brho} \sim \mathrm{PGM}} \Big[ \Fid(\widehat{\brho}, \rho) \geq 1 - \frac{1}{2} \epsilon^2 \Big] \geq 0.99. 
\end{equation*}
Equivalently, with this many samples we have learned $\rho$ to Bures distance at most $\epsilon$. 
\end{proof}

\section*{Acknowledgments}

We thank Aram Harrow and Henry Yuen for helpful conversations. J.S.\ and J.W.\ are supported by the NSF CAREER award CCF-233971.

\bibliographystyle{alpha}
\bibliography{wright}

\newcommand{\etalchar}[1]{$^{#1}$}
\begin{thebibliography}{PTTW25}

\bibitem[ANSV08]{ANSV08}
Koenraad Audenaert, Michael Nussbaum, Arleta Szko{\l}a, and Frank Verstraete.
\newblock Asymptotic error rates in quantum hypothesis testing.
\newblock {\em Communications in Mathematical Physics}, 279:251--283, 2008.

\bibitem[ARS88]{ARS88}
Robert Alicki, S{\l}awomir Rudnicki, and S{\l}awomir Sadowski.
\newblock Symmetry properties of product states for the system of {$N$} $n$-level atoms.
\newblock {\em Journal of mathematical physics}, 29(5):1158--1162, 1988.

\bibitem[Bel75]{Bel75}
Viacheslav Belavkin.
\newblock Optimal multiple quantum statistical hypothesis testing.
\newblock {\em Stochastics: An International Journal of Probability and Stochastic Processes}, 1(1-4):315--345, 1975.

\bibitem[BK02]{BK02}
Howard Barnum and Emanuel Knill.
\newblock Reversing quantum dynamics with near-optimal quantum and classical fidelity.
\newblock {\em Journal of Mathematical Physics}, 43(5):2097--2106, 2002.

\bibitem[BM99]{BM98}
Dagmar Bru{\ss} and Chiara Macchiavello.
\newblock Optimal state estimation for d-dimensional quantum systems.
\newblock {\em Physics Letters A}, 253:249--251, 1999.

\bibitem[Can20]{Can20}
Cl{\'e}ment Canonne.
\newblock A short note on learning discrete distributions, 2020.

\bibitem[CM06]{CM06}
Matthias Christandl and Graeme Mitchison.
\newblock The spectra of quantum states and the {K}ronecker coefficients of the symmetric group.
\newblock {\em Communications in mathematical physics}, 261(3):789--797, 2006.

\bibitem[FO24]{FO24}
Steven Flammia and Ryan O'Donnell.
\newblock Quantum chi-squared tomography and mutual information testing.
\newblock {\em Quantum}, 8:1381, 2024.

\bibitem[Ful97]{Ful97}
William Fulton.
\newblock {\em Young tableaux: with applications to representation theory and geometry}.
\newblock Cambridge University Press, 1997.

\bibitem[GM00]{GM00}
Richard Gill and Serge Massar.
\newblock State estimation for large ensembles.
\newblock {\em Physical Review A}, 61(4):042312, 2000.

\bibitem[GW09]{GW09}
Roe Goodman and Nolan Wallach.
\newblock {\em Symmetry, representations, and invariants}.
\newblock Springer, 2009.

\bibitem[Har05]{Har05}
Aram Harrow.
\newblock {\em Applications of coherent classical communication and the {S}chur transform to quantum information theory}.
\newblock PhD thesis, Massachusetts Institute of Technology, 2005.

\bibitem[Har13]{Har13}
Aram Harrow.
\newblock The church of the symmetric subspace.
\newblock Technical report, arXiv:1308.6595, 2013.

\bibitem[Hay98]{Hay98}
Masahito Hayashi.
\newblock Asymptotic estimation theory for a finite-dimensional pure state model.
\newblock {\em Journal of Physics A: Mathematical and General}, 31(20):4633, 1998.

\bibitem[HHJ{\etalchar{+}}16]{HHJ+16}
Jeongwan Haah, Aram Harrow, Zhengfeng Ji, Xiaodi Wu, and Nengkun Yu.
\newblock Sample-optimal tomography of quantum states.
\newblock In {\em Proceedings of the 48th Annual ACM Symposium on Theory of Computing}, August 2016.
\newblock Preprint.

\bibitem[HKOT23]{HKOT23}
Jeongwan Haah, Robin Kothari, Ryan O'Donnell, and Ewin Tang.
\newblock Query-optimal estimation of unitary channels in diamond distance.
\newblock In {\em Proceedings of the 64th Annual IEEE Symposium on Foundations of Computer Science}, pages 363--390, 2023.

\bibitem[HLSW04]{HLSW04}
Patrick Hayden, Debbie Leung, Peter Shor, and Andreas Winter.
\newblock Randomizing quantum states: Constructions and applications.
\newblock {\em Communications in Mathematical Physics}, 250(2):371--391, 2004.

\bibitem[HM02]{HM02}
Masahito Hayashi and Keiji Matsumoto.
\newblock Quantum universal variable-length source coding.
\newblock {\em Physical Review A}, 66(2):022311, 2002.

\bibitem[Hol79]{Hol79}
Alexander Holevo.
\newblock On asymptotically optimal hypothesis testing in quantum statistics.
\newblock {\em Theory of Probability \& Its Applications}, 23(2):411--415, 1979.

\bibitem[HW94]{HW94}
Paul Hausladen and William Wootters.
\newblock A `pretty good' measurement for distinguishing quantum states.
\newblock {\em Journal of Modern Optics}, 41(12):2385--2390, 1994.

\bibitem[Key06]{Key06}
Michael Keyl.
\newblock Quantum state estimation and large deviations.
\newblock {\em Reviews in Mathematical Physics}, 18(01):19--60, 2006.

\bibitem[KW01]{KW01}
Michael Keyl and Reinhard Werner.
\newblock Estimating the spectrum of a density operator.
\newblock {\em Physical Review A}, 64(5):052311, 2001.

\bibitem[Mel24]{Mel24}
Antonio~Anna Mele.
\newblock Introduction to {H}aar measure tools in quantum information: a beginner's tutorial.
\newblock {\em Quantum}, 8:1340, 2024.

\bibitem[NC10]{NC10}
Michael Nielsen and Isaac Chuang.
\newblock {\em Quantum computation and quantum information}.
\newblock Cambridge University Press, 2010.

\bibitem[OW15]{OW15b}
Ryan O'Donnell and John Wright.
\newblock A note on the {H}aah et al.\ tomography algorithm, 2015.
\newblock \url{https://people.eecs.berkeley.edu/~jswright/papers/tomography-note.pdf}.

\bibitem[OW16]{OW16}
Ryan O'Donnell and John Wright.
\newblock Efficient quantum tomography.
\newblock In {\em Proceedings of the 48th Annual ACM Symposium on Theory of Computing}, 2016.

\bibitem[OW17]{OW17a}
Ryan O'Donnell and John Wright.
\newblock Efficient quantum tomography {II}.
\newblock In {\em Proceedings of the 49th Annual ACM Symposium on Theory of Computing}, 2017.

\bibitem[PSW25]{PSW25}
Angelos Pelecanos, Jack Spilecki, and John Wright.
\newblock The debiased {K}eyl's algorithm: a new unbiased estimator for full state tomography.
\newblock Manuscript, 2025.

\bibitem[PTTW25]{PTTW25}
Angelos Pelecanos, Xinyu Tan, Ewin Tang, and John Wright.
\newblock Beating full state tomography for unentangled spectrum estimation.
\newblock {\em Technical report, arXiv:2504.02785}, 2025.

\bibitem[Reg06]{Reg06}
Oded Regev.
\newblock Lecture {2} from {0368-4057}:\ quantum computing.
\newblock Found at \url{https://cims.nyu.edu/~regev/teaching/quantum_fall_2005/ln/qma.pdf}, 2006.

\bibitem[Sag01]{Sag01}
Bruce~E Sagan.
\newblock {\em The symmetric group: representations, combinatorial algorithms, and symmetric functions}.
\newblock Springer, 2001.

\bibitem[Sta99]{Sta99}
Richard~P Stanley.
\newblock {\em Enumerative combinatorics {V}olume 2}.
\newblock Cambridge University Press, Cambridge, 1999.

\bibitem[Wal17]{Wal17}
Michael Walter.
\newblock Lecture {2} from {PHYSICS 491}: {S}ymmetry and quantum information.
\newblock Found at \url{https://qi.rub.de/courses/physics491/lecture2.pdf}, 2017.

\bibitem[Wat18]{Wat18}
John Watrous.
\newblock {\em The theory of quantum information}.
\newblock Cambridge University Press, 2018.

\bibitem[Wri15]{Wri15}
John Wright.
\newblock Lecture {22} from {15-859BB}:\ {Q}uantum computation and information.
\newblock Found at \url{https://www.cs.cmu.edu/~odonnell/quantum15/lecture22.pdf}, 2015.

\bibitem[Wri16]{Wri16}
John Wright.
\newblock {\em How to learn a quantum state}.
\newblock PhD thesis, Carnegie Mellon University, 2016.

\bibitem[Wri24]{Wri24b}
John Wright.
\newblock Lecture {12} from {CS 294}:\ {Q}uantum learning theory.
\newblock Found at \url{https://people.eecs.berkeley.edu/~jswright/quantumlearningtheory24/scribe\%20notes/lecture12.pdf}, 2024.

\bibitem[Yue23]{Yue23}
Henry Yuen.
\newblock An improved sample complexity lower bound for (fidelity) quantum state tomography.
\newblock {\em Quantum}, 7:890, 2023.

\end{thebibliography}

\end{document}